\documentclass[acmsmall,screen]{acmart}
\renewcommand\footnotetextcopyrightpermission[1]{}
\AtBeginDocument{%
  }

\setcopyright{rightsretained}
\acmDOI{10.1145/3632914}
\acmYear{2024}
\copyrightyear{2024}
\acmSubmissionID{popl24main-p380-p}
\acmJournal{PACMPL}
\acmVolume{8}
\acmNumber{POPL}
\acmArticle{72}
\acmMonth{1}
\received{2023-07-11}
\received[accepted]{2023-11-07}

\usepackage{simconv}
\usepackage{lstcoq}

\begin{document}

\title{Fully Composable and Adequate Verified Compilation with Direct Refinements between Open Modules}

\author{Ling Zhang}
\orcid{0000-0001-7190-6983}             
\affiliation{
  \department{John Hopcroft Center for Computer Science, School of Electronic Information and Electrical Engineering}              
  \institution{Shanghai Jiao Tong University}            
  \country{China}                    
}
\email{ling.zhang@sjtu.edu.cn}          

\author{Yuting Wang}
\authornote{Corresponding author}          
\orcid{0000-0003-3990-2418}             
\affiliation{
  \department{John Hopcroft Center for Computer Science, School of Electronic Information and Electrical Engineering}              
  \institution{Shanghai Jiao Tong University}            
  \country{China}                    
}
\email{yuting.wang@sjtu.edu.cn}          

\author{Jinhua Wu}
\orcid{0000-0001-5812-053X}             
\affiliation{
  \department{John Hopcroft Center for Computer Science, School of Electronic Information and Electrical Engineering}              
  \institution{Shanghai Jiao Tong University}            
  \country{China}                    
}
\email{jinhua.wu@sjtu.edu.cn}          

\author{J{\'e}r{\'e}mie Koenig}
\orcid{0000-0002-3168-5925}             
\affiliation{
  \institution{Yale University}            
  \country{USA}                    
}
\email{jeremie.koenig@yale.edu}          

\author{Zhong Shao}
\orcid{0000-0001-8184-7649}
\affiliation{
  \institution{Yale University}            
  \country{USA}                    
}
\email{zhong.shao@yale.edu}          


\begin{abstract}
  Verified compilation of open modules (i.e., modules whose
  functionality depends on other modules) provides a foundation for
  end-to-end verification of modular programs ubiquitous in contemporary
  software.
  However, despite intensive investigation in this topic for decades,
  the proposed approaches are still difficult to use in practice as
  they rely on assumptions about the internal working of compilers
  which make it difficult for external users to apply the
  verification results.
  We propose an approach to verified compositional
  compilation without such assumptions in the setting of
  verifying compilation of heterogeneous modules written in
  first-order languages supporting global memory and pointers.
  Our approach is based on the memory model of CompCert and a new
  discovery that a Kripke relation with a notion of memory protection
  can serve as a uniform and composable semantic interface for the
  compiler passes.
  By absorbing the rely-guarantee conditions on memory evolution for
  all compiler passes into this Kripke Memory Relation and by
  piggybacking requirements on compiler optimizations onto it, 
  we get compositional correctness theorems for realistic optimizing compilers
  as refinements that directly relate native semantics of open
  modules and that are ignorant of intermediate compilation processes.
  Such direct refinements support all the compositionality and
  adequacy properties essential for verified compilation of open
  modules.
  We have applied this approach to the full compilation chain of CompCert 
  with its Clight source language and
  demonstrated that our compiler correctness theorem is open to
  composition and intuitive to use with reduced verification
  complexity through end-to-end verification of non-trivial
  heterogeneous modules that may freely invoke each other (e.g.,
  mutually recursively).

\end{abstract}



\maketitle

\section{Introduction}
\label{sec:intro}

Verified compilation ensures that behaviors of source programs are
faithfully transported to target code, a property
desirable for end-to-end verification of software whose development involves compilation. As
software is usually composed of modules
independently developed and compiled, researchers have
developed a wide range of techniques for \emph{verified compositional
  compilation} or
VCC
that support modules invoking each other (i.e., open), being
written in different languages (i.e., heterogeneous) and transformed
by different compilers~\cite{patterson-icfp-2019}.

We are concerned with VCC for first-order languages
with global memory states and support of pointers (e.g.,
see~\citet{stewart15,compcertm,compcerto,cascompcert,wang2019,dscal15}).
As it stands now, the proposed approaches are inherently
limited at supporting open modules (e.g. libraries) as they either
deviate from the native semantics of modules or expose 
the semantics of intermediate representations for compilation, resulting in correctness theorems 
that are difficult to work with for external users.
In this paper, we investigate an approach that eliminates these limitations
while retaining the full benefits of VCC, i.e., obtaining
correctness of compiling open modules that is \emph{fully composable},
\emph{adequate}, and \emph{extensional}.
%

\subsection{Full Compositionality and Adequacy in Verified Compilation}
\label{ssec:vcc-props}
Correctness of compiling open modules is usually described as
refinement between semantics of source and target modules.
We shall write $L$ (possibly with subscripts) to denote semantics of
open modules and write $\refined{L_1}{L_2}$ to denote that $L_1$ is
refined by $L_2$. Therefore, the compilation of any module $M_2$ into
$M_1$ is correct iff $\refined{\sem{M_1}}{\sem{M_2}}$ where $\sem{M_i}$
denotes the semantics of $M_i$.

To support the most general form of VCC, it is critical that the
established refinements are \emph{fully composable}, i.e., 
both \emph{horizontally and vertically composable}, and \emph{adequate
  for native semantics}:
\begin{align*}
  \mbox{\bf Vertical Compositionality:}
  & \quad 
  \refined{L_1}{L_2} \imply 
  \refined{L_2}{L_3} \imply 
  \refined{L_1}{L_3}\\
  \mbox{\bf Horizontal Compositionality:} 
  & \quad 
  \refined{L_1}{L_1'} \imply 
  \refined{L_2}{L_2'} \imply 
  \refined{L_1 \semlink L_2}{L_1' \semlink L_2'}\\
  \mbox{\bf Adequacy for Native Semantics:}
  & \quad 
  \refined{\sem{M_1 \synlink M_2}}{\sem{M_1} \semlink \sem{M_2}}
\end{align*}
\begin{wrapfigure}{r}{.25\textwidth}
  \centering
  \begin{tikzpicture}
    \node [] (as) at (-5,0) {$\sem{\code{a.c}}$};
    \node [below = 0.3 of as] (aone) {$\sem{\code{a.i}_1}$};
    \node [below = 0.2 of aone] (atwo) {$\sem{\code{a.i}_2}$};
    \node [below = 0.2 of atwo] (at) {$\sem{\code{a.s}}$};

    \path (as) -- node[sloped,rotate=180] {$\refinesymb$} (aone);
    \path (aone) -- node[sloped,rotate=180] {$\refinesymb$} (atwo);
    \path (atwo) -- node[sloped,rotate=180] {$\refinesymb$} (at);

    \node [right = 1.0 cm of as] (bspec) {$L_b$};
    \path let \p1 = (bspec) in let \p2 = (at) in
        node (bt) at (\x1, \y2) {$\sem{\code{b.s}}$};

    \path (bspec) -- node[sloped, rotate=180] {$\refinesymb$} (bt);

    \path (as) -- node[sloped] {$\semlink$} (bspec);
    \path (at) -- node[sloped] {$\semlink$} (bt);

    \node [draw, rectangle, dashed, fit = (as) (bspec), inner sep = -0.0pt] (srcsem) {};
    \node [fit = (at) (bt), inner sep = -0.0pt] (tgtsem) {};
    \node [draw, rectangle, dashed, below = 0.25 cm of tgtsem] (tgtsyn) 
          {$\sem{\code{a.s} + \code{b.s}}$};

    \path (tgtsem) -- node[sloped,rotate=180] {$\refinesymb$} (tgtsyn);

  \end{tikzpicture}
  \caption{Motivating Example}
  \label{fig:mtv-exm}
\end{wrapfigure}
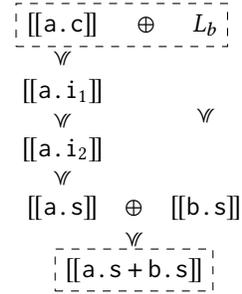
The first property states that refinements are transitive. It is
essential for composing proofs for multi-pass compilers.
The second property guarantees that refinements are preserved by
semantic linking (denoted by $\semlink$). It is essential for
composing correctness of compiling open modules (possibly through
different compilers).
The last one ensures that, given any modules, their semantic linking
coincides with their syntactic linking (denoted by $+$). It ensures
that linked semantics do not deviate from native semantics and is
essential to propagate verified properties to final target programs.

We use the example in~\figref{fig:mtv-exm} to illustrate the
importance of the above properties in VCC where heterogeneous modules
are compiled through different compilation chains and linked into a
final target module.
In this example, a source C module \code{a.c} is compiled into an assembly module
\code{a.s} through a multi-pass optimizing compiler like CompCert: it
is first compiled to $\code{a.i}_1$ in an intermediate representation
(IR) for optimization (e.g., the RTL language of CompCert) and then to
$\code{a.i}_2$ in another IR for code generation (e.g., the Mach
language of CompCert). Finally, it is linked with a library module
\code{b.s} which is not compiled at all (an extreme case where the
compilation chain is empty). The goal is to prove that the semantics
of linked target assembly $\code{a.s} + \code{b.s}$ refines the
combined source semantics $\sem{\code{a.c}} \semlink L_b$ where $L_b$
is the semantic specification of $\code{b.s}$, i.e.,
$\refined{\sem{\code{a.s} + \code{b.s}}}{\sem{\code{a.c}} \semlink L_b}$.
The proof proceeds as follows:
\begin{enumerate}
\item Prove every pass respects refinement, from which 
  $\refined{\sem{\code{a.i}_1}}{\sem{\code{a.c}}}$,
  $\refined{\sem{\code{a.i}_2}}{\sem{\code{a.i}_1}}$ and 
  $\refined{\sem{\code{a.s}}}{\sem{\code{a.i}_2}}$.
  Moreover, show $\code{b.s}$ meets its specification, i.e., 
  $\refined{\sem{\code{b.s}}}{L_b}$;

\item By vertically composing the refinement relations
  for compiling \code{a.c}, we get
  $\refined{\sem{\code{a.s}}}{\sem{\code{a.c}}}$;

\item By further horizontally composing with $\refined{\sem{\code{b.s}}}{L_b}$,
  we get
  $\refined{\sem{\code{a.s}} \semlink \sem{\code{b.s}}}
           {\sem{\code{a.c}} \semlink L_b}$;

\item By adequacy for assembly and vertical
  composition, conclude $\refined{\sem{\code{a.s} + \code{b.s}}}{\sem{\code{a.c}} \semlink L_b}$.
\end{enumerate}
%

\subsection{Problems with the Existing Approaches to Refinements}
\label{ssec:problems}
Despite the simplicity of VCC at an intuitive level, full
compositionality and adequacy are surprisingly difficult to prove for
any non-trivial multi-pass compiler.
First and foremost, the formal definitions must take into account the
facts that each intermediate representation has different semantics
and each pass may imply a different refinement relation.
To facilitate the discussion below, we classify different open
semantics by \emph{language interfaces} (or simply interfaces) which
formalize their interaction with environments. We write $L:\langi{I}$
to denote that $L$ has a language interface $\langi{I}$. For instance,
$\sem{\code{a.c}}:\langi{C}$ denotes that the semantics of \code{a.c} has the
interface $\langi{C}$ which only allows for interaction with environments through
function calls and returns in C. Similarly,
$\sem{\code{a.s}}:\langi{A}$ denotes the semantics of $\code{a.s}$ where
$\langi{A}$ only allows for interaction at the assembly level.
Note that the interface for a module may not match its native
semantics.  For example, $\sem{\code{a.s}}:\langi{C}$ asserts that
$\sem{\code{a.s}}$ actually converts assembly level calls/returns to C
function calls/returns for interacting with C environments (e.g.,
extracting arguments from registers and memory to form an argument
list for C function calls). In this case, $\sem{\code{a.s}}$ \emph{deviates}
from the native semantics of \code{a.s}.
When the interface of $\sem{M}$ is not explicitly given, it is
implicitly the native interface of $M$.
We write $\refinesymb: \sctype{\langi{I}_1}{\langi{I}_2}$ to denote a
refinement between two semantics with interfaces ${\langi{I}_1}$ and
${\langi{I}_2}$.
For instance, given $\refinesymb_{\texttt{ac}}:
\sctype{\langi{A}}{\langi{C}}$ that relates open semantics at the C
and assembly levels, $\refinedi{\sem{\code{b.s}}}{L_b}{{\texttt{ac}}}$
asserts that ${\sem{\code{b.s}}}$ is the native semantics of
\code{b.s} and is refined by the C level specification $L_b$.


For VCC, it is essential that variance of open semantics and
refinements does not impede compositionality and adequacy. The
existing approaches achieve this by imposing \emph{algebraic structures} on
refinements. We categorize them by their algebraic
structures below, and explain the problems facing them via three
well-known extensions of CompCert~\cite{compcert} (the
state-of-the-art verified C compiler) to support VCC, i.e.,
Compositional CompCert (CompComp)~\cite{stewart15},
CompCertM~\cite{compcertm} and CompCertO~\cite{compcerto}.

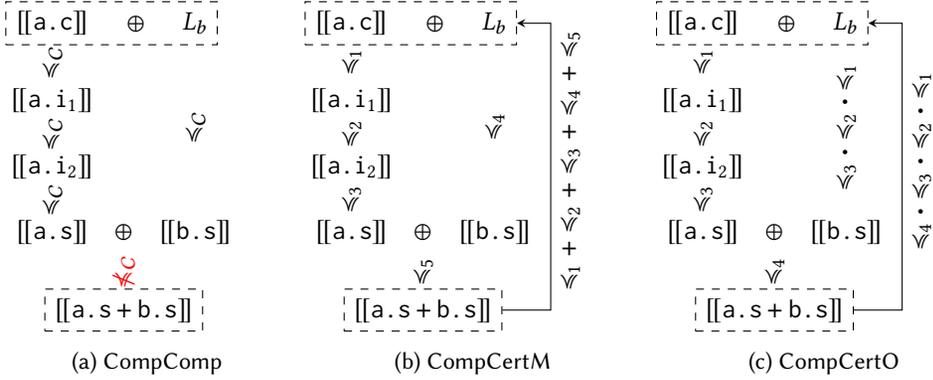
\begin{figure}
  \begin{subfigure}[b]{0.28\textwidth}
  \begin{tikzpicture}
    \node [] (as) at (-5,0) {$\sem{\code{a.c}}$};
    \node [below = 0.45 of as] (aone) {$\sem{\code{a.i}_1}$};
    \node [below = 0.3 of aone] (atwo) {$\sem{\code{a.i}_2}$};
    \node [below = 0.3 of atwo] (at) {$\sem{\code{a.s}}$};

    \path (as) -- node[sloped,rotate=180] {$\refinesymb_{\langi{C}}$} (aone);
    \path (aone) -- node[sloped,rotate=180] {$\refinesymb_{\langi{C}}$} (atwo);
    \path (atwo) -- node[sloped,rotate=180] {$\refinesymb_{\langi{C}}$} (at);

    \node [right = 1.0 cm of as] (bspec) {$L_b$};
    \path let \p1 = (bspec) in let \p2 = (at) in
        node (bt) at (\x1, \y2) {$\sem{\code{b.s}}$};

    \path (bspec) -- node[sloped, rotate=180] {$\refinesymb_{\langi{C}}$} (bt);

    \path (as) -- node[sloped] {$\semlink$} (bspec);
    \path (at) -- node[sloped] {$\semlink$} (bt);

    \node [draw, rectangle, dashed, fit = (as) (bspec), inner sep = -0.0pt] (srcsem) {};
    \node [fit = (at) (bt), inner sep = -0.0pt] (tgtsem) {};
    \node [draw, rectangle, dashed, below = 0.45 cm of tgtsem] (tgtsyn) 
          {$\sem{\code{a.s} + \code{b.s}}$};

    \path (tgtsem) -- node[sloped,rotate=180] 
          {\color{red}$\not\refinesymb_{\langi{C}}$} (tgtsyn);

  \end{tikzpicture}
  \caption{CompComp}
  \label{fig:const-refinement}
  \end{subfigure}
  \begin{subfigure}[b]{0.33\textwidth}
  \begin{tikzpicture}
    \node [] (as) at (-5,0) {$\sem{\code{a.c}}$};
    \node [below = 0.45 of as] (aone) {$\sem{\code{a.i}_1}$};
    \node [below = 0.3 of aone] (atwo) {$\sem{\code{a.i}_2}$};
    \node [below = 0.3 of atwo] (at) {$\sem{\code{a.s}}$};

    \path (as) -- node[sloped,rotate=180] {$\refinesymb_1$} (aone);
    \path (aone) -- node[sloped,rotate=180] {$\refinesymb_2$} (atwo);
    \path (atwo) -- node[sloped,rotate=180] {$\refinesymb_3$} (at);

    \node [right = 1.0 cm of as] (bspec) {$L_b$};
    \path let \p1 = (bspec) in let \p2 = (at) in
        node (bt) at (\x1, \y2) {$\sem{\code{b.s}}$};

    \path (bspec) -- node[sloped, rotate=180] {$\refinesymb_4$} (bt);

    \path (as) -- node[sloped] {$\semlink$} (bspec);
    \path (at) -- node[sloped] {$\semlink$} (bt);

    \node [draw, rectangle, dashed, fit = (as) (bspec), inner sep = -0.0pt] (srcsem) {};
    \node [fit = (at) (bt), inner sep = -0.0pt] (tgtsem) {};
    \node [draw, rectangle, dashed, below = 0.45 cm of tgtsem] (tgtsyn) 
          {$\sem{\code{a.s} + \code{b.s}}$};

    \path (tgtsem) -- node[sloped,rotate=180] 
          {$\refinesymb_5$} (tgtsyn);

    \draw [-stealth] (tgtsyn) -- +(1.7,0) |- 
          node[pos = .26, sloped, below] 
            {$\refinesymb_1 + \refinesymb_2 + \refinesymb_3 + \refinesymb_4 + \refinesymb_5$} (srcsem);

  \end{tikzpicture}
  \caption{CompCertM}
  \label{fig:sum-refinement}
  \end{subfigure}
  \begin{subfigure}[b]{0.33\textwidth}
  \begin{tikzpicture}
    \node [] (as) at (-5,0) {$\sem{\code{a.c}}$};
    \node [below = 0.45 of as] (aone) {$\sem{\code{a.i}_1}$};
    \node [below = 0.3 of aone] (atwo) {$\sem{\code{a.i}_2}$};
    \node [below = 0.3 of atwo] (at) {$\sem{\code{a.s}}$};

    \path (as) -- node[sloped,rotate=180] {$\refinesymb_1$} (aone);
    \path (aone) -- node[sloped,rotate=180] {$\refinesymb_2$} (atwo);
    \path (atwo) -- node[sloped,rotate=180] {$\refinesymb_3$} (at);

    \node [right = 1.0 cm of as] (bspec) {$L_b$};
    \path let \p1 = (bspec) in let \p2 = (at) in
        node (bt) at (\x1, \y2) {$\sem{\code{b.s}}$};

    \path (bspec) -- node[sloped, rotate=180] 
          {$\refinesymb_3 \compsymb \refinesymb_2 \compsymb \refinesymb_1$} (bt);

    \path (as) -- node[sloped] {$\semlink$} (bspec);
    \path (at) -- node[sloped] {$\semlink$} (bt);

    \node [draw, rectangle, dashed, fit = (as) (bspec), inner sep = -0.0pt] (srcsem) {};
    \node [fit = (at) (bt), inner sep = -0.0pt] (tgtsem) {};
    \node [draw, rectangle, dashed, below = 0.45 cm of tgtsem] (tgtsyn) 
          {$\sem{\code{a.s} + \code{b.s}}$};

    \path (tgtsem) -- node[sloped,rotate=180] 
          {$\refinesymb_4$} (tgtsyn);

    \draw [-stealth] (tgtsyn) -- +(1.7,0) |- 
          node[pos = .26, sloped, below] 
            {$\refinesymb_4 \compsymb \refinesymb_3 \compsymb \refinesymb_2 \compsymb \refinesymb_1$} (srcsem);

  \end{tikzpicture}
  \caption{CompCertO}
  \label{fig:prod-refinement}
  \end{subfigure}
  \caption{Refinements in the Existing Approaches to VCC}
  \label{fig:refinements}
\end{figure}
\paragraph{Constant Refinement}
An obvious way to account for different semantics in VCC is to force
every semantics to use the same language interface $\langi{I}$ and 
a constant refinement $\refinesymb_{\langi{I}} :
\sctype{\langi{I}}{\langi{I}}$. CompComp adopts this
``one-type-fits-all'' approach by having every language of CompCert to
use C function calls/returns for module-level interactions and using a
uniform refinement relation $\refinesymb_{\langi{C}}:
\sctype{\langi{C}}{\langi{C}}$ known as \emph{structured
simulation}~\cite{stewart15}. In this case, vertical and horizontal
compositionality is established by proving transitivity of
$\refinesymb_{\langi{C}}$ and symmetry of \emph{rely-guarantee conditions} of
$\refinesymb_{\langi{C}}$. However, because the C interface is adopted
for assembly semantics, adequacy at the target level is lost,
making end-to-end compiler correctness not provable as shown
in~\figref{fig:const-refinement}.
%

\paragraph{Sum of Refinements}
A more relaxed approach allows users to choose language interfaces for
different IRs from a finite collection $\{\langi{I}_1, \ldots,
\langi{I}_m\}$ and refinements for different
passes from a finite set $\{\refinesymb_1, \ldots, \refinesymb_n\}$
relating these interfaces, i.e., $\refinesymb_i :
\sctype{\langi{I}_1 + \ldots + \langi{I}_m}{\langi{I}_1 + \ldots +
  \langi{I}_m}$.
In essence, a constant refinement is split into a sum of refinements
s.t. ${L}\;{\refinesymb_1 + \ldots + \refinesymb_n}\;{L'}$ holds if
$\refinedi{L}{L'}{_i}$ for some $1 \leq i \leq n$. Then, every compiler
pass can use ${\refinesymb_1 + \ldots + \refinesymb_n}$ as the
uniform refinement relation, which is proven both composable and
adequate under certain well-formedness constraints.
%
%
\figref{fig:sum-refinement} depicts such an example where semantics
have both C and assembly interfaces (e.g., $\sem{\code{a.s}} :
\langi{A} + \langi{C}$) and the refinement relations $\refinesymb_i :
\sctype{\langi{A} + \langi{C}}{\langi{A} + \langi{C}} (1 \leq i \leq
5)$ are tailored for each pass. 
This is the approach adopted by CompCertM~\cite{compcertm}. 
However, the top-level refinement ${\refinesymb_1 + \ldots +
  \refinesymb_n}$ is difficult to use by a third party  
without introducing complicated dependency on intermediate
results of compilation.
For example, horizontal composition with ${\refinesymb_1 + \ldots +
  \refinesymb_n}$ only works for modules \emph{self-related} by all
the refinements $\refinesymb_i (1 \leq i \leq n)$. Since
${\refinesymb_i}$s are tailored for individual 
passes, they inevitably depend on the intermediate semantics used in
compilation. Such dependency is only exacerbated as new languages,
compilers and optimizations are introduced.

%
%

\paragraph{Product of Refinements}
The previous approach effectively ``flattens'' the refinements for
individual compiler passes into an end-to-end refinement. A different
approach adopted by CompCertO~\cite{compcerto} is to ``concatenate''
the refinements for individual passes into a chain of refinements by a
product operation $(\_ \compsymb \_)$ such that ${L}\;
{\comp{\refinesymb_1}{\refinesymb_2}}\; {L''}$ if
$\refinedi{L}{L'}{1}$ and $\refinedi{L'}{L''}{2}$ for some
$L'$. \figref{fig:prod-refinement} illustrates how it works. Vertical
composition is simply the concatenation of refinements. For example,
composing refinements for compiling \code{a.c} results
in ${\sem{\code{a.s}}}\; {\refinesymb_3 \compsymb
  \refinesymb_2 \compsymb \refinesymb_1}\; {\sem{\code{a.c}}}$. 
Adequacy is trivially guaranteed with native interfaces.
However, horizontal composition still depends on the intermediate
semantics of compilation because of the concatenation.
For example, in~\figref{fig:prod-refinement}, to horizontally compose
with ${\sem{\code{a.s}}}\; {\refinesymb_3 \compsymb \refinesymb_2
  \compsymb \refinesymb_1}\; {\sem{\code{a.c}}}$, it is necessary to
show $L_b$ refines $\sem{\code{b.s}}$ via the same product, i.e., to
construct intermediate semantics bridging $\refinesymb_1$,
$\refinesymb_2$ and $\refinesymb_3$.

\paragraph{Summary} The existing approaches for VCC either lack adequacy
because they force non-native language interfaces on semantics for
open modules (e.g., CompComp) or lack compositionality that is truly
extensional because they depend on intermediate semantics used
in compilation (e.g., CompCertM and CompCertO). Such dependency makes
their correctness theorems for compiling open modules (e.g.,
libraries) difficult to further compose with and incurs a
high cost in verification.


\subsection{Challenges for Direct Refinement of Open Modules}
The ideal approach to VCC should produce refinements that directly
relate the native semantics of source and target open modules without
mentioning any intermediate semantics and support both vertical and
horizontal composition. We shall call them \emph{direct refinements of
open modules}.
For example, a direct refinement between $\code{a.c}$ and $\code{a.s}$
could be $\refinesymb_{\texttt{ac}} : \sctype{\langi{A}}{\langi{C}}$
s.t. $\refinedi{\sem{\code{a.s}}}{\sem{\code{a.c}}}{{\texttt{ac}}}$. It
relates assembly and C without mentioning intermediate
semantics, and could be further horizontally composed with
$\refinedi{\sem{\code{b.s}}}{L_b}{{\texttt{ac}}}$ and vertically
composed by adequacy to get $\refinedi{\sem{\code{a.s} +
    \code{b.s}}}{\sem{\code{a.c}} \semlink L_b}{\texttt{ac}}$. Note
that even the top-level refinement is still open to horizontal
and vertical composition, making direct refinements effective for
supporting VCC for open modules.

The main challenge in getting direct refinements is tied to their
``real'' vertical composition, i.e., given any direct refinements
${\refinesymb_1}$ and ${\refinesymb_2}$, how to show
$\comp{\refinesymb_1}{\refinesymb_2}$ is equivalent to a direct
refinement $\refinesymb_3$. This is considered very technical and
involved
(see~\citet{compcertm,neis15icfp,patterson-icfp-2019,hur2012tr})
because of the difficulty in constructing \emph{interpolating} program
states for transitively relating evolving source and target states
across \emph{external calls} of open modules. This problem also
manifests in proving transitivity for \emph{logical relations} where
construction of interpolating terms of higher-order types is not in
general possible~\cite{Ahmed06esop}. In the setting of compiling
first-order languages with global memory, all previous work avoids
proving real vertical composition of direct refinements. Some produce
refinement without adequacy by introducing intrusive changes to
semantics to make construction of interpolating states
possible. For example, CompComp instruments the semantics of languages
with \emph{effect annotations} to expose internal effects for this
purpose. Some essentially restrict vertical composition to
\emph{closed} programs (e.g., CompCertM). Some leave the
top-level refinement a combination of refinements that
still exposes the intermediate steps of compilation (e.g., CompCertO).
Finally, even if the problem of vertical composition was solved,
it is not clear if the solution can support realistic optimizing
compilers.


\subsection{Our Contributions}
\label{ssec:contrib}
In this paper, we propose an approach to direct refinements
for VCC of imperative programs that addresses all of the above challenges.
Our approach is based on the memory model of CompCert which
supports first-order states and pointers. 
%
We show that in this memory model
interpolating states for proving vertical compositionality of
refinements can be constructed by exploiting the properties on memory
invariants known as \emph{memory injections}. The solution is based on a new discovery
that a \emph{Kripke relation with memory protection} can
serve as a uniform and composable relation for characterizing the
evolution of memory states across external calls. With this
relation we successfully
combined the correctness theorems of CompCert's passes
into a direct refinement between C
and assembly modules.
We summarize our technical contributions below:
\begin{itemize}
\item We prove that \kinjp---a Kripke Memory Relation with a notion
  of memory protection---is
  both uniform (i.e., memory transformation in every compiler pass
  respects this relation) and composable (i.e., transitive modulo an
  equivalence relation). The critical observation making this proof
  possible is that interpolating memory states can be constructed by
  exploiting memory protection \emph{inherent} to memory injections
  and the \emph{functional} nature of injections. 

\item Based on the above observation, we show that a direct refinement
  from C to assembly can be derived by composing open refinements for
  all of CompCert's passes starting from Clight. In particular, we
  show that compiler passes can use different Kripke relations
  sufficient for their proofs (which may be weaker than \kinjp) and
  these relations will later be absorbed into \kinjp via refinements
  of open semantics. Furthermore, we show that assumptions for
  compiler optimizations can be formalized as \emph{semantic invariants} and, when
  piggybacked onto \kinjp, can be transitively composed. Based on these
  techniques, we upgrade the proofs in CompCertO to get a direct
  refinement from C to assembly for the full CompCert, including all
  of its optimization passes. These experiments show that direct
  refinements can be obtained without fundamental changes to the
  verification framework of CompCert.

\item We demonstrate the simplicity and usefulness of direct
  refinements by applying it to end-to-end verification of several
  non-trivial examples with heterogeneous modules that \emph{mutually}
  invoke each other. In particular, we observe that C level
  refinements can be absorbed into the direct refinement of CompCert
  by transitivity of \kinjp. Combining direct refinements with full
  compositionality and adequacy, we derive end-to-end refinements from
  high-level source specifications to syntactically linked assembly
  modules in a straightforward manner.

\end{itemize}

The above developments are fully formalized in Coq based on the latest
CompCertO which is in turn built on top of CompCert v3.10 (see the
data-availability statement at the end of the paper for more details).
While the formalisation of our approach is tied to CompCert’s
block-based memory model~\cite{compcert-mem-v2}, and applied to its
particular chain of compilation, we present evidence
in~\secref{sec:generality} that variants of \kinjp could be adapted
for alternate memory models for first-order languages, and that it may
be extended to support new optimizations. Therefore, this work
provides a promising direction for further evolving the techniques for
VCC.


\subsection{Structure of the Paper}
Below we first introduce the key ideas supporting
this work in~\secref{sec:ideas}. We then introduce necessary
background and discuss the technical challenges for building
direct refinements in~\secref{sec:background-challenges}. We
present our technical contributions
in~\secref{sec:injp},~\secref{sec:refinement}
and~\secref{sec:application}. We 
discuss the generality and limitations of our approach 
in~\secref{sec:generality}. We discuss evaluation and related
work in~\secref{sec:related} and finally conclude
in~\secref{sec:conc}.

\section{Key Ideas}
\label{sec:ideas}


\begin{figure}
\begin{subfigure}[b]{0.3\textwidth}
\begin{lstlisting}[language = C]
/* client.c */
int result;

void encrypt(int i,
     void(*p)(int*));

void process(int *r) 
{
  result = *r;
}

int request(int i) 
{
  encrypt(i,process);
  return i;
}
\end{lstlisting}
\caption{Client in C}
\label{fig:client}
\end{subfigure}
%
\begin{subfigure}[b]{0.33\textwidth}
\begin{lstlisting}[language = C]
/* server.s */  
key:
  .long 42
encrypt:
  // allocate frame
  Pallocframe 24 16 0
  // RSP[8] = i XOR key
  Pmov key RAX 
  Pxor RAX RDI
  Pmov RDI 8(RSP)
  // call p(RSP + 8)
  Plea 8(RSP) RDI
  Pcall RSI   
  // free frame 
  Pfreeframe 24 16 0 
  Pret
\end{lstlisting}  
\caption{Server in Asm}
\label{fig:server}
\end{subfigure}
\begin{subfigure}[b]{0.33\textwidth}
  \begin{lstlisting}[language = C]
/* server_opt.s 
 * key is an constant
 * and inlined in code */
encrypt:
  // allocate frame
  Pallocframe 24 16 0 
  // RSP[8] = i XOR 42
  Pxori 42 RDI

  Pmov RDI 8(RSP)
  // call p(RSP + 8)
  Plea 8(RSP) RDI
  Pcall RSI       
  // free frame 
  Pfreeframe 24 16 0 
  Pret
\end{lstlisting}
\caption{Optimized Server}
\label{fig:server_opt}
\end{subfigure}

\caption{An Example of Encryption Client and Server}
\label{fig:clientserver}
\end{figure}

We introduce a running example with heterogeneous modules and callback
functions to illustrate the key ideas of our work.
This example is representative of mutual dependency between modules
that often appears in practice and it shows how free-form invocation
between modules can be supported by our approach. As we shall see
in~\secref{sec:application}, our approach also handles more
complicated programs with mutually \emph{recursive} heterogeneity
without any problem.

The example is given in~\figref{fig:clientserver}. It consists of a
client written in C (\figref{fig:client}) and an encryption server
hand-written in x86 assembly by using CompCert's assembly syntax where
instruction names begin with \code{P} (\figref{fig:server}). For now,
let us ignore \figref{fig:server_opt} which illustrates how
optimizations work in direct refinements. Users invoke \code{request}
to initialize an encryption request. It is relayed to the function
\code{encrypt} in the server with the prototype
\code{void encrypt(int i, void (*p)(int*))}
which respects a calling convention placing the first and second
arguments in registers \code{RDI} and \code{RSI}, respectively. The main job of
the server is to encrypt \code{i} (\code{RDI}) by XORing it with an
encryption key (stored in the global variable \code{key}) and invoke
the callback function \code{p} (\code{RSI}). Finally, the client takes
over and stores the encrypted value in the global variable
\code{result}.
The pseudo instruction \code{Pallocframe m n o} allocates a stack
frame of \code{m} bytes and stores its address in register \code{RSP}. In
this frame, a pointer to the caller's stack frame is stored at the
\code{o}-th byte and the return address is stored at the \code{n}-th
byte. Note that \code{Pallocframe 24 16 0} in \code{encrypt} reserves
$8$ bytes on the stack from \code{RSP + 8} to \code{RSP + 16} for
storing the encrypted value whose address is passed to
the callback function \code{p}.
\code{Pfreeframe m n o} frees the frame and
restores \code{RSP} and the return address \code{RA}.

\begin{wrapfigure}{R}{.4\textwidth}
  \centering
  \begin{tikzpicture}
    \node [] (as) at (-5,0) {$\sem{\code{client.c}}$};
    \node [below = 1.2 of as] (at) {$\sem{\code{client.s}}$};

    \draw [-stealth, dashed, line width=1] (as) -- 
      node [below, sloped] {\footnotesize CompCert} 
      node [below, sloped, rotate=180] {$\refinesymb_{\texttt{ac}}$}
      (at);
    
    \node [right = 1.0 cm of as] (bspec) {$L_{\texttt{S}}$};
    \path let \p1 = (bspec) in let \p2 = (at) in
        node (bt) at (\x1, \y2) {$\sem{\code{server.s}}$};

    \path (bspec) -- node[sloped, rotate=180] 
          {$\refinesymb_{\texttt{ac}}$} (bt);

    \path (as) -- node[pos=0.2, sloped] {$\semlink$} (bspec);
    \path (at) -- node[sloped] {$\semlink$} (bt);

    \node [fit = (as) (bspec), inner sep = -0.0pt] (srcsem) {};
    \node [fit = (at) (bt), inner sep = -0.0pt] (tgtsem) {};
    \node [draw, rectangle, dashed, below = 0.55 cm of tgtsem] (tgtsyn) 
          {$\sem{\code{client.s} + \code{server.s}}$};

    \path (tgtsem) -- node[sloped,rotate=180] 
          {$\refinesymb_{\texttt{id}}$} (tgtsyn);

    \node [draw, rectangle, dashed, above = 3.3cm of tgtsyn] (topspec)
          {$L_{\texttt{CS}}$};

    \path (tgtsyn) -- node[pos=0.9, sloped] 
          {$\refinesymb_{\texttt{c}}$} (topspec);

    \draw [-stealth] (tgtsyn) -- +(2.2,0) |- 
          node[pos = .26, sloped, below] 
            {$\refinesymb_{\texttt{ac}}$}
            (topspec);
  \end{tikzpicture}
  \caption{Verifying the Running Example}
  \label{fig:running-exm-refinement}
\end{wrapfigure}
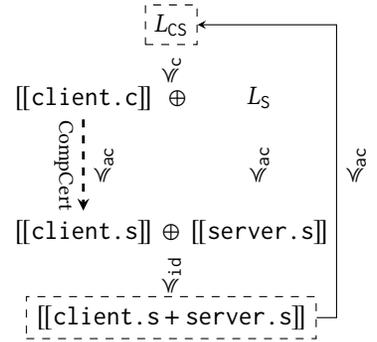

With the running example, our goal is to verify its end-to-end
correctness by exploiting the direct refinement
$\refinesymb_{\texttt{ac}}: \sctype{\langi{A}}{\langi{C}}$ 
derived from CompCert's compilation chain
as shown in~\figref{fig:running-exm-refinement}. The verification
proceeds as follows. First, we establish
$\refinedi{\sem{\code{client.s}}}{\sem{\code{client.c}}}{\texttt{ac}}$ by the
correctness of compilation. Then, we prove
$\refinedi{\sem{\code{server.s}}}{L_{\texttt{S}}}{\texttt{ac}}$ manually by
providing a specification ${L_{\texttt{S}}}$ for the server that
respects the direct refinement. At the source level, the combined
semantics is further refined to a single top-level
specification $L_{\texttt{CS}}$. Finally, the source and
target level refinements are absorbed into the direct refinement by
vertical composition and adequacy, resulting in a \emph{single direct refinement} between the top-level specification and the
target program:
\[
\refinedi{\sem{\code{client.s} +
    \code{server.s}}}{L_{\texttt{CS}}}{\texttt{ac}}
\]
%


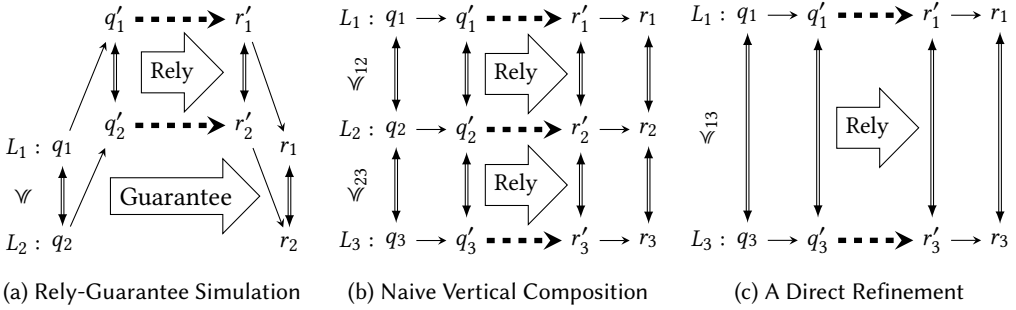
\begin{figure}






  \begin{subfigure}[b]{0.32\textwidth}
  \centering
  \begin{tikzpicture}
    \node (q1) {\small $q_1$};
    \node[left= -0.1cm of q1] (l1) {\small $L_1:$};
    \node[right = 2.5cm of q1] (r1) {\small $r_1$};
    \node[below = 0.8cm of q1] (q2) {\small {$q_2$}};
    \path let \p1 = (r1) in let \p2 = (q2) in
        node (r2) at (\x1, \y2) {\small $r_2$};

    \path let \p1 = (q1) in let \p2 = (r2) in
        node at ({(\x1+\x2)*0.5}, {(\y1+\y2)*0.5}) [single arrow,draw]{Guarantee};

    \draw[double, latex-latex] (q1) -- (q2);
    \draw[double, latex-latex] (r1) -- (r2);

    \node (q2p) at ($(q1)+(0.7cm,0.3cm)$) {\small $q_2'$};
    \node[left= -0.1cm of q2] (l2) {\small $L_2:$};
    \node[right = 1.2cm of q2p] (r2p) {\small $r_2'$};
    \node[above = 0.8cm of q2p] (q1p) {\small $q_1'$};
    \path let \p1 = (r2p) in let \p2 = (q1p) in
        node (r1p) at (\x1, \y2) {\small $r_1'$};

    \draw[double, latex-latex] (q1p) -- (q2p);
    \draw[double, latex-latex] (r1p) -- (r2p);

    \draw[-stealth] (q1) -- (q1p);
    \draw[-stealth] (q2) -- (q2p);
    \draw[-stealth, line width = 2, dashed] (q1p) -- (r1p);
    \draw[-stealth, line width = 2, dashed] (q2p) -- (r2p);
    \draw[-stealth] (r1p) -- (r1);
    \draw[-stealth] (r2p) -- (r2);

    \path let \p1 = (q1p) in let \p2 = (r2p) in
        node at ({(\x1+\x2)*0.5-0.1cm}, {(\y1+\y2)*0.5}) [single arrow,draw]{\small Rely};

    \path (l2) -- node[sloped] {$\refinesymb$} (l1);


  \end{tikzpicture}
  \caption{Rely-Guarantee Simulation}
  \label{sfig:base-sim}
  \end{subfigure}
  \begin{subfigure}[b]{0.33\textwidth}
  \centering
  \begin{tikzpicture}
    \node (q1) {\small $q_1$};
    \node[left= -0.1cm of q1] (l1) {\small {$L_1:$}};
    \node[right = 0.4cm of q1] (q1p) {\small $q_1'$};
    \node[right = 1cm of q1p] (r1p) {\small $r_1'$};
    \node[right = 0.4cm of r1p] (r1) {\small $r_1$};

    \node[below = 1cm of q1] (q2) {\small $q_2$};
    \node[left= -0.1cm of q2] (l2) {\small $L_2:$};
    \path let \p1 = (q2) in let \p2 = (q1p) in
        node (q2p) at (\x2, \y1) {\small $q_2'$};
    \path let \p1 = (q2) in let \p2 = (r1p) in
        node (r2p) at (\x2, \y1) {\small $r_2'$};
    \path let \p1 = (q2) in let \p2 = (r1) in
        node (r2) at (\x2, \y1) {\small $r_2$};

    \node[below = 1cm of q2] (q3) {\small {$q_3$}};
    \node[left= -0.1cm of q3] (l3) {\small $L_3:$};
    \path let \p1 = (q3) in let \p2 = (q1p) in
        node (q3p) at (\x2, \y1) {\small $q_3'$};
    \path let \p1 = (q3) in let \p2 = (r1p) in
        node (r3p) at (\x2, \y1) {\small $r_3'$};
    \path let \p1 = (q3) in let \p2 = (r1) in
        node (r3) at (\x2, \y1) {\small $r_3$};

    \path (l2) -- node[sloped] {$\refinesymb_{12}$} (l1);
    \path (l3) -- node[sloped] {$\refinesymb_{23}$} (l2);

    \draw[double, latex-latex] (q1) -- (q2);
    \draw[double, latex-latex] (q1p) -- (q2p);
    \draw[double, latex-latex] (r1p) -- (r2p);
    \draw[double, latex-latex] (r1) -- (r2);

    \draw[double, latex-latex] (q2) -- (q3);
    \draw[double, latex-latex] (q2p) -- (q3p);
    \draw[double, latex-latex] (r2p) -- (r3p);
    \draw[double, latex-latex] (r2) -- (r3);

    \draw[-stealth] (q1) -- (q1p);
    \draw[-stealth] (q2) -- (q2p);
    \draw[-stealth] (q3) -- (q3p);
    \draw[-stealth, line width = 2, dashed] (q1p) -- (r1p);
    \draw[-stealth, line width = 2, dashed] (q2p) -- (r2p);
    \draw[-stealth, line width = 2, dashed] (q3p) -- (r3p);
    \draw[-stealth] (r1p) -- (r1);
    \draw[-stealth] (r2p) -- (r2);
    \draw[-stealth] (r3p) -- (r3);

    \path let \p1 = (q1p) in let \p2 = (r2p) in
        node at ({(\x1+\x2)*0.5-0.1cm}, {(\y1+\y2)*0.5}) [single arrow,draw]{\small Rely};

    \path let \p1 = (q2p) in let \p2 = (r3p) in
        node at ({(\x1+\x2)*0.5-0.1cm}, {(\y1+\y2)*0.5}) [single arrow,draw]{\small Rely};
  \end{tikzpicture}
  \caption{Naive Vertical Composition}
  \label{sfig:naive-vcomp}
  \end{subfigure}
  \begin{subfigure}[b]{0.33\textwidth}
  \centering
  \begin{tikzpicture}
    \node (q1) {\small $q_1$};
    \node[left= -0.1cm of q1] (l1) {\small $L_1:$};

    \node[right = 0.4cm of q1] (q1p) {\small $q_1'$};
    \node[right = 1cm of q1p] (r1p) {\small $r_1'$};
    \node[right = 0.4cm of r1p] (r1) {\small $r_1$};

    \node[below = 2.5cm of q1] (q3) {\small $q_3$};
    \node[left= -0.1cm of q3] (l3) {\small $L_3:$};
    \path let \p1 = (q3) in let \p2 = (q1p) in
        node (q3p) at (\x2, \y1) {\small $q_3'$};
    \path let \p1 = (q3) in let \p2 = (r1p) in
        node (r3p) at (\x2, \y1) {\small $r_3'$};
    \path let \p1 = (q3) in let \p2 = (r1) in
        node (r3) at (\x2, \y1) {\small $r_3$};

    \path (l3) -- node[sloped] {\small $\refinesymb_{13}$} (l1);

    \draw[double, latex-latex] (q1) -- (q3);
    \draw[double, latex-latex] (q1p) -- (q3p);
    \draw[double, latex-latex] (r1p) -- (r3p);
    \draw[double, latex-latex] (r1) -- (r3);

    \draw[-stealth] (q1) -- (q1p);
    \draw[-stealth] (q3) -- (q3p);
    \draw[-stealth, line width = 2, dashed] (q1p) -- (r1p);
    \draw[-stealth, line width = 2, dashed] (q3p) -- (r3p);
    \draw[-stealth] (r1p) -- (r1);
    \draw[-stealth] (r3p) -- (r3);

    \path let \p1 = (q1p) in let \p2 = (r3p) in
        node at ({(\x1+\x2)*0.5 - 0.1cm}, {(\y1+\y2)*0.5}) [single arrow,draw]{\small Rely};

  \end{tikzpicture}
  \caption{A Direct Refinement}
  \label{sfig:real-vcomp}
  \end{subfigure}

  \caption{Basic Concepts of Open Simulations}
  \label{fig:open-sim-basics}
\end{figure}

The refinements of open modules discussed in our paper are based on
forward simulations between small-step operational semantics (often in
the form of \emph{labeled transition systems} or LTS) which have been
witnessed in a wide range of verification
projects~\cite{stewart15,compcertm,compcerto,cascompcert,wang2019,dscal15}.
\figref{sfig:base-sim} depicts a refinement $\refined{L_2}{L_1}$
between two open semantics (LTS) $L_1$ and $L_2$. The source (target)
semantics $L_1$ ($L_2$) is initialized with a query (i.e., function call)
$q_1$ ($q_2$) and may invoke an external call $q_1'$ ($q_2'$) as the
execution goes. The execution continues when $q_1'$ ($q_2'$) returns
with a reply $r_1'$ ($r_2'$) and finishes with a reply $r_1$ ($r_2$).
For the refinement to hold, an invariant between the source and target
program states must hold throughout the execution which is denoted by
the vertical double arrows in~\figref{sfig:base-sim}.
Furthermore, this refinement relies on external calls satisfying
certain well-behavedness conditions (known as \emph{rely-conditions};
e.g., external calls do not modify the private memory of callers). In
turn, it guarantees the entire source and target execution satisfy
some well-behavedness conditions (known as
\emph{guarantee-conditions}, e.g., they do not modify the private
memory of their calling environments).
The rely-guarantee conditions are essential for horizontal
composition: two refinements $\refined{L_1}{L_2}$ and
$\refined{L_1'}{L_2'}$ with complementary rely-guarantee conditions
can be composed into a single refinement $\refined{L_1 \semlink
  L_2}{L_1' \semlink L_2'}$.
However, vertical composition of such refinements is difficult. A
naive vertical composition of two refinements (one between $L_1$ and
$L_2$ and another between $L_2$ and $L_3$) simply concatenates them
together like~\figref{sfig:naive-vcomp}, instead of generating a
single refinement between $L_1$ and $L_3$
like~\figref{sfig:real-vcomp}. \footnote{To simplify the presentation, we often
  elide the guarantee conditions in figures for simulation.}
This exposes the intermediate semantics
(i.e., $L_2$) and imposes serious limitations on VCC as discussed
in~\secref{sec:intro}. Therefore,
to the best of our knowledge, none of the existing approaches fully
support the verification outlined
in~\figref{fig:running-exm-refinement}.



To address the above problem, we develop direct refinements with the
following distinguishing features: \emph{1)} they always relate
the semantics of modules at their native interfaces, thereby supporting
adequacy; \emph{2)} they do not mention the intermediate process of
compilation, thereby supporting heterogeneous modules and compilers;
\emph{3)} they provide direct memory protection for source and target
semantics via a Kripke relation, thereby enabling horizontal
composition of refinements for heterogeneous modules; \emph{4)} most
importantly, they are vertically composable.
The first three features are manifested in the very definition of
direct refinements, which we shall discuss
in~\secref{ssec:key-idea-direct-refinement} below. We then discuss the
vertical composition of direct refinements
in~\secref{ssec:key-idea-vcomp}, which relies on the discovery of the
uniformity and transitivity of a Kripke relation for memory protection.

\subsection{Refinement Supporting Adequacy, Heterogeneity and Horizontal Composition}
\label{ssec:key-idea-direct-refinement}

%
To illustrate the key ideas, we use the top-level direct refinement
$\refinesymb_{\code{ac}}$ in~\figref{fig:running-exm-refinement} as an
example. In the remaining discussions we adopt the block-based memory
model of CompCert~\cite{compcert-mem-v2} where a memory state consists
of a disjoint set of \emph{memory blocks}.
$\refinesymb_{\code{ac}}$ is a forward simulation that directly
relates C and assembly modules with their native language interfaces.
%
%
By the definition of these interfaces (See~\secref{ssec:background}),
a C query $q_\cli = \cquery{v_f}{\sig}{\vec{v}}{m}$ is a function call
to $v_f$ with signature $\sig$, a list of arguments $\vec{v}$
and a memory state $m$; a C reply $r_\cli = \creply{v'}{m'}$ carries a
return value $v'$ and an updated memory state $m'$. An assembly query
$q_\asmli = \asmquery{\regset}{m}$ invokes a function with the current
register set $\regset$ and memory state $m$. An assembly reply
$r_\asmli = \asmreply{\regset'}{m'}$ returns from a function with the
updated registers $\regset'$ and memory $m'$. By definition, $L_2
\refinesymb_{\code{ac}} L_1$ means that $L_1$ and $L_2$ \emph{behave}
like C and assembly programs at the boundary of modules,
respectively. However, there is no restriction on how $L_1$ and $L_2$
are actually \emph{implemented} internally, which enables
specifications like $L_{\code{S}}$ in~\figref{fig:running-exm-refinement}.

The rely and guarantee conditions imposed by $\refinesymb_{\code{ac}}$
are symmetric and bundled with the simulation invariants at the
boundary of modules. They make assumptions about how C and assembly
queries should be related at the call sites and provide conclusions
about how the replies should be related after the calls return.
Given any matching source and target queries $q_\cli =
\cquery{v_f}{\sig}{\vec{v}}{m_1}$ and $q_\asmli =
\asmquery{\regset}{m_2}$, it is assumed that
\begin{enumerate}
\item The memory states are related by an invariant $j$ known
  as a \emph{memory injection function}~\cite{compcert-mem-v2}, i.e.,
  memory blocks in $m_1$ are projected by $j$ into those in $m_2$;
\item The function pointer $v_f$ is related to the program counter register in $\regset$;
\item The source arguments $\vec{v}$ are projected either to
  registers in $\regset$ or to outgoing argument slots in the stack frame
  $\spreg$ in $m_2$ according to the C calling convention;
\item The outgoing arguments on the target stack frame are \emph{freeable}
  and not in the image of $j$.
\end{enumerate}
The first three requirements ensure that C arguments and memory are
related to assembly registers and memory according to CompCert's C calling
convention. The last one ensures outgoing arguments are
protected, thereby preserving the invariant of open simulation across
external calls.

After the function calls return, the source and target queries $r_\cli
= \creply{\res}{m_1'}$ and $r_\asmli = \asmreply{\regset'}{m_2'}$
must satisfy the following requirements:
\begin{enumerate}
\item The updated memory states $m_1'$ and $m_2'$ are related by an updated memory injection $j'$;
\item The C-level return value $\res$ is related to the value stored
  in the register for return value;
\item For any callee-saved register $\mathit{r}$, $\regset'(\mathit{r})
  = \regset(\mathit{r})$;
\item The stack pointer register and program counter are restored.
\item The access to memory during the function call is protected by a
  \emph{Kripke Memory Relation} $\kinjp$ such that the private stack
  data for other function calls are not modified.
\end{enumerate}
The first two requirements ensure that return values and memories are
related according to the calling convention. The following two ensure
that registers are correctly restored before returning. The last
requirement plays a critical role in rely-guarantee reasoning and
enables horizontal composition of direct refinements as we shall see soon.
%

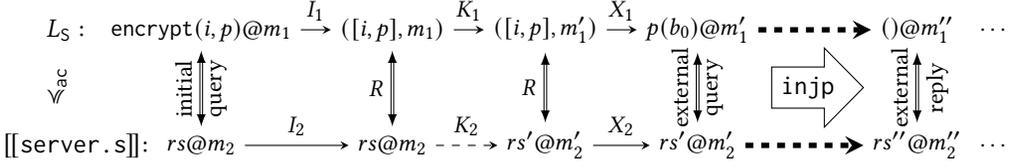
\begin{figure}
\begin{center}
\begin{tikzpicture}[
    world/.style={draw,minimum height = 0.5cm, minimum width = 0.7cm}
  ]
  
  \node (q1) {\small $\scquery{\texttt{encrypt}}{i,p}{m_1}$};
  \node[right = 0.4 of q1] (s1) {\small $([i,p],m_1)$};
  \node[right = 0.4 of s1] (s1p) {\small $([i,p],m_1')$};
  \node[right = 0.4 of s1p] (q1p) {\small $\scquery{p}{b_0}{m_1'}$};
  \node[right = 1.5 of q1p] (r1p) {\small $\creply{()}{m_1''}$};
  \node[right = 0.1 of r1p] (s1pp) {\small $\ldots$};

  \node[below = 1 of q1] (q2) {\small $\asmquery{rs}{m_2}$};

  \path let \p1 = (s1) in let \p2 = (q2) in
     node (s2) at (\x1,\y2) {\small $\asmquery{rs}{m_2}$};
  \path let \p1 = (s1p) in let \p2 = (q2) in
     node (s2p) at (\x1,\y2) {\small $\asmquery{rs'}{m_2'}$};
  \path let \p1 = (q1p) in let \p2 = (q2) in
     node (q2p) at (\x1,\y2) {\small $\asmquery{rs'}{m_2'}$};
  \path let \p1 = (r1p) in let \p2 = (q2) in
     node (r2p) at (\x1,\y2) {\small $\asmreply{rs''}{m_2''}$};
  \path let \p1 = (s1pp) in let \p2 = (q2) in
     node (s2pp) at (\x1,\y2) {\small $\ldots$};

  \draw[-stealth] (q1) -- node[sloped,above] {\small $I_1$} (s1);
  \draw[-stealth] (s1) -- node[sloped,above] {\small $K_1$}(s1p);
  \draw[-stealth] (s1p) -- node[sloped,above] {\small $X_1$} (q1p);
  \draw[-stealth, line width = 2, dashed] (q1p) -- (r1p);

  \draw[-stealth] (q2) -- node[sloped,above] {\small $I_2$} (s2);
  \draw[-stealth, dashed] (s2) -- node[sloped,above] {\small $K_2$} (s2p);
  \draw[-stealth] (s2p) -- node[sloped,above] {\small $X_2$} (q2p);
  \draw[-stealth, line width = 2, dashed] (q2p) -- (r2p);

  \draw [double, latex-latex] (q2) -- 
      node[sloped,above] {\small initial} node[sloped,below] {\small query} (q1);
  \draw [double, latex-latex] (s2) -- 
      node[left] {\small $R$} (s1);
  \draw [double, latex-latex] (s2p) -- 
      node[left] {\small $R$} (s1p);
  \draw [double, latex-latex] (q2p) -- 
      node[sloped,above] (q12p) {\small external} 
      node[sloped,below] (q12p) {\small query} (q1p);
  \draw [double, latex-latex] (r2p) -- 
      node[sloped,above] (r12p) {\small external} 
      node[sloped,below] (q12p) {\small reply} (r1p);

  \path let \p1 = (q1p) in let \p2 = (r2p) in
      node[single arrow,draw] (acca) at ({(\x1+\x2)*.5},{(\y1+\y2)*.5}) {\kinjp};

  \node[left = 0.1 of q1] (lb) {$L_{\texttt{S}}:$};
  \node[left = 0.0 of q2] {$\sem{\code{server.s}}$:};
  \node[below left = 0.1 and 0 of lb, rotate = 90] {$\refinesymb_{\code{ac}}$};

\end{tikzpicture}  
\end{center}
\caption{Direct Refinement of the Hand-written Server}
\label{fig:server-refinement}
\end{figure}

\subsubsection{Adequacy and Heterogeneity via Direct Refinement}
By definition, $\refinesymb_{\code{ac}}$ is basically a formalized C
calling convention for CompCert with direct relations between C and
assembly operational semantics and with invariants for protecting register
values and memory states. Adequacy is
automatically guaranteed as syntactic linking coincides with semantics
linking at the assembly level. That is, given any assembly modules
\code{a.s} and \code{b.s}, 
${\sem{\code{a.s} \synlink \code{b.s}}} \refinesymb_{\code{id}} {\sem{\code{a.s}} \semlink
  \sem{\code{b.s}}}$.

Moreover, $\refinesymb_{\code{ac}}$ does not
mention anything about compilation. It works for any heterogeneous
module and compilation chain that meet its requirements, 
even for hand-written assembly. Take the
refinement of ${\sem{\code{server.s}}} \refinesymb_{\code{ac}} {L_{\texttt{S}}}$
in~\figref{fig:running-exm-refinement} as an example. 
The first few steps of the simulation are depicted
in~\figref{fig:server-refinement}, where $L_{\texttt{S}}$ is an LTS
hand-written by us and ${\sem{\code{server.s}}}$ is derived from
the CompCert assembly semantics. Because 
$L_{\texttt{S}}$ is only required to respect the C interface, we
choose a form easy to comprehend where its internal executions are 
in big steps. Now, suppose the environment calls \code{encrypt} with
source and target queries initially related by CompCert's calling
convention s.t. $rs(\code{RDI}) = i$ and $rs(\code{RSI}) = p$. After
the initialization $I_1$ and $I_2$, the execution enters internal
states related by an invariant $R$. Then, the target execution takes
internal steps $K_2$ until reaching an external call. This corresponds to
executing lines 5-13 in~\figref{fig:server}, which allocates the stack
frame \code{RSP}, performs encryption by storing $\code{i XOR key}$ at
the address \code{RSP+8}, and calls back $p$ with
\code{RSP+8}. At the source level, these steps correspond 
to one big-step execution $K_1$ which allocates a memory block $b_0$, 
stores $\code{i XOR key}$ at $b_0$, and prepares to call
 $p$ with $b_0$. Therefore, the
memory injection in $R$ maps $b_0$ to $\code{RSP+8}$.
%
%
The source and target execution continue with transitions $X_1$ and
$X_2$ to the external calls to $p$, return from $p$ and go on until
they return from \code{encrypt}.

\subsubsection{Horizontal Composition via Kripke Memory Relations}

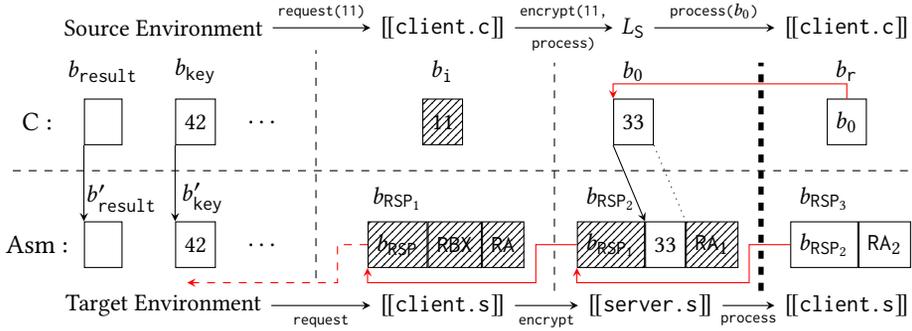
\begin{figure}
\def\bheight{0.6cm}
\begin{tikzpicture}

  \node (senv) {\small Source Environment};
  \node [right = 1.3 of senv] (sclient) {\small $\sem{\code{client.c}}$};
  \node [right = 1.3 of sclient] (ls) {\small $L_{\texttt{S}}$};
  \node [right = 1.6 of ls] (sclient1) {\small $\sem{\code{client.c}}$};
  \draw [-stealth] (senv) -- node[above] {\tiny \texttt{request(11)}} (sclient);
  \draw [-stealth] (sclient) -- 
        node[above] {\tiny \texttt{encrypt(11,}}
        node[below] {\tiny \texttt{process)}} (ls);
  \draw [-stealth] (ls) -- node[above] {\tiny \texttt{process($b_0$)}} (sclient1);


  \hblock{\bheight}{draw, minimum width = 0.5cm,
                    below left = 0.7 and -0.9 of senv} (bres) {};
  \node[above = 0.1 of bres] (brestxt) {\small $b_{\texttt{result}}$};

  \node[left = 0.3 of bres] (ctxt) {C :};

  \hblock{\bheight}{draw, right = 0.7 of bres} (bkey) {\small $42$};
  \path let \p1 = (bkey) in let \p2 = (brestxt) in
      node (bkeytxt) at (\x1, \y2) {\small $b_{\texttt{key}}$};
  
  \node[right = 0.3 of bkey] (sdots) {$\ldots$};

  \path let \p1 = (sclient) in let \p2 = (bres) in
      \hblocki{\bheight}{draw, pattern = north east lines} (bi) 
          at (\x1, \y2) {\small $11$};
  \path let \p1 = (bi) in let \p2 = (brestxt) in
      node (bitxt) at (\x1, \y2) {\small $b_{\texttt{i}}$};

  \path let \p1 = (ls) in let \p2 = (bi) in
      \hblocki{\bheight}{draw} (b0) 
          at (\x1, \y2) {\small $33$};
  \node[above = 0.1 of b0] (b0txt) {\small $b_0$};

  \path let \p1 = (sclient1) in let \p2 = (bi) in
      \hblocki{\bheight}{draw} (br) 
          at (\x1, \y2) {\small $b_0$};
  \path let \p1 = (br) in let \p2 = (brestxt) in
      node (brtxt) at (\x1, \y2) {\small $b_{\texttt{r}}$};
  

  \hblock{\bheight}{draw, minimum width = 0.5cm,
                    below = 1 of bres} (bresp) {};
  \node[above right = 0 and -0.6 of bresp] (bresptxt) 
       {\small $b_{\texttt{result}}'$};

  \path let \p1 = (ctxt) in let \p2 = (bresp) in
      node at (\x1, \y2) {Asm :};

  \path let \p1 = (bkey) in let \p2 = (bresp) in
     \hblocki{\bheight}{draw} (bkeyp) at (\x1, \y2)
         {\small $42$};
  \path let \p1 = (bkeyp) in let \p2 = (bresptxt) in 
      node (bkeyptxt) at ($(\x1, \y2)+(0.1,0)$) {\small $b_{\texttt{key}}'$};

  \path let \p1 = (sdots) in let \p2 = (bresp) in
      node (tdots) at (\x1, \y2) {$\ldots$};
  
  \path let \p1 = (bi) in let \p2 = (bresp) in
      \hblocki{\bheight}{draw, pattern = north east lines} (rsp) 
              at ($(\x1, \y2)+(-0.6,0)$) {\small $b_{\texttt{RSP}}$};
  \hblock{\bheight}{draw, right = 0 of rsp, pattern = north east lines} (rbx) 
        {\small $\texttt{RBX}$};
  \hblock{\bheight}{draw, right = 0 of rbx, pattern = north east lines} (ra) 
        {\small $\texttt{RA}$};
  \path let \p1 = (rsp) in let \p2 = (bresptxt) in 
      node (rsptxt) at (\x1, \y2) {\small $b_{\texttt{RSP}_1}$};

  \path let \p1 = (b0) in let \p2 = (rsp) in
      \hblocki{\bheight}{draw, pattern = north east lines} (rsp1) 
             at ($(\x1, \y2)+(-0.3,0)$) {\small $b_{\texttt{RSP}_1}$};
  \hblock{\bheight}{draw, right = 0 of rsp1} (output) {\small $33$};
  \hblock{\bheight}{draw, right = 0 of output, pattern = north east lines} (ra1) 
        {\small $\texttt{RA}_1$};
  \path let \p1 = (rsp1) in let \p2 = (bresptxt) in 
      node (rsp2txt) at (\x1, \y2) {\small $b_{\texttt{RSP}_2}$};

  \path let \p1 = (br) in let \p2 = (rsp) in
      \hblocki{\bheight}{draw} (rsp2) 
             at ($(\x1, \y2)+(-0.3,0)$) {\small $b_{\texttt{RSP}_2}$};
  \hblock{\bheight}{draw, right = 0 of rsp2} (ra2) 
        {\small $\texttt{RA}_2$};
  \path let \p1 = (rsp2) in let \p2 = (bresptxt) in 
      node (rsp3txt) at (\x1, \y2) {\small $b_{\texttt{RSP}_3}$};
  

  \draw[-stealth] (bres.south west) -- (bresp.north west);

  \draw[-stealth] (bkey.south west) -- (bkeyp.north west);


  \draw[-stealth] (b0.south west) -- (output.north west);
  \draw[dotted] (b0.south east) -- (output.north east);



  \path let \p1 = (senv) in let \p2 = (bresp) in
      node (tenv) at ($(\x1,\y2) + (0, -0.8)$) {\small Target Environment};
  \path let \p1 = (sclient) in let \p2 = (tenv) in
      node (tclient) at (\x1, \y2) {\small $\sem{\code{client.s}}$};
  \path let \p1 = (output) in let \p2 = (tenv) in
      node (server) at ($(\x1, \y2)+(-0.2,0)$) {\small $\sem{\code{server.s}}$};
  \path let \p1 = (sclient1) in let \p2 = (tenv) in
      node (tclient1) at (\x1, \y2) {\small $\sem{\code{client.s}}$};

  \draw[-stealth] (tenv) -- 
      node[sloped,below] {\tiny \texttt{request}} (tclient);
  \draw[-stealth] (tclient) -- 
      node[sloped,below] {\tiny \texttt{encrypt}} (server);
  \draw[-stealth] (server) -- 
      node[sloped,below] {\tiny \texttt{process}} (tclient1);

  
  \draw[-,dashed] ($(tdots.south east) + (0.4, -0.3)$) -- +(0, 3);
  \draw[-,dashed] ($(ra.south east) + (0.4, -0.3)$) -- +(0, 3);
  \draw[-,dashed,line width = 2] ($(ra1.south east) + (0.3, -0.3)$) -- +(0, 3);

  \path let \p1 = ($(ctxt.south west) + (0, -0.4)$) in
        let \p2 = (ra2.east) in
        (\x1,\y1) edge[dashed] (\x2, \y1);

        
  \draw[-stealth,red] (br) -- ++(0,0.5) -| (b0.north west);
  \draw[-stealth,red,dashed] (rsp) -- ++(-0.8,0) -- ++(0, -0.5) -- ++(-2,0);
  \draw[-stealth,red] (rsp1) -- ++(-1,0) -- ++(0, -0.5) -| (rsp.south west);
  \draw[-stealth,red] (rsp2) -- ++(-1,0) -- ++(0, -0.5) -| (rsp1.south west);

\end{tikzpicture}

\caption{Snapshot of the Memory State after Call Back}
\label{fig:mem-protect-exm}
\end{figure}

The Kripke Memory Relation (KMR) $\kinjp$ provides essential
protection for private values on the stack, which ensures
that simulations between heterogeneous modules can be established and
their horizontal composition is feasible.

We illustrate these points via our running example. Assume that the environment
calls \code{request} in the client with $11$ which in turn calls
\code{encrypt} in the server to get the value $11 \;
\code{XOR}\; 42 = 33$ whose address is passed back to the client by
calling \code{process}.  \figref{fig:mem-protect-exm} depicts a
snapshot of the memory states and the injection right after
\code{process} is entered (i.e., at
line 8 in~\figref{fig:client}), where boxes denote allocated memory 
blocks, black arrows between blocks
denote injections, and red arrows denote pointers.
The source semantics allocates one block for each local
variable ($b_i$ for \code{i}, $b_0$ for the encrypted value $33$ and $b_r$
for \code{r}) while the target semantics stores their values in
registers or stacks ($11$ is stored in \code{RDI} while $33$ on the stack 
because its address is taken and may be modified by
the callee). One stack frame is allocated for each function call which
stores private data including pointers to previous frames
($b_{\texttt{RSP}}$), return addresses (\code{RA}), and callee-saved
registers (e.g., \code{RBX}).

\kinjp is essential for proving simulation for open modules as it
guarantees simulation can be re-established after external calls
return. Informally, at every external call site, $\kinjp$ marks all
memory regions outside the footprint (domain and image) of the current
injection as \emph{private} and does not allow the external call to
modify those memory regions.
%
%
From the perspective of \code{server.s}, when the snapshot
in~\figref{fig:mem-protect-exm} is taken, the execution is inside the
thick dashed line in~\figref{fig:server-refinement} and protected by
\kinjp. Therefore, all the shaded memory
in~\figref{fig:mem-protect-exm} are marked as private and protected
against the callback to \code{process}. Indeed, they correspond to
either memory values turned into temporary variables (e.g., $b_i$) or
private stack data (e.g., $b_{\code{RSP}}$, \code{RBX} and \code{RA} in block $b_{\texttt{RSP}_1}$)
that should not be touched by \code{process}.
Such protection ensures that when \code{process} returns, all the
private values are still valid,
thereby re-establishing the simulation invariant.

The role of \kinjp is reversed for the incoming calls from the
environment: it guarantees that the entire execution from the initial
query to the final reply will not touch any private memory of the
environment. Therefore, \kinjp is used to impose a reliance on 
memory protection by external calls and 
to provide a \emph{symmetric} guarantee of memory
protection for the environment callers. Any simulations with
compatible language interfaces that satisfy this rely-guarantee
condition can be horizontally composed. For example, we can
horizontally compose
$\refinedi{\sem{\code{client.s}}}{\sem{\code{client.c}}}{\texttt{ac}}$ and
$\refinedi{\sem{\code{server.s}}}{L_{\code{S}}}{\texttt{ac}}$ into
$\refinedi{\sem{\code{client.s}} \semlink \sem{\code{server.s}}}
         {\sem{\code{client.c}} \semlink L_{\code{S}}}{\texttt{ac}}$ 
in~\figref{fig:running-exm-refinement}.

%

\subsection{Uniform and Transitive KMR for Vertical Composition of Direct Refinements}
\label{ssec:key-idea-vcomp}

Direct refinements are only useful if they can be vertically composed,
which is critical for composing refinements obtained from individual
compiler passes into a single top-level refinement such as
$\refinesymb_{\code{ac}}$ and for further composition with source-level
refinements as shown in~\figref{fig:running-exm-refinement}.

We discuss our approach for addressing this problem by using CompCert
and CompCertO as the concrete platforms.
It is based on the following two observations.
First, \kinjp in fact captures the rely-guarantee conditions for
memory protection needed by every compiler pass in
CompCert. At a high-level, it means that the rely-guarantee conditions
as depicted in~\figref{fig:open-sim-basics} can all be replaced
by \kinjp (modulo the details on language interfaces).
Second, \kinjp is transitively composable, i.e., any vertical pairing
of \kinjp can be proved equivalent to a single \kinjp. It means that
given two refinements $L_2 \refinesymb_{12} L_1$ and $L_3
\refinesymb_{23} L_2$ as depicted in~\figref{sfig:naive-vcomp},
when their rely-guarantee conditions are uniformly represented by
\kinjp, they can be merged into the direct
refinement $L_3 \refinesymb_{13} L_1$ in ~\figref{sfig:real-vcomp}
with a single \kinjp as the rely-guarantee condition.
We shall present the technical challenges leading to these
observations in~\secref{sec:background-challenges} and elaborate on the observations themselves in~\secref{sec:injp}.

By the above observations, an obvious approach for applying direct
refinements to realistic optimizing compilers is to 
prove open simulation for every
compiler pass using \kinjp, and
vertically compose those simulations into a single
simulation. However, for a non-trivial compiler like CompCert, it
means we need to rewrite a significant part of its proofs. More
importantly, optimization passes in CompCert need additional rely-guarantee
conditions as they are based on value
analysis.
To address the first problem, we start from the refinement proofs with
least restrictive KMRs for individual passes in
CompCertO~\cite{compcerto}, and exploit the properties that these KMRs
can eventually be ``absorbed'' into \kinjp in vertical composition to
generate a direct refinement parameterized only by \kinjp. To
address the second problem, we propose a notion of \emph{semantic
invariant} that captures the rely-guarantee conditions for value
analysis. When piggybacked onto \kinjp, this semantic invariant can be
transitively composed along with \kinjp and eventually pushed to the C
level. It then becomes a condition for enabling optimizations at the
source level, e.g., for supporting the refinement of the optimized
server in~\figref{fig:server_opt}. We discuss those solutions
in~\secref{sec:refinement}. 

Finally, we observe that source-level refinements can also be
parameterized by \kinjp, which enables end-to-end program verification
as depicted in~\figref{fig:running-exm-refinement} as we shall
discuss in~\secref{sec:application}.




\section{Background and Challenges}
\label{sec:background-challenges}

\subsection{Background}
\label{ssec:background}

We introduce necessary background, including the memory model, the
framework for simulation-based refinement, and \kinjp which is
critical for direct refinements.



\subsubsection{Block-based Memory Model}
\label{sssec:mem-model}

By~\citet{compcert-mem-v2}, a memory state $m$ (of type \kmem)
consists of a disjoint set of \emph{memory blocks} with unique
identifiers and linear address space. A memory address or pointer $(b,
o)$ points to the $o$-th byte in the block $b$ where $b$ has type
\kblock and $o$ has type $\kz$ (integers). The value at $(b,o)$ is
denoted by $m[b,o]$. Values (of type \kval) are either undefined
(\Vundef), 32- or 64-bit integers or floats, or pointers of the form
$\vptr{b}{o}$.
For simplicity, we often write $b$ for $\vptr{b}{0}$.
The memory operations including allocation, free, read and write
 are provided and governed by permissions of cells.
The permission of a memory cell is ordered from high to
low as $\kfreeable \geqslant \kwritable \geqslant \kreadable \geqslant
\knonempty$ where \kfreeable enables all operations, \kwritable
enables all but free, \kreadable enables only read, and \knonempty
enables none. If $p_1 \geqslant p_2$ then any cell with permission
$p_1$ also implicitly has permission $p_2$. $\perm{m}{p}$ denotes the
set of memory cells with at least permission $p$. For example, $(b, o)
\in \perm{m}{\kreadable}$ iff the cell at $(b,o)$ in $m$ is
readable. An address with no permission at all is not in the footprint
of memory.

Transformations of memory states are captured via partial functions
$j: \kblock \to \some{\kblock \times \kz}$ called \emph{injection
  functions}, s.t. $j(b) = \none$ if $b$ is removed from memory
and $j(b) = \some{(b', o)}$ if $b$ is shifted (injected) to $(b',o)$
in the target memory. We define $\kmeminj = \kblock \to \some{\kblock
  \times \kz}$.
$v_1$ and $v_2$ are related under $j$ (denoted by
$\vinj{j}{v_1}{v_2}$) if either $v_1$ is \Vundef, or they are both
equal scalar values, or pointers shifted according to $j$,
i.e., $v_1 = \vptr{b}{o}$, $j(b) = \some{(b',o')}$ and $v_2 =
\vptr{b'}{o+o'}$.

Given this relation, there is a \emph{memory injection} between the source memory
state $m_1$ and the target state $m_2$ under $j$ (denoted by
$\minj{j}{m_1}{m_2}$) if the following properties are satisfied which
ensure preservation of permissions and values under injection:
\begin{tabbing}
  \quad\=\;\;\=\kill
  \> $\forall\app b_1\app b_2\app o\app o'\app p,\app
  j(b_1) = \some{(b_2,o')} \imply
  (b_1, o) \in \perm{m_1}{p} \imply
  (b_2, o+o') \in \perm{m_2}{p}.$\\
  \> $\forall\app b_1\app b_2\app o\app o',\app
  j(b_1) = \some{(b_2,o')} \imply
  (b_1,o) \in \perm{m_1}{\kreadable}
  \imply \vinj{j}{m_1[b_1,o]}{m_2[b_2,o+o']}.$
\end{tabbing}
Memory injections are \emph{transitive} and necessary for verifying compiler transformations
of memory structures (e.g., merging local variables into
stack-allocated data and generating a concrete stack frame). For the
remaining passes, a simpler relation called \emph{memory extension} is
used instead, which employs an identity injection.
Reasoning about permissions under refinements is a major source of complexity.

\subsubsection{A Framework for Open Simulations}
\label{sssec:framework}


%

In CompCertO~\cite{compcerto}, a \emph{language interface} $A = \linterface{A^q}{A^r}$ is a pair of
sets $A^q$ and $A^r$ denoting acceptable queries and replies for open
modules, respectively. Different interfaces may be used for
different languages. The relevant ones for our discussion
have been introduced in~\secref{ssec:key-idea-direct-refinement} and
listed as follows:
\begin{center}
  \begin{tabular}{c c c c}
    \hline
    \textbf{Languages} & \textbf{Interfaces} & \textbf{Queries} & \textbf{Replies}\\
    \hline
    C/Clight & $\cli = \linterface{\kval \times \ksig \times \kval^* \times \kmem}{\kval \times \kmem}$ 
       & $\cquery{v_f}{\sig}{\vec{v}}{m}$ & $\creply{v'}{m'}$\\
    Asm & $\asmli = \linterface{\kregset \times \kmem}{\kregset \times \kmem}$ & $\asmquery{\regset}{m}$ & $\asmreply{\regset'}{m'}$
  \end{tabular}
\end{center}

\emph{Open labeled transition systems} (LTS) represent semantics of
modules that may accept queries and provide replies at the
\emph{incoming side} and provide queries and accept replies at the
\emph{outgoing side} (i.e., calling external functions). An open LTS
$L : A \arrli B$ is a tuple $\olts{D}{S}{I}{\to}{F}{X}{Y}$ where $A$
($B$) is the language interface for outgoing (incoming) queries and
replies, $D \subseteq B^q$ a set of initial queries, $S$ a set of
internal states, $I \subseteq D \times S$ ($F \subseteq S \times B^r$)
transition relations for incoming queries (replies), $X \subseteq S
\times A^q$ ($Y \subseteq S \times A^r \times S$) transitions
for outgoing queries (replies), and $\to \subseteq S \times \mathcal{E}^* \times
S$ internal transitions emitting events of type $\mathcal{E}$. Note
that $(s, q^O) \in X$ iff an outgoing query $q^O$ happens at $s$; $(s,
r^O, s') \in Y$ iff after $q^O$ returns with $r^O$ the
execution continues with an updated state $s'$. 
%
%

\emph{Kripke relations} are used to describe evolution of program states in
open simulations between LTSs. A Kripke relation $R : W \to \pset{S}{S
  \subseteq A \times B}$ is a family of relations indexed by a
\emph{Kripke world} $W$; for simplicity, we define $\krtype{W}{A}{B} =
W \to \pset{S}{S \subseteq A \times B}$. A \emph{simulation
  convention} relating two language interfaces $A_1$ and $A_2$ is a
tuple $\scname{R} = \simconv{W}{\scname{R}^q :
  \krtype{W}{A^q_1}{A^q_2}}{\scname{R}^r :
  \krtype{W}{A^r_1}{A^r_2}}$ which we write as $\scname{R}:
\sctype{A_1}{A_2}$. Simulation conventions serve as interfaces of open
simulations by relating source and target language interfaces.
For example, a C-level convention $\kc: \sctype{\cli}{\cli} =
\simconv{\kmeminj}{\scname{R}_{\kwd{c}}^q}{\scname{R}_{\kwd{c}}^r}$
relates C queries and replies as follows, where the Kripke world
consists of injections and, in a given world $j$, the values and
memory in queries and replies are related by $j$.
\begin{tabbing}
  \quad\=$(\cquery{v_f}{sg}{\vec{v}}{m}, \cquery{v_f'}{sg}{\vec{v'}}{m'}) \in \scname{R}_{\kwd{c}}^q(j)$\quad\=$\iff$\quad\=\kill
  \>$(\cquery{v_f}{sg}{\vec{v}}{m}, \cquery{v_f'}{sg}{\vec{v'}}{m'}) \in \scname{R}_{\kwd{c}}^q(j)$
  \>$\iff$\>$\vinj{j}{v_f}{v_f'} \land \vinj{j}{\vec{v}}{\vec{v'}} \land \minj{j}{m}{m'}$\\
  \>$(\creply{v}{m}, \creply{v'}{m'}) \in {\scname{R}_{\kwd{c}}^r(j)}$
  \>$\iff$\>$\vinj{j}{v}{v'} \land \minj{j}{m}{m'}$
\end{tabbing}

\begin{figure}
\begin{center}
\begin{tikzpicture}[
    world/.style={draw,minimum height=0.6cm, minimum width=0.6cm}
  ]
  
  \node (q1) {\small $q_1$};
  \node[right = 0.7 of q1] (s1) {\small $s_1$};
  \node[right = 1 of s1] (s1p) {\small $s_1'$};
  \node[right = 0.7 of s1p] (q1p) {\small $q_1'$};
  \node[right = 2 of q1p] (r1p) {\small $r_1'$};
  \node[right = 1 of r1p] (s1pp) {\small $s_1''$};
  \node[right = 1 of s1pp] (s1ppp) {\small $s_1'''$};
  \node[right = 0.7 of s1ppp] (r1) {\small $r_1$};

  \node[below = 1 of q1] (q2) {\small $q_2$};

  \path let \p1 = (s1) in let \p2 = (q2) in
     node (s2) at (\x1,\y2) {\small $s_2$};
  \path let \p1 = (s1p) in let \p2 = (q2) in
     node (s2p) at (\x1,\y2) {\small $s_2'$};
  \path let \p1 = (q1p) in let \p2 = (q2) in
     node (q2p) at (\x1,\y2) {\small $q_2'$};
  \path let \p1 = (r1p) in let \p2 = (q2) in
     node (r2p) at (\x1,\y2) {\small $r_2'$};
  \path let \p1 = (s1pp) in let \p2 = (q2) in
     node (s2pp) at (\x1,\y2) {\small $s_2''$};
  \path let \p1 = (s1ppp) in let \p2 = (q2) in
     node (s2ppp) at (\x1,\y2) {\small $s_2'''$};
  \path let \p1 = (r1) in let \p2 = (q2) in
     node (r2) at (\x1,\y2) {\small $r_2$};

  \draw[-stealth] (q1) -- node[sloped,above] {\small $I_1$} (s1);
  \draw[-stealth, dashed] (s1) -- (s1p);
  \draw[-stealth] (s1p) -- node[sloped,above] {\small $X_1$} (q1p);
  \draw[-stealth, line width = 2, dashed] (q1p) -- (r1p);
  \draw[-stealth] (r1p) -- node[sloped,above] {\small $Y_1(s_1')$} (s1pp);
  \draw[-stealth, dashed] (s1pp) -- (s1ppp);
  \draw[-stealth] (s1ppp) -- node[sloped,above] {\small $F_1$} (r1);

  \draw[-stealth] (q2) -- node[sloped,above] {\small $I_2$} (s2);
  \draw[-stealth, dashed] (s2) -- (s2p);
  \draw[-stealth] (s2p) -- node[sloped,above] {\small $X_2$} (q2p);
  \draw[-stealth, line width = 2, dashed] (q2p) -- (r2p);
  \draw[-stealth] (r2p) -- node[sloped,above] {\small $Y_2(s_2')$} (s2pp);
  \draw[-stealth, dashed] (s2pp) -- (s2ppp);
  \draw[-stealth] (s2ppp) -- node[sloped,above] {\small $F_2$} (r2);

  \draw [double, latex-latex] (q2) -- 
      node[right] (q12) {\small $\scname{R}^q_B$} (q1);
  \draw [double, latex-latex] (s2) -- 
      node[left] {\small $R$} (s1);
  \draw [double, latex-latex] (s2p) -- 
      node[left] {\small $R$} (s1p);
  \draw [double, latex-latex] (q2p) -- 
      node[left] (q12p) {\small $\scname{R}^q_A$} (q1p);
  \draw [double, latex-latex] (r2p) -- 
      node[right] (r12p) {\small $\scname{R}^r_A$} (r1p);
  \draw [double, latex-latex] (s2pp) -- 
      node[left] {\small $R$} (s1pp);
  \draw [double, latex-latex] (s2ppp) -- 
      node[left] {\small $R$} (s1ppp);
  \draw [double, latex-latex] (r2) -- 
      node[left] (r12) {\small $\scname{R}^r_B$} (r1);

  \node[world, left = 0.2 of q12] (wb) {$w_B$};
  \node[world, right = 0.2 of r12] (wbp) {$w_B'$};

  \node[world, right = 0.2 of q12p] (wa) {$w_A$};
  \node[world, left = 0.2 of r12p] (wap) {$w_A'$};

  \path (wa) -- node (acca) {$\accsymb_A$} (wap);
  \draw[-stealth, dashed] (wb) -- +(0,1.3) -| (wbp);
  \node [above = 1 of acca] {$\accsymb_B$};

\end{tikzpicture}  
\end{center}
\caption{Open Simulation between LTS}
\label{fig:open-sim}
\end{figure}
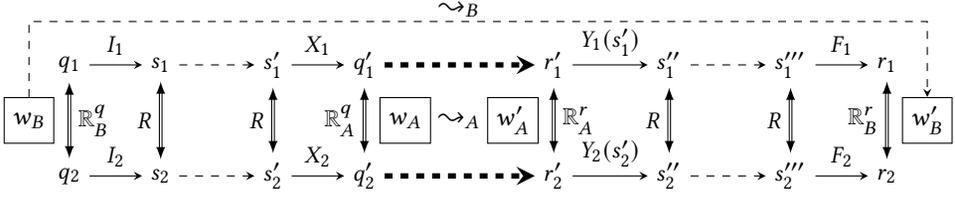

\emph{Open forward simulations} describe refinement between LTS. To
establish an open (forward) simulation between $L_1: A_1 \arrli B_1$
and $L_2: A_2 \arrli B_2$, one needs to find two simulation
conventions $\scname{R}_A : \sctype{A_1}{A_2}$ and $\scname{R}_B :
\sctype{B_1}{B_2}$ that connect queries and replies at the outgoing
and incoming sides, and show the internal execution steps and external
interactions of open modules are related by an invariant $R$.
This simulation is denoted by
$\osim{\scname{R}_A}{\scname{R}_B}{L_1}{L_2}$ and formally defined as
follows (for simplicity, we shall write $\osims{\scname{R}}{L_1}{L_2}$
to denote $\osim{\scname{R}}{\scname{R}}{L_1}{L_2}$):
\begin{definition}\label{def:open-sim}
  Given $L_1: A_1 \arrli B_1$, $L_2: A_2 \arrli B_2$, $\scname{R}_A :
  \sctype{A_1}{A_2}$ and $\scname{R}_B : \sctype{B_1}{B_2}$,
  $\osim{\scname{R}_A}{\scname{R}_B}{L_1}{L_2}$ holds if there is some
  Kripke relation $R \in \krtype{W_B}{S_1}{S_2}$ that satisfies:
  \begin{tabbing}
    \quad\=\quad(1)\=\quad\=\kill
    \>(1) $\forall\app q_1\app q_2,\app (q_1, q_2) \in \scname{R}_B^q(w_B) \imply (q_1 \in D_1 \iff q_2 \in D_2)$\\
    \>(2) $\forall\app w_B\app q_1\app q_2\app s_1,\app (q_1, q_2) \in \scname{R}_B^q(w_B) \imply (q_1, s_1) \in I_1 \imply
    \exists\app s_2, (s_1, s_2) \in R(w_B) \land (q_2, s_2) \in I_2.$\\
    \>(3) $\forall\app w_B\app s_1\app s_2\app t,\app (s_1, s_2) \in R(w_B) \imply \step{s_1}{t}{s_1'} \imply 
    \exists\app s_2', (s_1', s_2') \in R(w_B) \land s_2 \overset{t}{\rightarrow^*} s_2'.$\\
    \>(4) $\forall\app w_B\app s_1\app s_2\app q_1,\app (s_1, s_2) \in R(w_B) \imply (s_1, q_1) \in X_1 \imply$\\
    \>\>$\exists w_A\app q_2,\app (q_1, q_2) \in \scname{R}_A^q(w_A) \land (s_2, q_2) \in X_2\; \land$\\
    \>\>\>$\forall\app r_1\app r_2\app s_1', (r_1, r_2) \in \scname{R}_A^r(w_A) \imply (s_1, r_1, s_1') \in Y_1 \imply \exists\app s_2', (s_1', s_2') \in R(w_B) \land (s_2, r_2, s_2') \in Y_2.$\\
    \>(5) $\forall\app w_B\app s_1\app s_2\app r_1,\app (s_1, s_2) \in R(w_B) \imply (s_1, r_1) \in F_1 \imply 
    \exists\app r_2, (r_1, r_2) \in \scname{R}_B^r(w_B) \land (s_2, r_2) \in F_2.$
  \end{tabbing}
\end{definition}
\noindent Here, property (1) requires initial queries to match; (2) requires
initial states to hold under the invariant $R$; (3) requires internal
execution to preserve $R$; (4) requires $R$ to be preserved across
external calls, and (5) requires final replies to match.
According to these properties, a complete forward simulation looks
like~\figref{fig:open-sim}.
From the above definition, it is easy to prove the horizontal and
vertical compositionality of open simulations and adequacy for
assembly modules, i.e., 
$\forall \app L_1\app L_2\app L_1'\app L_2', \app
    \osims{\scname{R}}{L_1}{L_2} \imply
    \osims{\scname{R}}{L_1'}{L_2'} \imply
    \osims{\scname{R}}{L_1 \semlink L_1'}{L_2 \semlink
      L_2'}$ 
and
$\forall \app (M_1 \app M_2 : \kwd{Asm}), \app
     \osims{\kwd{id}}{\sem{M_1} \semlink \sem{M_2}}{\sem{M_1 + M_2}}$.

%
%

The Kripke worlds (e.g., memory injections) may evolve as the
execution goes on. \emph{Rely-guarantee} reasoning about such evolution
is essential for horizontal composition of simulations. 
For illustration, the Kripke worlds
at the boundary of modules are displayed in~\figref{fig:open-sim}.
The evolution of worlds across external calls is governed by an
\emph{accessibility relation} $w_A \accsymb_A w_A'$ for describing the
\emph{rely-condition}. By assuming $w_A \accsymb_A w_A'$, one needs to prove the
\emph{guarantee condition} $w_B \accsymb_B w_B'$, i.e., the evolution of worlds in the whole execution respects $\accsymb_B$.
Simulations with symmetric rely-guarantee conditions can
be horizontally composed, even with mutual calls between modules.

Note that the accessibility relation and evolution of worlds between
queries and replies is not encoded explicitly in the definition of
simulation conventions. Instead, they are implicit by assuming a
modality operator $\Diamond$ is always applied to $\scname{R}^r$
s.t. $r \in {\mkrel{\scname{R}^r}(w)} \iff \exists\app w', w \accsymb
w' \land r \in \scname{R}^r(w')$. For simplicity, we often ignore
accessibility and modality when talking \emph{purely} about
simulation conventions in the remaining discussion.

Accessibility relations are mainly for describing evolution of memory
states across external calls. For this, simulation conventions are
parameterized by \emph{Kripke Memory Relations} or KMR. 
\begin{definition}
  A Kripke Memory Relation is a tuple
  $\cklr{W}{f}{\accsymb}{R}$ where $W$ is a set of worlds, $f:W \to
  \kmeminj$ a function for extracting injections from worlds,
  $\leadsto \subseteq W \times W$ an accessibility relation between
  worlds and $R: \krtype{W}{\kmem}{\kmem}$ a Kripke relation over
  memory states that is compatible with the memory operations.
  We write ${w} \accsymb {w'}$ for $(w, w') \in \accsymb$.
\end{definition}
\noindent We write $\scname{R}_K$ to emphasize that a simulation convention
$\scname{R}$ is parameterized by the KMR ${K}$, meaning $\scname{R}_K$
shares the same type of worlds with $K$ and inherits its accessibility
relation.

The most interesting KMR is $\kinjp$ as it provides protection on
memory w.r.t. injections.
\begin{definition}[Kripke Relation with Memory Protection]\label{def:injp}
  $\kinjp = \cklr{W_{\kinjp}}{f_\kinjp}{\accsymb_{\kinjp}}{R_{\kinjp}}$ where
  $W_{\kinjp} = (\kmeminj \times \kmem \times \kmem)$, $f_\kinjp(j, \_,\_) = j$, 
  $(m_1, m_2) \in R_\kinjp(j, m_1, m_2) \iff \minj{j}{m_1}{m_2}$ and 
  \begin{tabbing}
    \quad\=$\injpacc{(j, m_1, m_2)}{(j', m_1', m_2')} \; \iff \;$\=\kill
    \>$\injpacc{(j, m_1, m_2)}{(j', m_1', m_2')} \; \iff \;j \subseteq j' \land \unmapped{j} \subseteq \unchangedon{m_1}{m_1'}$\\
    \>\>$\land\; \outofreach{j}{m_1} \subseteq \unchangedon{m_2}{m_2'}.$\\
    \>\>$\land\; \macc{m_1}{m_1'} \land\; \macc{m_2}{m_2'}$
  \end{tabbing}
  Here, $\macc{m}{m'}$ denotes monotonicity of memory states such as
  valid blocks can only increase and read-only data does not change in
  value.
  $\unchangedon{m}{m'}$ denotes memory cells whose permissions
    and values are not changed from $m$ to $m'$ and 
    \begin{tabbing}
      \quad\=$(b_2, o_2) \in \outofreach{j}{m_1}$\;\=\kill
      \>$(b_1, o_1) \in \unmapped{j}$\>$\iff \; j(b_1) = \none$ \\
      \>$(b_2, o_2) \in \outofreach{j}{m_1}$\>$\iff \;
         \forall\app b_1\app o_2', \app j(b_1) = \some{(b_2, o_2')} \imply
         (b_1, o_2-o_2') \not\in \perm{m_1}{\knonempty}.$
    \end{tabbing}
\end{definition}
\begin{wrapfigure}{R}{.45\textwidth}
\begin{tikzpicture}
  \snode{0.5cm}{0.3cm}{0}{draw} (n1) {};
  \snode{0.4cm}{0.3cm}{0}{draw, right = 0.3cm of n1, pattern=north east lines} (n2) {};
  \snode{0.6cm}{0.3cm}{0}{draw, right = 0.3cm of n2} (n3) {};
  \snode{0.4cm}{0.3cm}{0}{draw, right = 0.3cm of n3, pattern=north east lines} (n4) {};

  \draw[dashed] ($(n4.north east) + (0.4cm, 0.1cm)$) --++(0,-1.5cm);

  \snode{0.2cm}{0.3cm}{0}{draw, right = 0.6cm of n4} (n5) {};
  \snode{0.4cm}{0.3cm}{0}{draw, right = 0.3cm of n5} (n6) {};

  \snode{2.5cm}{0.3cm}{0}{draw, below left = 0.6cm and -2.3cm of n1} (m1) {};
  \snode{0.5cm}{0.3cm}{0}{draw, right = 0.3cm of m1, pattern=north east lines} (m2) {};
  \snode{0.8cm}{0.3cm}{0}{draw, right = 0.3cm of m2} (m3) {};

  \draw[-stealth] (n1.south west) -- ($(m1.north west)+(0.4cm,0)$);
  \draw[dotted] (n1.south east) -- ($(m1.north west)+(0.9cm,0)$);
  \draw[-stealth] (n3.south west) -- ($(m1.north west)+(1.5cm,0)$);
  \draw[dotted] (n3.south east) -- ($(m1.north west)+(2.1cm,0)$);
  \draw[-stealth] (n6.south west) -- ($(m3.north west)+(0.2cm,0)$);
  \draw[dotted] (n6.south east) -- ($(m3.north west)+(0.6cm,0)$);

  \fill[pattern=north east lines] (m1.north west) rectangle ($(m1.south west)+(0.4cm,0)$);
  \fill[pattern=north east lines] ($(m1.north west)+(0.9cm,0)$) rectangle ($(m1.south west)+(1.5cm,0)$);
  \fill[pattern=north east lines] ($(m1.north west)+(2.1cm,0)$) rectangle (m1.south east);
  \draw ($(m1.north west)+(0.4cm,0)$) -- ($(m1.south west)+(0.4cm,0)$);
  \draw ($(m1.north west)+(0.9cm,0)$) -- ($(m1.south west)+(0.9cm,0)$);
  \draw ($(m1.north west)+(1.5cm,0)$) -- ($(m1.south west)+(1.5cm,0)$);
  \draw ($(m1.north west)+(2.1cm,0)$) -- ($(m1.south west)+(2.1cm,0)$);

  \node[left = 0.4cm of n1] (txtm1) {\small $m_1$};
  \path let \p1 = (txtm1) in let \p2 = (m1) in
    node (txtm2) at (\x1, \y2) {\small $m_2$};
  \path let \p1 = (txtm1) in let \p2 = (txtm2) in
    node (txtj1) at (\x1, {(\y1+\y2)*0.5}) {\small $j$};

  \node[right = 0.2cm of n6] (txtmp1) {\small $m_1'$};
  \path let \p1 = (txtmp1) in let \p2 = (m1) in
    node (txtmp2) at (\x1, \y2) {\small $m_2'$};
  \path let \p1 = (txtmp1) in let \p2 = (txtmp2) in
    node (txtjp1) at (\x1, {(\y1+\y2)*0.5}) {\small $j'$};

  \end{tikzpicture}
\caption{Kripke Worlds Related by \kinjp}
\label{fig:injp}
\end{wrapfigure}
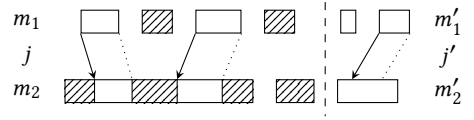
By definition, a world $(j, m_1, m_2)$ evolves to $(j', m_1',
m_2')$ under \kinjp only if $j'$ is strictly larger than $j$ and any
memory cells in $m_1$ and $m_2$ not in the domain (i.e., unmapped by
$j$), or image of $j$ (i.e., out-of-reach by $j$ from $m_1$) will be
protected, meaning their values and permissions are unchanged from
$m_1$ ($m_2$) to $m_1'$ ($m_2'$). An example is shown
in~\figref{fig:injp} where the shaded regions in $m_1$
are unmapped by $j$ and unchanged while those in $m_2$ are
out-of-reach from $j$ and unchanged. $m_1'$ and $m_2'$ may contain
newly allocated blocks which are not protected by
\kinjp. When \kinjp is used at the outgoing side, it denotes that the
simulation relies on knowing that the unmapped and out-of-reach
regions at the call side are not modified by external calls. When
\kinjp is used at the incoming side, it denotes that the simulation
guarantees such regions at initial queries
are not modified by the simulation itself.

\subsection{Challenges for Vertically Composing Open Simulations}
\label{ssec:challenges}

%
As discussed in~\secref{ssec:key-idea-vcomp}, the challenge for
constructing direct refinements for multi-pass optimizing compilers
lies in their vertical composition. 
The most basic vertical composition for open simulations is stated
below which is easily proved by pairing of individual
simulations~\cite{compcerto}.
\begin{theorem}[V. Comp]\label{thm:v-comp}
  Given $L_1 : A_1 \arrli B_1$, $L_2 : A_2 \arrli B_2$ and $L_3: A_3
  \arrli B_3$, and given $\scname{R}_{12}:\sctype{A_1}{A_2}$,
  $\scname{S}_{12}:\sctype{B_1}{B_2}$, $\scname{R}_{23}:\sctype{A_2}{A_3}$
  and $\scname{S}_{23}:\sctype{B_2}{B_3}$,
    \[
    \osim{\scname{R}_{12}}{\scname{S}_{12}}{L_1}{L_2} \imply
    \osim{\scname{R}_{23}}{\scname{S}_{23}}{L_2}{L_3} \imply
    \osim{\comp{\scname{R}_{12}}{\scname{R}_{23}}}
         {\comp{\scname{S}_{12}}{\scname{S}_{23}}}
         {L_1}{L_3}.
    \]
\end{theorem}
\noindent Here, $(\_ \compsymb \_)$ is a composed simulation convention s.t.
$\comp{\scname{R}}{\scname{S}} = \simconv{W_{\scname{R}} \times
  W_{\scname{S}}} {{\scname{R}^q} \compsymb {\scname{S}^q}}
{{\scname{R}^r} \compsymb {\scname{S}^r}}$ where for any $q_1$ and
$q_3$, ${ (q_1, q_3) \in {\scname{R}^q} \compsymb
  {\scname{S}^q}}(w_{\scname{R}}, w_{\scname{S}}) \iff \exists q_2,
(q_1, q_2) \in {\scname{R}^q(w_{\scname{R}})} \land (q_2, q_3) \in
{\scname{S}^q(w_{\scname{S}})}$ (similarly for ${{\scname{R}^r}
  \compsymb {\scname{S}^r}}$). 
%
Then, given any compiler with $N$ passes and their
refinement relations $\osim{\scname{R}_{12}}{\scname{S}_{12}}{L_1}{L_2},
\ldots, \osim{\scname{R}_{N,N+1}}{\scname{S}_{N,N+1}}{L_N}{L_{N+1}}$,
we get their concatenation
$\osim {\scname{R}_{12} \compsymb \ldots \compsymb \scname{R}_{N,N+1}}
{\scname{S}_{12} \compsymb \ldots \compsymb \scname{S}_{N,N+1}}
{L_1}{L_{N+1}}$, which exposes internal compilation and weakens
compositionality as we have discussed in~\secref{ssec:problems}.

The above problem may be solved if the composed simulation convention can
be \emph{refined} into a single convention directly relating source
and target queries and replies. 
Given two simulation conventions $\scname{R}, \scname{S}:
\sctype{A_1}{A_2}$, $\scname{R}$ is \emph{refined} by $\scname{S}$ if
\begin{tabbing}
  \quad\=$\forall\app w_{\scname{S}}\app q_1\app q_2,\app (q_1, q_2) \in \scname{S}^q(w_{\scname{S}}) \imply \exists\app w_{\scname{R}},\app$\=\kill
  \>$\forall\app w_{\scname{S}}\app q_1\app q_2,\app
  (q_1, q_2) \in \scname{S}^q(w_{\scname{S}}) \imply
  \exists\app w_{\scname{R}},\app
  (q_1, q_2) \in \scname{R}^q(w_{\scname{R}}) \land$\\
  \>\>$\forall\app r_1\app r_2,\app (r_1,r_2) \in \scname{R}^r(w_{\scname{R}}) \imply 
(r_1, r_2) \in \scname{S}^r(w_{\scname{S}})$
\end{tabbing}
which we write as $\screfine{\scname{R}}{\scname{S}}$. If both
$\screfine{\scname{R}}{\scname{S}}$ and
$\screfine{\scname{S}}{\scname{R}}$, then $\scname{R}$ and
$\scname{S}$ are equivalent and written as
$\scequiv{\scname{R}}{\scname{S}}$.
By definition, $\screfine{\scname{R}}{\scname{S}}$ indicates any query
for ${\scname{S}}$ can be converted into a query for ${\scname{R}}$
and any reply resulting from the converted query can be 
converted back to a reply for ${\scname{S}}$. 
%
%
By wrapping 
the incoming side of an open simulation with a more general
convention and its outgoing side with a more specialized convention,
one gets another valid open simulation~\cite{compcerto}:
\begin{theorem}\label{thm:sim-refine}
  Given $L_1 : A_1 \arrli B_1$ and $L_2 : A_2 \arrli B_2$, if\;
  $\screfine{\scname{R}'_A}{\scname{R}_A}:\sctype{A_1}{A_2}$,
  $\screfine{\scname{R}_B}{\scname{R}'_B}:\sctype{B_1}{B_2}$ and
  $\osim{\scname{R}_A}{\scname{R}_B}{L_1}{L_2}$, then
  $\osim{\scname{R}'_A}{\scname{R}'_B}{L_1}{L_2}$.
\end{theorem}

\begin{figure}
  \centering
  \begin{tikzpicture}

    \node (q1i) at (0,0) {\small $q_1$};
    \node (q3i) [below = 1.5cm of q1i] {\small $q_3$};
    \path (q1i) -- node[] (q2i) {\small ${q_2}$} (q3i);

    \node (q1ix) [left = 0.5cm of q1i] {\small $q_1$};
    \node (q3ix) [left = 0.5cm of q3i] {\small $q_3$};
    \path (q1ix) -- node[] (q2ix) {} (q3ix);

    \node (l1) [left = 0.5cm of q1ix] {\small $L_1:$};
    \node (l3) [left = 0.5cm of q3ix] {\small $L_3:$};
    \path (l1) -- node[] (l2) {\small $L_2:$} (l3);
    
    \node (q1o) [right = 1.5 cm of q1i] {\small $q_1'$};
    \node (q3o) [right = 1.5 cm of q3i] {\small $q_3'$};
    \path (q1o) -- node[] (q2o) {\small $q_2'$} (q3o);

    \node (q1ox) [right = 0.5 cm of q1o] {\small $q_1'$};
    \node (q3ox) [right = 0.5 cm of q3o] {\small $q_3'$};
    \path (q1ox) -- node[] (q2ox) {} (q3ox);

    \node (r1ox) [right = 1.5 cm of q1ox] {\small $r_1'$};
    \node (r3ox) [right = 1.5 cm of q3ox] {\small $r_3'$};
    \path (r1ox) -- node[] (r2ox) {} (r3ox);

    \node (r1o) [right = 0.5 cm of r1ox] {\small $r_1'$};
    \node (r3o) [right = 0.5 cm of r3ox] {\small $r_3'$};
    \path (r1o) -- node[] (r2o) {\small ${r_2'}$} (r3o);
  

    \node (r1i) [right = 1.5 cm of r1o] {\small $r_1$};
    \node (r3i) [right = 1.5 cm of r3o] {\small $r_3$};
    \path (r1i) -- node[] (r2i) {\small $r_2$} (r3i);
    
    \node (r1ix) [right = 0.5 cm of r1i] {\small $r_1$};
    \node (r3ix) [right = 0.5 cm of r3i] {\small $r_3$};
    \path (r1ix) -- node[] (r2ix) {} (r3ix);
    
  \draw [double, latex-latex] (q1ix) -- node[left] {\small $\scr_{13}^q$} (q3ix);
  \draw [double, latex-latex] (q1i) -- node[right] {\small $\scr_{12}^q$} (q2i);
  \draw [double, latex-latex] (q2i) -- node[right] {\small $\scr_{23}^q$} (q3i);
  
  \draw [double, latex-latex] (q1o) -- node[left] {\small $\scr_{12}^q$} (q2o);
  \draw [double, latex-latex] (q2o) -- node[left] {\small $\scr_{23}^q$} (q3o);

  \draw [double, latex-latex] (q1ox) -- node[right] {\small $\scr_{13}^q$} (q3ox);

  \draw [double, latex-latex] (r1ox) -- node[left] {\small $\scr_{13}^r$} (r3ox);

  \draw [double, latex-latex] (r1o) -- node[right] {\small $\scr_{12}^r$} (r2o);
  \draw [double, latex-latex] (r2o) -- node[right] {\small $\scr_{23}^r$} (r3o);

  \draw [double, latex-latex] (r1i) -- node[left] {\small $\scr_{12}^r$} (r2i);
  \draw [double, latex-latex] (r2i) -- node[left] {\small $\scr_{23}^r$} (r3i);

  \draw [double, latex-latex] (r1ix) -- node[right] {\small $\scr_{13}^r$}(r3ix);  
  \draw [-stealth, dashed] (q1i) -- (q1o);
  \draw [-stealth, dashed] (q2i) -- (q2o);
  \draw [-stealth, dashed] (q3i) -- (q3o);
  \draw [-stealth, dashed] (r1o) -- (r1i);
  \draw [-stealth, dashed] (r2o) -- (r2i);
  \draw [-stealth, dashed] (r3o) -- (r3i);

  \path (q2ix) -- node {\small $\imply$} (q2i);
  \path (q2o) -- node {\small $\imply$} (q2ox);
  \path (r2ox) -- node {\small $\imply$} (r2o);
  \path (r2i) -- node {\small $\imply$} (r2ix);
  \path (q2ox) -- node {\small $\accsymb$} (r2ox);

  \node[draw, dashed, fit=(q1i) (q3i) (q1o) (q3o)] (box) {};
  \node[draw, dashed, fit=(r1i) (r3i) (r1o) (r3o)] (box) {};

\end{tikzpicture}
\caption{Vertical Composition of Open Simulations by Refinement of Simulation Conventions}
\label{fig:opensimcomp}
\end{figure}
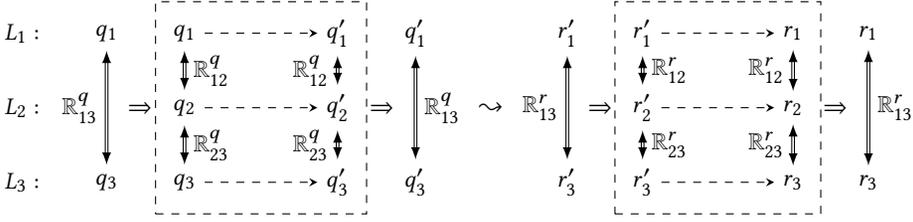
Now, we would like to prove the ``real'' vertical composition
generating direct refinements (simulations).
Given any $\osim{\scname{R}_{12}}{\scname{R}_{12}}{L_1}{L_2}$ and
$\osim{\scname{R}_{23}}{\scname{R}_{23}}{L_2}{L_3}$, if we can show
the existence of simulation conventions $\scname{R}_{13}$ directly
relating source and target semantics s.t.
$\scequiv{\scname{R}_{13}}{\comp{\scname{R}_{12}}{\scname{R}_{23}}}$,
then $\osim{\scname{R}_{13}}{\scname{R}_{13}}{L_1}{L_3}$ holds
by~\thmref{thm:v-comp} and~\thmref{thm:sim-refine}, which is the
desired direct refinement.
%
This composition is illustrated in \figref{fig:opensimcomp} where the
parts enclosed by dashed boxes represent the concatenation of
$\osim{\scname{R}_{12}}{\scname{R}_{12}}{L_1}{L_2}$ and
$\osim{\scname{R}_{23}}{\scname{R}_{23}}{L_2}{L_3}$.
%
%
The direct queries and replies are split and merged for interaction
with parallelly running simulations underlying the direct refinement.

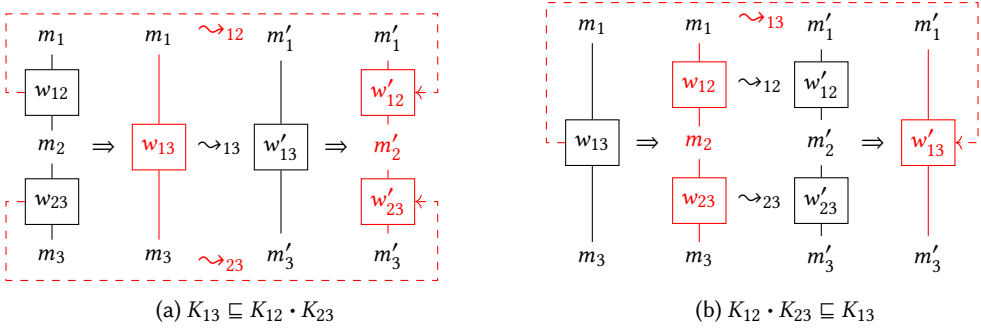
\begin{figure}
  \begin{subfigure}[b]{0.46\textwidth}
  \begin{tikzpicture}[
      world/.style = {draw,minimum height=0.6cm, minimum width=0.6cm}
    ]
    \node (m1) {\small $m_1$};
    \node [below = 1cm of m1] (m2) {\small $m_2$};
    \node [below = 1cm of m2] (m3) {\small $m_3$};
    \path (m1) -- node[world, draw] (w1) {\small $w_{12}$} (m2);
    \path (m2) -- node[world, draw] (w2) {\small $w_{23}$} (m3);
    \draw (m1) -- (w1) -- (m2);
    \draw (m2) -- (w2) -- (m3);

    \node [right = 0.8cm of m1] (m11) {\small $m_1$};
    \path let \p1 = (m11) in let \p2 = (m3) in 
      node (m33) at (\x1, \y2) {\small $m_3$};
    \path (m11) -- node[world, draw,red] (w) {\small $w_{13}$} (m33);
    \draw[red] (m11) -- (w) -- (m33);

    \node [right = 1cm of m11] (m11p) {\small $m_1'$};
    \path let \p1 = (m11p) in let \p2 = (m33) in 
      node (m33p) at (\x1, \y2) {\small $m_3'$};
    \path (m11p) -- node[world, draw] (wp) {\small $w_{13}'$} (m33p);
    \draw (m11p) -- (wp) -- (m33p);

    \node [right = 0.8cm of m11p] (m1p) {\small $m_1'$};
    \path let \p1 = (m1p) in let \p2 = (m2) in 
      node[red] (m2p) at (\x1, \y2) {\small $m_2'$};
    \path let \p1 = (m1p) in let \p2 = (m33) in 
      node (m3p) at (\x1, \y2) {\small $m_3'$};
    \path (m1p) -- node[world, draw,red] (w1p) {\small $w_{12}'$} (m2p);
    \path (m2p) -- node[world, draw,red] (w2p) {\small $w_{23}'$} (m3p);

    \draw[red] (m1p) -- (w1p) -- (m2p);
    \draw[red] (m2p) -- (w2p) -- (m3p);

    \path (m2) -- node {\small $\imply$} (w);
    \path (w) -- node {\small $\accsymb_{13}$} (wp);
    \path (wp) -- node {\small $\imply$} (m2p);


    \draw[->, dashed, red] 
          (w1) -| ($(m1.north west) + (-0.3, 0.1)$) 
               -| node[sloped, below, pos = 0.25] 
                  {$\accsymb_{12}$} ($(w1p.east) + (0.3, 0)$)
               -- (w1p);
    \draw[->, dashed, red] 
          (w2) -| ($(m3.south west) + (-0.3, -0.1)$) 
               -| node[sloped, above, pos = 0.25] 
                  {$\accsymb_{23}$} ($(w2p.east) + (0.3, 0)$)
               -- (w2p);

  \end{tikzpicture}
  \caption{$\screfine{K_{13}}{\comp{K_{12}}{K_{23}}}$}
  \label{fig:trans-comp1}
  \end{subfigure}
  \qquad
  \begin{subfigure}[b]{0.46\textwidth}
  \begin{tikzpicture}[
      world/.style = {draw,minimum height=0.6cm, minimum width=0.6cm}
    ]
    \node (m1) {\small $m_1$};
    \node [world, below = 1cm of m1,draw] (w) {\small $w_{13}$};
    \node [below = 1cm of w] (m3) {\small $m_3$};
    
    \node [right = 0.8cm of m1] (m11) {\small $m_1$};
    \path let \p1 = (m11) in let \p2 = (w) in 
      node[red] (m2) at (\x1, \y2) {\small $m_2$};
    \path let \p1 = (m11) in let \p2 = (m3) in 
      node (m33) at (\x1, \y2) {\small $m_3$};
    \path (m11) -- node[world, draw,red] (w1) {\small $w_{12}$} (m2);
    \path (m2) -- node[world, draw,red] (w2) {\small $w_{23}$} (m33);

    \draw (m1) -- (w) -- (m3);
    \draw[red] (m11) -- (w1) -- (m2);
    \draw[red] (m2) -- (w2) -- (m33);

    \node [right = 1cm of m11] (m1p) {\small $m_1'$};
    \path let \p1 = (m1p) in let \p2 = (m2) in 
      node (m2p) at (\x1, \y2) {\small $m_2'$};
    \path let \p1 = (m1p) in let \p2 = (m33) in 
      node (m3p) at (\x1, \y2) {\small $m_3'$};
    \path (m1p) -- node[world, draw] (w1p) {\small $w_{12}'$} (m2p);
    \path (m2p) -- node[world, draw] (w2p) {\small $w_{23}'$} (m3p);

    \draw (m1p) -- (w1p) -- (m2p);
    \draw (m2p) -- (w2p) -- (m3p);

    \node [right = 0.8cm of m1p] (m11p) {\small $m_1'$};
    \path let \p1 = (m11p) in let \p2 = (m3p) in 
      node (m33p) at (\x1, \y2) {\small $m_3'$};
    \path (m11p) -- node[world, draw,red] (w3p) {\small $w_{13}'$} (m33p);
    
    \draw[red] (m11p) -- (w3p) -- (m33p);

    \path (w) -- node {\small $\imply$} (m2);
    \path (w1) -- node {\small $\accsymb_{12}$} (w1p);
    \path (w2) -- node {\small $\accsymb_{23}$} (w2p);
    \path (m2p) -- node {\small $\imply$} (w3p);


    \draw[->, dashed, red] 
          (w) -| ($(m1.north west) + (-0.3, 0.1)$) 
              -| node[sloped, below, pos = 0.25] 
                  {$\accsymb_{13}$} ($(w3p.east) + (0.3, 0)$)
              -- (w3p);

  \end{tikzpicture}
  \caption{{$\screfine{\comp{K_{12}}{K_{23}}}{K_{13}}$}}
  \label{fig:trans-comp2}
\end{subfigure}
  \caption{Composition of KMRs}
  \label{fig:trans-comp}
\end{figure}

Since simulation conventions are parameterized by KMRs, a major obstacle to
the real vertical composition of open simulations is to prove KMRs for individual
simulations can be composed into a single KMR.
For this, one needs to define refinements between KMRs. Given any KMRs
$K$ and $L$, $\screfine{K}{L}$ (i.e., $K$ is refined by $L$) holds if
the following is true:
\begin{tabbing}
  \quad\=$\forall\app m_1\app m_2\app w_L,\app (m_1, m_2) \in$\=$ R_L(w_L)$\=$ \imply
  \exists\app$\kill
  \>$\forall\app w_L\app m_1\app m_2,\app (m_1, m_2) \in
  R_L(w_L) \imply \exists\app w_K,\app (m_1, m_2) \in R_K(w_K) \land
  f_L(w_L) \subseteq f_K(w_K)\; \land$\\ 
  \>\>$\forall\app w_K'\app
  m_1'\app m_2',\app \acc{K}{w_K}{w_K'} \imply (m_1', m_2') \in
  R_K(w_K') \imply$\\ 
  \>\>\>$\exists\app w_L',\app
  \acc{L}{w_L}{w_L'} \land (m_1', m_2') \in R_L(w_L') \land
  f_K(w_K') \subseteq f_L(w_L')$.
\end{tabbing}
We write $K \equiv L$ to denote that $K$ and $L$ are equivalent, i.e.,
$\screfine{K}{L}$ and $\screfine{L}{K}$.

Continue with the proof of real vertical composition, 
i.e., proving $\scequiv{\scname{R}_{13}}{\comp{\scname{R}_{12}}{\scname{R}_{23}}}$. Assume
$\scname{R}_i$ is parameterized by KMR $K_i$, showing the existence of
$\scname{R}_{13}$ s.t.
$\screfine{\scname{R}_{13}}{\comp{\scname{R}_{12}}{\scname{R}_{23}}}$
amounts to proving a parallel refinement over the parameterizing KMRs,
i.e., there exists $K_{13}$ s.t.
$\screfine{K_{13}}{\comp{K_{12}}{K_{23}}}$ where $\comp{K_{12}}{K_{23}}
= \cklr{W_{12} \times W_{23}}{f_{12} \times f_{23}}{\accsymb_{12}
  \times \accsymb_{23}}{R_{12} \times R_{23}}$. A more intuitive
interpretation is depicted in~\figref{fig:trans-comp1} where black
symbols are $\forall$-quantified (assumptions we know) and red ones
are $\exists$-quantified (conclusions we need to construct). Note that
\figref{fig:trans-comp1} exactly mirrors the refinement on the
outgoing side in~\figref{fig:opensimcomp}.
For simplicity, we use $w_i$ not only to represent worlds, but also to
denote $R_i(w_i)$ (where $R_i$ is the Kripke relation given by KMR
$K_i$) when it connects memory states through vertical lines.
A dual property we need to prove for the incoming side is shown
in~\figref{fig:trans-comp2}.

In both cases in~\figref{fig:trans-comp}, we need to construct
interpolating states for relating source and target memory (i.e.,
$m_2'$ in~\figref{fig:trans-comp1} and $m_2$
in~\figref{fig:trans-comp2}). 
The construction of $m_2'$ is especially challenging, for which we
need to decompose the evolved world $w_{13}'$ into $w_{12}'$ and
$w_{23}'$ s.t. they are accessible from the original worlds $w_{12}$
and $w_{23}$. It is not clear at all how this construction is possible
because \emph{1)} $m_2'$ may have many forms since Kripke
\emph{relations} are in general non-deterministic
and \emph{2)} KMRs (e.g., \kinjp) introduce memory protection for
external calls which may not hold after the (de-)composition.

Because of the above difficulties, existing approaches either make
substantial changes to semantics for constructing interpolating
states, thereby destroying adequacy~\cite{stewart15}, or do not even try
to merge Kripke memory relations, but instead leave them as separate
entities~\cite{compcertm, compcerto}. As a result, direct refinements
cannot be achieved.

\section{A Uniform and Transitive Kripke Memory Relation}
\label{sec:injp}


To overcome the challenge for vertically composing open simulations,
we exploit the observation that \kinjp in fact can be viewed as a most
general KMR. Then, the compositionality of KMRs discussed
in~\secref{ssec:challenges} is reduced to transitivity of \kinjp,
i.e., $\kinjp \equiv \comp{\kinjp}{\kinjp}$. 

\subsection{Uniformity of \kinjp}
\label{ssec:injp-uniform}

We show that \kinjp is both a reasonable guarantee condition and a
reasonable rely condition for all the compiler passes in CompCert.  It
is based on the observation that a notion of private and public memory
can be derived from injections and coincides with the
protection provided by \kinjp.

\subsubsection{Public and Private Memory via Memory Injections}

\begin{definition}\label{def:pub-mem}
Given $\minj{j}{m_1}{m_2}$, the public memory regions in $m_1$ and
$m_2$ are defined as follows:
\begin{tabbing}
  \quad\=$\pubtgtmem{j}{m_1}$\;\=\kill
  \>$\pubsrcmem{j}$\>$ = \pset{(b,o)}{j(b) \neq \none};$\\
  \>$\pubtgtmem{j}{m_1}$\>$ = \pset{(b,o)}{\exists
    b'\app o', j(b') = \some{(b, o')} \land (b',o-o') \in
    \perm{m_1}{\knonempty}}.$
\end{tabbing}
\end{definition}
By definition, a cell $(b, o)$ is public in the source memory if it is
in the domain of $j$, and $(b,o)$ is public in the target memory
if it is mapped by $j$ from some valid public source memory. Any memory not
public with respect to $j$ is private. We can see that private memory
corresponds exactly to unmapped and out-of-reach memory defined by
\kinjp, i.e., for any $b$ and $o$, $(b, o) \in \pubsrcmem{j} \iff (b,
o) \not\in \unmapped{j}$ and $(b, o) \in \pubtgtmem{j}{m} \iff (b, o)
\not\in \outofreach{j}{m}$.

\begin{figure}
\begin{minipage}[b]{.35\textwidth}
\begin{figure}[H]
\centering
\begin{tikzpicture}
  \path node (n1) {\small $(b_1, o_1)$} 
        --++(2,0) node (n2) {\small $(b_2, o_2)$}
        --++(1.8,0) node (n3) {$\ldots$};

  \draw[-stealth] (n1) -- node[sloped, above] {\small \kwd{read}} (n2);
  \draw[-stealth] (n2) -- node[sloped, above] {\small \kwd{read}} (n3);

  \node[below = 0.8cm of n1] (m1) {\small $(b_1', o_1')$};
  \path let \p1 = (n2) in let \p2 = (m1) in
     node (m2) at (\x1,\y2) {\small $(b_2', o_2')$};
  \path let \p1 = (n3) in let \p2 = (m1) in
     node (m3) at (\x1,\y2) {$\ldots$};

  \draw[-stealth] (m1) -- node[sloped, above] {\small \kwd{read}} (m2);
  \draw[-stealth] (m2) -- node[sloped, above] {\small \kwd{read}} (m3);

  \draw[-stealth] (n1) -- node[left] {$j$} (m1);
  \draw[-stealth] (n2) -- node[left] {$j$} (m2);

\end{tikzpicture}  
\caption{Closure of Public Memory}
\label{fig:read-inj}
\end{figure}
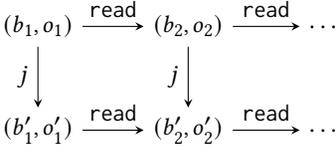
\end{minipage}
\begin{minipage}[b]{0.6\textwidth}
  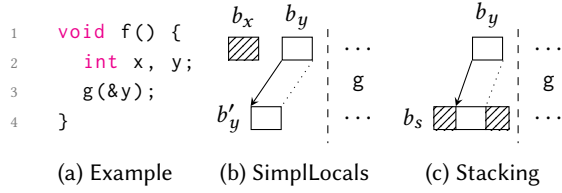
\begin{figure}[H]
  \centering
  \begin{subfigure}[b]{0.27\textwidth}
    \centering
\begin{lstlisting}[language = C]
  void f() { 
    int x, y; 
    g(&y); 
  }
\end{lstlisting}   
  \caption{Example}
  \label{fig:injp-rely-code}
  \end{subfigure}
  %
  %
  \begin{subfigure}[b]{0.28\textwidth}
    \centering
    \begin{tikzpicture}
      \snode{0.4cm}{0.3cm}{0}{draw, pattern=north east lines} (x) {};
      \node[above = 0 of x] (bx) {\small $b_x$};
      \snode{0.4cm}{0.3cm}{0}{draw, right = 0.3cm of x} (y) {};
      \path let \p1 = (bx) in let \p2 = (y) in
          node at (\x2, \y1) {\small $b_y$};
      
      \snode{0.4cm}{0.3cm}{0}{draw, below left = 0.6cm and 0.0cm of y} (yp) {};
      \node[left = 0 of yp] {\small $b_y'$};
      
      \draw[-stealth] (y.south west) -- (yp.north west);
      \draw[dotted] (y.south east) -- (yp.north east);

      \draw[dashed] ($(y.north east) + (0.2cm, 0.1cm)$) --++(0,-1.5cm);

      \snode{0.4cm}{0.3cm}{0}{right = 0.4cm of y} (b1) {$\ldots$};
      \path let \p1 = (b1) in let \p2 = (yp) in
        node (b2) at (\x1, \y2) {$\ldots$};

      \path let \p1 = (b1) in let \p2 = (b2) in
        node at (\x1, {(\y1+\y2)*0.5}) {\small \kwd{g}};
    \end{tikzpicture}
  \caption{SimplLocals}
  \label{fig:injp-rely-simpl}
  \end{subfigure}
  %
  %
  \begin{subfigure}[b]{0.3\textwidth}
    \centering
    \begin{tikzpicture}
      \snode{0.4cm}{0.3cm}{0}{draw} (y) {};
      \node[above = 0 of y] {\small $b_y$};
      
      \snode{1cm}{0.3cm}{0}{draw, below left = 0.6cm and -0.5cm of y} (sf) {};
      \node[left = 0 of sf] {\small $b_s$};
      
      \draw[-stealth] (y.south west) -- ($(sf.north west)+(0.3cm, 0)$);
      \draw[dotted] (y.south east) -- ($(sf.north west)+(0.7cm,0)$);

      \fill[pattern=north east lines] (sf.north west) rectangle ($(sf.south west)+(0.3cm,0)$);
      \fill[pattern=north east lines] ($(sf.north west)+(0.7cm,0)$) rectangle (sf.south east);
      \draw ($(sf.north west)+(0.3cm,0)$) -- ($(sf.south west)+(0.3cm,0)$);
      \draw ($(sf.north west)+(0.7cm,0)$) -- ($(sf.south west)+(0.7cm,0)$);

      \draw[dashed] ($(y.north east) + (0.2cm, 0.1cm)$) --++(0,-1.5cm);

      \snode{0.4cm}{0.3cm}{0}{right = 0.4cm of y} (b1) {$\ldots$};
      \path let \p1 = (b1) in let \p2 = (yp) in
        node (b2) at (\x1, \y2) {$\ldots$};

      \path let \p1 = (b1) in let \p2 = (b2) in
        node at (\x1, {(\y1+\y2)*0.5}) {\small \kwd{g}};
    \end{tikzpicture}
  \caption{Stacking}
  \label{fig:injp-rely-stacking}
\end{subfigure}
  \caption{Protection of Private Memory by \kinjp}
  \label{fig:injp-rely}
  \end{figure}
\end{minipage}
\end{figure}

With~\defref{def:pub-mem} and the properties of memory injection
(see~\secref{sssec:mem-model}), we can easily prove access of pointers
in a readable and public source location gets back another public
location.
\begin{lemma}\label{lem:pub-closure}
  Given $\minj{j}{m_1}{m_2}$,
  \begin{tabbing}
    \quad\=$\forall\app b_1\app o_1,\app$\=\kill
    \>$\forall\app b_1\app o_1,\app (b_1, o_1) \in \pubsrcmem{j}
                         \imply (b_1, o_1) \in \perm{m_1}{\kreadable}
                         \imply$\\
    \>\>$\mcontents{m_1}{b_1}{o_1} = \vptr{b_1'}{o_1'}
         \imply (b_1', o_1') \in \pubsrcmem{j}.$
  \end{tabbing}
\end{lemma}
\noindent It implies that readable public memory regions form a
``closure'' such that the sequences of reads are bounded inside these
regions, as shown in~\figref{fig:read-inj}.
The horizontal arrows indicates a pointer value $(b_{i+1}, o_{i+1})$
is read from $(b_i, o_i)$ with possible adjustment with pointer
arithmetic. Note that all memory cells at $(b_i, o_i)$s and $(b_i',
o_i')$s have \kreadable permission. By~\lemref{lem:pub-closure},
$(b_i, o_i)$s are all in public regions. By~\defref{def:pub-mem}, the
mirroring reads $(b_i', o_i')$s are also in public regions.

\subsubsection{\kinjp as a Uniform Rely Condition}
%
%
\kinjp is adequate for preventing external calls from interfering with
internal execution for all the compiler passes of CompCert.~\footnote{In
fact, the properties in~\defref{def:injp} are exactly from CompCert's assumptions on external
calls.} To illustrate this point, we discuss the
effect of \kinjp on two of CompCert's passes using ~\figref{fig:injp-rely-code}
as an example where \kwd{g} is an
external function.
The first pass is \kwd{SimplLocals} which converts local variables
whose memory addresses are not taken into temporary ones. As shown in
~\figref{fig:injp-rely-simpl}, $\kwd{x}$ is turned into a temporary variable
at the target level which is not visible to $\kwd{g}$.  Therefore, $\kwd{x}$ at
the source level becomes \emph{private data} as its block $b_x$ is
unmapped by $j$, thereby protected by \kinjp and cannot be modified by
$\kwd{g}$.
%
%
%
The second pass is \kwd{Stacking} which expands the stack frames with
private regions for return addresses, spilled registers, arguments,
etc. 
Continuing with our example, the only public stack data in
~\figref{fig:injp-rely-stacking} is $y$. All the private data is
out-of-reach, thereby protected by \kinjp.

\subsubsection{\kinjp as a Uniform Guarantee Condition}
%
%
%

For \kinjp to serve as a uniform guarantee condition, it suffices to
show the private memory of the environment is protected between initial
calls and final replies.
%
During an open forward simulation,
all incoming values and memories are related by some initial injection
$j$ (e.g., $\vinj{j}{\vec{v_1}}{\vec{v_2}}$ and
$\minj{j}{m_1}{m_2}$). In particular, the pointers in them are related
by $j$. Therefore, any sequence of reads starting from pointers stored
in the initial queries only inspect public memories in the source
and target, as already shown in~\figref{fig:read-inj}.
The private
(i.e., unmapped or out-of-reach) regions of the initial memories are
\emph{not modified} by internal execution.
Moreover, because injection functions only grow bigger during execution
but never change in value and the outgoing calls have \kinjp as a
rely-condition, the initially unmapped (out-of-reach) regions will
stay unmapped (out-of-reach) and be protected during external calls.
As a result, we conclude that $\kinjp$ is a reasonable
guarantee condition for any open simulation.

\subsection{Transitivity of \kinjp}
\label{ssec:injp-trans}
The goal is to show the two refinements in~\figref{fig:trans-comp}
hold when $K_{ij} = \kinjp$, i.e., $\kinjp \equiv
\comp{\kinjp}{\kinjp}$. As discussed in~\secref{ssec:challenges} the
critical step is to construct interpolating memory states that
transitively relate source and target states.
The construction is based on two observations: \emph{1)} the memory
injections deterministically decide the value and permissions of
public memory because they encode \emph{partial functional
transformations} on memory states, and \emph{2)} any memory not in the
domain or range of the partial functions is protected (private) and
unchanged throughout external calls. Although the proof is quite
involved, the result can be reused for all compiler passes thanks to
\kinjp's uniformity.
The formal proof of transitivity of \kinjp can be found in
~\apdxref{sec:injp-trans-proof}.

%
%
%

\subsubsection{$\screfine{\kinjp}{\comp{\kinjp}{\kinjp}}$}
By definition, we need to prove the following lemma:
\begin{lemma}\label{lem:injp-refine-injp-comp}
$\screfine{\kinjp}{\comp{\kinjp}{\kinjp}}$ holds. That is,
\begin{tabbing}
  \quad\=\quad\=\quad\=$\exists m_2'\app j_{12}'\app j_{23}',$\=\kill
  \>$\forall j_{12}\app j_{23}\app m_1\app m_2\app m_3,\app \minj{j_{12}}{m_1}{m_2}
  \imply \minj{j_{23}}{m_2}{m_3} \imply \exists j_{13},\app \minj{j_{13}}{m_1}{m_3} \app \land$\\
  \>\>$\forall m_1'\app m_3'\app j_{13}',\app \injpacc{(j_{13}, m_1, m_3)}{(j_{13}', m_1', m_3')} \imply \minj{j_{13}'}{m_1'}{m_3'}  \imply$\\
  \>\>\>$\exists m_2'\app j_{12}'\app j_{23}', \injpacc{(j_{12},m_1,m_2)}{(j_{12}',m_1',m_2')} \land \minj{j_{12}'}{m_1'}{m_2'}$\\
  \>\>\>\>$\land \injpacc{(j_{23},m_2,m_3)}{(j_{23}',m_2',m_3')} \land \minj{j_{23}'}{m_2'}{m_3'}.$
\end{tabbing}
\end{lemma}

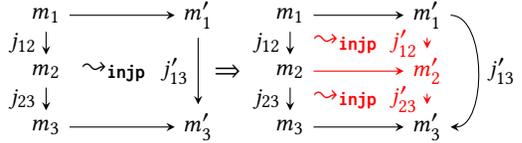
\begin{wrapfigure}{R}{.5\textwidth}
  \center
  \begin{tikzpicture}
    \node (m1) {\small $m_1$};
    \node[below = 0.3cm of m1] (m2) {\small $m_2$};
    \node[below = 0.3cm of m2] (m3) {\small $m_3$};

    \draw[-stealth] (m1) -- node[left] {\small $j_{12}$} (m2);
    \draw[-stealth] (m2) -- node[left] {\small $j_{23}$} (m3);

    \node[right = 1.4cm of m1] (m1p) {\small $m_1'$};
    \node[right = 1.4cm of m3] (m3p) {\small $m_3'$};
    \path let \p1 = (m2) in let \p2 = (m1p) in
      node at ({(\x1+\x2)*0.45}, \y1) {\bf \small $\accsymb_{\kinjp}$};

    \draw[-stealth] (m1) -- (m1p);
    \draw[-stealth] (m3) -- (m3p);

    \draw[-stealth] (m1p) -- node[left] {\small $j_{13}'$} (m3p);

    \node[right = 1.8cm of m2] {$\Rightarrow$};

    \node[right = 0.6cm of m1p] (m11) {\small $m_1$};
    \path let \p1 = (m11) in let \p2 = (m2) in
      node (m22) at (\x1, \y2) {\small $m_2$};
    \path let \p1 = (m11) in let \p2 = (m3) in
      node (m33) at (\x1, \y2) {\small $m_3$};

    \draw[-stealth] (m11) -- node[left] {\small $j_{12}$} (m22);
    \draw[-stealth] (m22) -- node[left] {\small $j_{23}$} (m33);

    \node[right = 1.2cm of m11] (m11p) {\small $m_1'$};
    \node[right = 1.2cm of m33] (m33p) {\small $m_3'$};

    \draw[-stealth] (m11) -- (m11p);
    \draw[-stealth] (m33) -- (m33p);

    \draw[-stealth] (m11p) .. controls ($(m11p)+(0.8cm,0)$) and ($(m33p)+(0.8cm,0cm)$) .. node[right] {\small $j_{13}'$} (m33p);

    \node[right = 1.2cm of m22,red] (m22p) {\small $m_2'$};
    \draw[-stealth,red] (m11p) -- node[left] {\small $j_{12}'$} (m22p);
    \draw[-stealth,red] (m22p) -- node[left] {\small $j_{23}'$} (m33p);
    \draw[-stealth,red] (m22) -- (m22p);
    \path let \p1 = (m11) in let \p2 = (m22p) in
      node[red] at ({(\x1+\x2)*0.48},{(\y1+\y2)*0.5}) {\bf \small  $\accsymb_{\kinjp}$};
    \path let \p1 = (m22) in let \p2 = (m33p) in
      node[red] at ({(\x1+\x2)*0.48},{(\y1+\y2)*0.5}) {\bf \small $\accsymb_{\kinjp}$};
    \end{tikzpicture}
  \caption{Construction of Interpolating States}
  \label{fig:injp-int-st}
\end{wrapfigure}
This lemma conforms to the graphic representation
in~\figref{fig:trans-comp1}. To prove it, an obvious choice is to pick
$j_{13} = \comp{j_{23}}{j_{12}}$. Then, we are left to prove the existence of
interpolating state $m_2'$ and the memory and accessibility relations
as shown in~\figref{fig:injp-int-st}.
By definition, $m_2'$ consists of memory blocks newly allocated with
respect to $m_2$ and blocks that already exist in $m_2$. The latter
can be further divided into public and private memory regions with
respect to injections $j_{12}$ and $j_{23}$.
Then, $m_2'$ is constructed following the ideas that \emph{1)} the
public and newly allocated memory should be projected from the updated
source memory $m_1'$ by $j_{12}'$, and \emph{2)} the private memory is
protected by \kinjp and should be copied over from $m_2$ to $m_2'$.
%



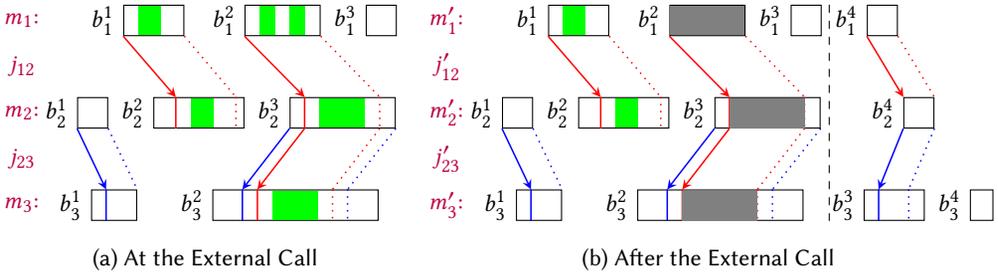
\begin{figure}[ht]
\def\sx{0}
\def\sy{0}
\def\done{0.5cm}
\def\dtwo{0.6cm}
\def\dthree{0.4cm}
\def\dfourth{0.8cm}
\def\bheight{0.4cm}
\def\fsep{0.02cm}

\begin{subfigure}[b]{0.4\textwidth}
\begin{tikzpicture}
  \snode{0.8cm}{\bheight}{0.1cm}{draw} (m1b1) {};
  \fill[green] ($(m1b1.north west)+(0.2cm,-\fsep)$) rectangle ($(m1b1.south west)+(0.5cm,\fsep)$);
  \node[left = 0 of m1b1] {\small $b_1^1$};

  \snode{1cm}{\bheight}{0.1cm}{draw, right = 0.8cm of m1b1} (m1b2) {};
  \fill[green] ($(m1b2.north west)+(0.2cm,-\fsep)$) rectangle ($(m1b2.south west)+(0.4cm,\fsep)$);
  \fill[green] ($(m1b2.north west)+(0.6cm,-\fsep)$) rectangle ($(m1b2.south west)+(0.8cm,\fsep)$);
  \node[left = 0 of m1b2] {\small $b_1^2$};

  \snode{0.4cm}{\bheight}{0.1cm}{draw, right = 0.6cm of m1b2} (m1b3) {};
  \node[left = 0 of m1b3] {\small $b_1^3$};

  \snode{0.4cm}{\bheight}{0.1cm}{draw, below left = 0.8cm and 0.2cm of m1b1} (m2b1) {};
  \node[left = 0 of m2b1] {\small $b_2^1$};

  \snode{1.2cm}{\bheight}{0.1cm}{draw, right = 0.6cm of m2b1} (m2b2) {};
  \fill[green] ($(m2b2.north west)+(0.5cm,-\fsep)$) rectangle ($(m2b2.south west)+(0.8cm,\fsep)$);
  \draw[line width=0.02cm, -stealth,red] (m1b1.south west) -- ($(m2b2.north west)+(0.3cm,0)$);
  \draw[line width=0.02cm,red] ($(m2b2.north west)+(0.3cm,0)$) -- ($(m2b2.south west)+(0.3cm,0)$);
  \draw[line width=0.02cm,dotted,red] (m1b1.south east) -- ($(m2b2.north west)+(1.1cm,0)$);
  \draw[line width=0.02cm,dotted,red] ($(m2b2.north west)+(1.1cm,0)$) -- ($(m2b2.south west)+(1.1cm,0)$);
  \node[left = 0 of m2b2] {\small $b_2^2$};

  \snode{1.4cm}{\bheight}{0.1cm}{draw, right = 0.6cm of m2b2} (m2b3) {};
  \fill[green] ($(m2b3.north west)+(0.4cm,-\fsep)$) rectangle ($(m2b3.south west)+(1cm,\fsep)$);
  \draw[line width=0.02cm, -stealth,red] (m1b2.south west) -- ($(m2b3.north west)+(0.2cm,0)$);
  \draw[line width=0.02cm,red] ($(m2b3.north west)+(0.2cm,0)$) -- ($(m2b3.south west)+(0.2cm,0)$);
  \draw[line width=0.02cm,dotted,red] (m1b2.south east) -- ($(m2b3.north west)+(1.2cm,0)$);
  \draw[line width=0.02cm,dotted,red] ($(m2b3.north west)+(1.2cm,0)$) -- ($(m2b3.south west)+(1.2cm,0)$);
  \node[left = 0 of m2b3] {\small $b_2^3$};

  \snode{0.6cm}{\bheight}{0.1cm}{draw, below left = 0.8cm and -0.8cm of m2b1} (m3b1) {};
  \draw[line width=0.02cm, -stealth,blue] (m2b1.south west) -- ($(m3b1.north west)+(0.2cm,0)$);
  \draw[line width=0.02cm,blue] ($(m3b1.north west)+(0.2cm,0)$) -- ($(m3b1.south west)+(0.2cm,0)$);
  \draw[line width=0.02cm,dotted,blue] (m2b1.south east) -- (m3b1.north east);
  \node[left = 0 of m3b1] {\small $b_3^1$};

  \snode{2.2cm}{\bheight}{0.1cm}{draw, right = 1cm of m3b1} (m3b2) {};
  \fill[green] ($(m3b2.north west)+(0.8cm,-\fsep)$) rectangle ($(m3b2.south west)+(1.4cm,\fsep)$);
  \draw[line width=0.02cm, -stealth,blue] (m2b3.south west) -- ($(m3b2.north west)+(0.4cm,0)$);
  \draw[line width=0.02cm,blue] ($(m3b2.north west)+(0.4cm,0)$) -- ($(m3b2.south west)+(0.4cm,0)$);
  \draw[line width=0.02cm,dotted,blue] (m2b3.south east) -- ($(m3b2.north west)+(1.8cm,0)$);
  \draw[line width=0.02cm,dotted,blue] ($(m3b2.north west)+(1.8cm,0)$) -- ($(m3b2.south west)+(1.8cm,0)$);
  \draw[line width=0.02cm, -stealth,red] ($(m2b3.south west)+(0.2cm,0)$) -- ($(m3b2.north west)+(0.6cm,0)$);
  \draw[line width=0.02cm,red] ($(m3b2.north west)+(0.6cm,0)$) -- ($(m3b2.south west)+(0.6cm,0)$);
  \draw[line width=0.02cm, dotted,red] ($(m2b3.south west)+(1.2cm,0)$) -- ($(m3b2.north west)+(1.6cm,0)$);
  \draw[line width=0.02cm,dotted,red] ($(m3b2.north west)+(1.6cm,0)$) -- ($(m3b2.south west)+(1.6cm,0)$);
  \node[left = 0 of m3b2] {\small $b_3^2$};

  \node[left = 1cm of m1b1] (txtm1) {\small\color{purple} $m_1$:};
  \path let \p1 = (txtm1) in let \p2 = (m2b1) in
    node (txtm2) at (\x1, \y2) {\small\color{purple} $m_2$:};
  \path let \p1 = (txtm1) in let \p2 = (m3b1) in
    node (txtm3) at (\x1, \y2) {\small\color{purple} $m_3$:};
  \path let \p1 = (txtm1) in let \p2 = (txtm2) in
    node (txtj12) at (\x1, {(\y1+\y2)*0.5}) {\small\color{purple} $j_{12}$};
  \path let \p1 = (txtm2) in let \p2 = (txtm3) in
    node (txtj23) at (\x1, {(\y1+\y2)*0.5}) {\small\color{purple} $j_{23}$};

\end{tikzpicture}
\caption{At the External Call}
\label{fig:inj-before}
\end{subfigure}
\begin{subfigure}[b]{0.55\textwidth}
\begin{tikzpicture}
  \snode{0.8cm}{\bheight}{0.1cm}{draw} (m1b1) {};
  \fill[green] ($(m1b1.north west)+(0.2cm,-\fsep)$) rectangle ($(m1b1.south west)+(0.5cm,\fsep)$);
  \node[left = 0 of m1b1] {\small $b_1^1$};

  \snode{1cm}{\bheight}{0.1cm}{draw, right = 0.8cm of m1b1} (m1b2) {};
  \fill[green] ($(m1b2.north west)+(0.2cm,-\fsep)$) rectangle ($(m1b2.south west)+(0.4cm,\fsep)$);
  \fill[green] ($(m1b2.north west)+(0.6cm,-\fsep)$) rectangle ($(m1b2.south west)+(0.8cm,\fsep)$);
  \node[left = 0 of m1b2] {\small $b_1^2$};

  \snode{0.4cm}{\bheight}{0.1cm}{draw, right = 0.6cm of m1b2} (m1b3) {};
  \node[left = 0 of m1b3] {\small $b_1^3$};

  \snode{0.4cm}{\bheight}{0.1cm}{draw, below left = 0.8cm and 0.2cm of m1b1} (m2b1) {};
  \node[left = 0 of m2b1] {\small $b_2^1$};

  \snode{1.2cm}{\bheight}{0.1cm}{draw, right = 0.6cm of m2b1} (m2b2) {};
  \fill[green] ($(m2b2.north west)+(0.5cm,-\fsep)$) rectangle ($(m2b2.south west)+(0.8cm,\fsep)$);
  \draw[line width=0.02cm, -stealth,red] (m1b1.south west) -- ($(m2b2.north west)+(0.3cm,0)$);
  \draw[line width=0.02cm,red] ($(m2b2.north west)+(0.3cm,0)$) -- ($(m2b2.south west)+(0.3cm,0)$);
  \draw[line width=0.02cm,dotted,red] (m1b1.south east) -- ($(m2b2.north west)+(1.1cm,0)$);
  \draw[line width=0.02cm,dotted,red] ($(m2b2.north west)+(1.1cm,0)$) -- ($(m2b2.south west)+(1.1cm,0)$);
  \node[left = 0 of m2b2] {\small $b_2^2$};

  \snode{1.4cm}{\bheight}{0.1cm}{draw, right = 0.6cm of m2b2} (m2b3) {};
  \fill[green] ($(m2b3.north west)+(0.4cm,-\fsep)$) rectangle ($(m2b3.south west)+(1cm,\fsep)$);
  \draw[line width=0.02cm, -stealth,red] (m1b2.south west) -- ($(m2b3.north west)+(0.2cm,0)$);
  \draw[line width=0.02cm,red] ($(m2b3.north west)+(0.2cm,0)$) -- ($(m2b3.south west)+(0.2cm,0)$);
  \draw[line width=0.02cm,dotted,red] (m1b2.south east) -- ($(m2b3.north west)+(1.2cm,0)$);
  \draw[line width=0.02cm,dotted,red] ($(m2b3.north west)+(1.2cm,0)$) -- ($(m2b3.south west)+(1.2cm,0)$);
  \node[left = 0 of m2b3] {\small $b_2^3$};

  \snode{0.6cm}{\bheight}{0.1cm}{draw, below left = 0.8cm and -0.8cm of m2b1} (m3b1) {};
  \draw[line width=0.02cm, -stealth,blue] (m2b1.south west) -- ($(m3b1.north west)+(0.2cm,0)$);
  \draw[line width=0.02cm,blue] ($(m3b1.north west)+(0.2cm,0)$) -- ($(m3b1.south west)+(0.2cm,0)$);
  \draw[line width=0.02cm,dotted,blue] (m2b1.south east) -- (m3b1.north east);
  \node[left = 0 of m3b1] {\small $b_3^1$};

  \snode{2.2cm}{\bheight}{0.1cm}{draw, right = 1cm of m3b1} (m3b2) {};
  \fill[green] ($(m3b2.north west)+(0.8cm,-\fsep)$) rectangle ($(m3b2.south west)+(1.4cm,\fsep)$);
  \draw[line width=0.02cm, -stealth,blue] (m2b3.south west) -- ($(m3b2.north west)+(0.4cm,0)$);
  \draw[line width=0.02cm,blue] ($(m3b2.north west)+(0.4cm,0)$) -- ($(m3b2.south west)+(0.4cm,0)$);
  \draw[line width=0.02cm,dotted,blue] (m2b3.south east) -- ($(m3b2.north west)+(1.8cm,0)$);
  \draw[line width=0.02cm,dotted,blue] ($(m3b2.north west)+(1.8cm,0)$) -- ($(m3b2.south west)+(1.8cm,0)$);
  \draw[line width=0.02cm, -stealth,red] ($(m2b3.south west)+(0.2cm,0)$) -- ($(m3b2.north west)+(0.6cm,0)$);
  \draw[line width=0.02cm,red] ($(m3b2.north west)+(0.6cm,0)$) -- ($(m3b2.south west)+(0.6cm,0)$);
  \draw[line width=0.02cm, dotted,red] ($(m2b3.south west)+(1.2cm,0)$) -- ($(m3b2.north west)+(1.6cm,0)$);
  \draw[line width=0.02cm,dotted,red] ($(m3b2.north west)+(1.6cm,0)$) -- ($(m3b2.south west)+(1.6cm,0)$);
  \node[left = 0 of m3b2] {\small $b_3^2$};

  \path [draw,dashed] let \p1 = (m3b1.south) in let \p2 = (m1b3.north east) in
    ({\x2 + 0.1cm}, \y2) -- ({\x2+0.1cm}, \y1);

  \snode{0.4cm}{\bheight}{0.1cm}{draw, right = 0.6cm of m1b3} (m1b4) {};
  \node[left = 0 of m1b4] {\small $b_1^4$};

  \snode{0.4cm}{\bheight}{0.1cm}{draw, right = 1.1cm of m2b3} (m2b4) {};
  \node[left = 0 of m2b4] {\small $b_2^4$};

  \snode{0.8cm}{\bheight}{0.1cm}{draw, right = 0.8cm of m3b2} (m3b3) {};
  \node[left = 0 of m3b3] {\small $b_3^3$};

  \snode{0.3cm}{\bheight}{0.1cm}{draw, right = 0.6cm of m3b3} (m3b4) {};
  \node[left = 0 of m3b4] {\small $b_3^4$};

  \draw[line width=0.02cm, -stealth,red] (m1b4.south west) -- (m2b4.north west);
  \draw[line width=0.02cm, dotted,red] (m1b4.south east) -- (m2b4.north east);
  
  \draw[line width=0.02cm,-stealth,blue] (m2b4.south west) -- ($(m3b3.north west)+(0.2cm,0)$);
  \draw[line width=0.02cm,blue] ($(m3b3.north west)+(0.2cm,0)$) -- ($(m3b3.south west)+(0.2cm,0)$);
  \draw[line width=0.02cm,dotted,blue] (m2b4.south east) -- ($(m3b3.north west)+(0.6cm,0)$);
  \draw[line width=0.02cm,dotted,blue] ($(m3b3.north west)+(0.6cm,0)$) -- ($(m3b3.south west)+(0.6cm,0)$);

  \node[left = 1cm of m1b1] (txtm1) {\small\color{purple} $m_1'$:};
  \path let \p1 = (txtm1) in let \p2 = (m2b1) in
    node (txtm2) at (\x1, \y2) {\small\color{purple} $m_2'$:};
  \path let \p1 = (txtm1) in let \p2 = (m3b1) in
    node (txtm3) at (\x1, \y2) {\small\color{purple} $m_3'$:};
  \path let \p1 = (txtm1) in let \p2 = (txtm2) in
    node (txtj12) at (\x1, {(\y1+\y2)*0.5}) {\small\color{purple} $j_{12}'$};
  \path let \p1 = (txtm2) in let \p2 = (txtm3) in
    node (txtj23) at (\x1, {(\y1+\y2)*0.5}) {\small\color{purple} $j_{23}'$};

  \fill[gray] ($(m1b2.north west)+(\fsep,-\fsep)$) rectangle ($(m1b2.south east)+(-\fsep,\fsep)$);
  \fill[gray] ($(m2b3.north west)+(0.2cm,-\fsep)$) rectangle ($(m2b3.south west)+(1.2cm,0)$);
  \fill[gray] ($(m3b2.north west)+(0.6cm,-\fsep)$) rectangle ($(m3b2.south west)+(1.6cm,0)$);
\end{tikzpicture}
\caption{After the External Call}
\label{fig:inj-after}
\end{subfigure}
\caption{Constructing of an Interpolating Memory State}
\label{fig:inj}
\end{figure}

We use the concrete example in~\figref{fig:inj} to motivate the
construction of $m_2'$. Here, the white and green
areas correspond to locations in $\perm{\_}{\knonempty}$ (with
at least some permission) and in $\perm{\_}{\kreadable}$ (with
at least readable permission), respectively. Given $\minj{j_{12}}{m_1}{m_2}$,
$\minj{j_{23}}{m_2}{m_3}$ and $\injpacc{(\comp{j_{23}}{j_{12}}, m_1,
  m_3)}{(j_{13}', m_1', m_3')}$, we need to define
$j_{12}'$ and $j_{23}'$ and build $m_2'$ satisfying
$\minj{j_{12}'}{m_1'}{m_2'}$, $\minj{j_{23}'}{m_2'}{m_3'}$,
$\injpacc{(j_{12},m_1,m_2)}{(j_{12}',m_1',m_2')}$, and
$\injpacc{(j_{23},m_2,m_3)}{(j_{23}',m_2',m_3')}$.
%
$m_1'$ and
$m_3'$ are expansions of $m_1$ and $m_3$ with new blocks and possible
modification to the public regions of $m_1$ and $m_3$. 
Here, $m_1'$ has a new block
$b_1^4$ and $m_3'$ has two new block $b_3^3$ and $b_3^4$. 

We first fix $j_{12}'$, $j_{23}'$ and the shape of blocks in
$m_2'$. We begin with $m_2$ and introduce a newly allocated block
$b_2^4$ whose shape matches $b_1^4$ in $m_1'$. Then, $j_{12}'$ is
obtained by expanding $j_{12}$ with identity mapping from $b_1^4$ to
$b_2^4$.  Furthermore, $j_{23}'$ is also expanded with a mapping from
$b_2^4$ to a block in $m_3'$; this mapping is determined by $j_{13}'$.

We then set the values and permissions for memory cells in $m_2'$ so
that it satisfies injection and the \kunchangedon properties for
readable memory regions implied by
$\injpacc{(j_{12},m_1,m_2)}{(j_{12}',m_1',m_2')}$ and
$\injpacc{(j_{23},m_2,m_3)}{(j_{23}',m_2',m_3')}$. The values and
permissions for newly allocated blocks are obviously mapped from
$m_1'$ by $j_{12}'$. Those for old blocks are fixed as follows. By
memory protection provided in $\injpacc{(\comp{j_{23}}{j_{12}}, m_1,
  m_3)}{(j_{13}', m_1', m_3')}$, the only memory cells in $m_1$ that
may have been modified in $m_1'$ are those mapped all the way to $m_3$
by $\comp{j_{23}}{j_{12}}$, while the cells in $m_3$ that may be
modified in $m_3'$ must be in the image of $\comp{j_{23}}{j_{12}}$. To
match this fact, the only old memory regions in $m_2'$ whose values
and permissions may be modified are those both in the image of
$j_{12}$ and the domain of $j_{23}$. Those are the public memory with
respect to $j_{12}$ and $j_{23}$ and displayed as the gray areas
in~\figref{fig:inj-after}.
%
%
Following idea \emph{1)} above, the values and permissions in those
regions are projected from $m_1'$ by applying the injection function
$j_{12}$. Note that there is an exception: values in read-only public regions are copied over from $m_2$.
Following idea \emph{2)}
above, the remaining old memory regions are private with respect to
$j_{12}$ and $j_{23}$ and should have the same values and permissions
as in $m_2$.


Note that the accessibility relations
$\injpacc{(j_{12},m_1,m_2)}{(j_{12}',m_1',m_2')}$ and
$\injpacc{(j_{23},m_2,m_3)}{(j_{23}',m_2',m_3')}$ can be derived from
$\injpacc{(\comp{j_{23}}{j_{12}}, m_1, m_3)}{(j_{13}', m_1', m_3')}$
because the latter enforces \emph{stronger} protection than the
former. This is due to unmapped and out-of-reach regions getting
\emph{bigger} as memory injections get composed. For example,
in~\figref{fig:inj}, $b_1^1$ is mapped by $j_{12}$ but becomes
unmapped by $\comp{j_{23}}{j_{12}}$; the image of $b_2^1$ in $b_3^1$
is in reach by $j_{23}$ but becomes out-of-reach by
$\comp{j_{23}}{j_{12}}$.

\subsubsection{$\screfine{\comp{\kinjp}{\kinjp}}{\kinjp}$}
By definition, we need to prove:
\begin{lemma}\label{lem:injp-comp-refine-injp}
$\screfine{\comp{\kinjp}{\kinjp}}{\kinjp}$ holds. That is,
\begin{tabbing}
  \quad\=\quad\=\quad\=$\exists m_2'\app j_{12}'\app j_{23}',$\=\kill
  \>$\forall j_{13}\app m_1\app m_3,\app \minj{j_{13}}{m_1}{m_3} 
     \imply \exists j_{12}\app j_{23}\app m_2,\app \minj{j_{12}}{m_1}{m_2} \land
     \minj{j_{23}}{m_2}{m_3} \app \land$\\
  \>\>$\forall m_1'\app m_2'\app m_3'\app j_{12}'\app j_{23}',\app
       \injpacc{(j_{12}, m_1, m_2)}{(j_{12}', m_1', m_2')} \imply
       \injpacc{(j_{23}, m_2, m_3)}{(j_{23}', m_2', m_3')} \imply$\\
  \>\>\>$\minj{j_{12}'}{m_1'}{m_2'} \imply 
       \minj{j_{23}'}{m_2'}{m_3'} \imply
       \exists j_{13}',\app \injpacc{(j_{13}, m_1, m_3)}{(j_{13}', m_1', m_3')}
         \land \minj{j_{13}'}{m_1'}{m_3'}.$
\end{tabbing}
\end{lemma}
This lemma conforms to~\figref{fig:trans-comp2}.
To prove it, we pick $j_{12}$ to be an partial identity injection ($j_{12} (b) = \some{b,0}$ when $j_{13}(b) \neq \none$)
, $j_{23} =j_{13}$ and $m_2 = m_1$. Then the lemma is reduced to proving the
existence of $j_{13}'$ that satisfies $\injpacc{(j_{13}, m_1,
  m_3)}{(j_{13}', m_1', m_3')}$ and $\minj{j_{13}'}{m_1'}{m_3'}$. By
picking $j_{13}' = \comp{j_{12}'}{j_{23}'}$, we can easily prove these
properties by exploiting the properties of \kinjp.

\section{Derivation of the Direct Refinement for CompCert}
\label{sec:refinement}

In this section, we discuss the proofs and composition of open
simulations for the compiler passes of CompCert into the direct
refinement $\refinesymb_{\texttt{ac}}$ following the ideas discussed
in~\secref{ssec:key-idea-vcomp}.
CompCert compiles \kwd{Clight} programs into \kwd{Asm} programs
through 19 passes~\cite{compcert}, including several optimization
passes working on the \kwd{RTL} intermediate language.
First, we prove the open simulations for all these passes with
appropriate simulation conventions. In particular, we directly reuse
the proofs of non-optimizing passes in CompCertO and update the
proofs of optimizing passes with semantic invariants.
Second, we prove a collection of properties for refining simulation
conventions in preparation for vertical composition. Those properties
enable absorption of KMRs into \kinjp and composition of semantic
invariants. They rely critically on transitivity of \kinjp.
Finally, we vertically compose the simulations and refine the
incoming and outgoing simulation conventions into a single simulation
convention $\scc$, thereby establishing $\leqslant_\scc$ as the
top-level refinement $\refinesymb_{\texttt{ac}}$.

\begin{table}
  \caption{Significant Passes of CompCert}
  \small
  \begin{tabular}{c c}
    \hline
    \textbf{Languages/Passes} & \textbf{Outgoing $\arrli$ Incoming}\\
    \hline
    \textbf{\kwd{Clight}} & $\cli \arrli \cli$ \\
    \selfsim{\kwd{Self-Sim}} & \selfsim{$\kro \compsymb \kc_{\kinjp} \arrli \kro \compsymb \kc_{\kinjp}$} \\
    \kwd{SimplLocals} & $\kc_{\kinjp} \arrli \kc_{\kinj}$ \\
    \textbf{\kwd{Csharpminor}} & $\cli \arrli \cli$ \\
    \kwd{Cminorgen}   & $\kc_{\kinjp} \arrli \kc_{\kinj}$ \\
    \textbf{\kwd{Cminor}} & $\cli \arrli \cli$ \\
    \kwd{Selection} & $\kwt \compsymb \kc_{\kext} \arrli \kwt \compsymb \kc_{\kext}$ \\
    \textbf{\kwd{CminorSel}} & $\cli \arrli \cli$ \\
    \kwd{RTLgen} & $\kc_{\kext} \arrli \kc_{\kext}$ \\
    \textbf{\kwd{RTL}} & $\cli \arrli \cli$ \\
    \selfsim{\kwd{Self-Sim}} & \selfsim{$\kc_{\kinj} \arrli \kc_{\kinj}$} \\
    \kwd{Tailcall} & $\kc_{\kext} \arrli \kc_{\kext}$ \\
    \kwd{Inlining} & $\kc_{\kinjp} \arrli \kc_{\kinj}$ \\
    \selfsim{\kwd{Self-Sim}} & \selfsim{$\kc_{\kinjp} \arrli \kc_{\kinjp}$} \\
     \end{tabular}
     \begin{tabular}{c c}
     \hline
    \textbf{Language/Pass} & \textbf{Outgoing $\arrli$ Incoming}\\
       \hline
    \textcolor{red}{\kwd{Constprop}} & 
      \textcolor{red}{$\kro \compsymb \kc_{\kinjp} \arrli \kro \compsymb \kc_{\kinjp}$} \\
    \textcolor{red}{\kwd{CSE}} & \textcolor{red}{$\kro \compsymb \kc_{\kinjp} \arrli \kro \compsymb \kc_{\kinjp}$} \\
    \textcolor{red}{\kwd{Deadcode}} & \textcolor{red}{$\kro \compsymb \kc_{\kinjp} \arrli \kro \compsymb \kc_{\kinjp}$} \\
   
    \textcolor{red}{\kwd{Unusedglob}}  & \textcolor{red}{$\kc_{\kinj} \arrli \kc_{\kinj}$} \\
    \kwd{Allocation} & $\kwt \compsymb \kc_{\kext} \compsymb \kcl \arrli \kwt \compsymb \kc_{\kext} \compsymb \kcl $ \\
    \textbf{\kwd{LTL}} & $\ltlli \arrli \ltlli$ \\
    \kwd{Tunneling} & $\kwd{ltl}_{\kext} \arrli \kwd{ltl}_{\kext}$ \\
    \textbf{\kwd{Linear}} & $\ltlli \arrli \ltlli$ \\
    \kwd{Stacking} & $\kwd{ltl}_{\kinjp} \compsymb \klm \arrli \klm \compsymb \kwd{mach}_{\kinj}$ \\
    \textbf{\kwd{Mach}} & $\machli \arrli \machli$ \\
    \kwd{Asmgen} & $\kmach_{\kext} \compsymb \kma \arrli \kmach_{\kext} \compsymb \kma$\\
    \textbf{\kwd{Asm}} & $\asmli \arrli \asmli$ \\
    \selfsim{\kwd{Self-Sim}} & \selfsim{$\kasm_{\kinj} \arrli \kasm_{\kinj}$} \\
    \selfsim{\kwd{Self-Sim}} & \selfsim{$\kasm_{\kinjp} \arrli \kasm_{\kinjp}$}
  \end{tabular}
  \label{tab:compcerto}
\end{table}

\subsection{Open Simulation of Individual Passes}
\label{ssec:single-pass}
We list the compiler passes and their simulation types
in~\tabref{tab:compcerto} (passes on the right follow the passes
on the left) together with their source and target languages and
interfaces (in bold fonts).
The passes in black are reused from CompCertO, while those in red are
reproved optimizing passes.  
The passes in blue are \emph{self-simulating} passes we inserted; they
will be used in~\secref{ssec:gen-direct-refinement} for refining composed
simulation conventions.
Note that we have omitted passes with the identity simulation
convention (i.e., simulations of the form $\osims{\kid}{L_1}{L_2}$)
in~\tabref{tab:compcerto} as they do not affect the proofs.
\footnote{The omitted passes are \kwd{Cshmgen}, \kwd{Renumber}, \kwd{Linearize},
  \kwd{CleanupLabels} and \kwd{Debugvar}.}

\subsubsection{Simulation Conventions and Semantic Invariants}
%
We first introduce relevant simulation conventions and semantic
invariants shown in~\tabref{tab:compcerto}.
The simulation conventions $\kc_K: \sctype{\cli}{\cli}$, $\kltl_K:
\sctype{\ltlli}{\ltlli}$, $\kmach_K: \sctype{\machli}{\machli}$, and
$\kasm_K: \sctype{\asmli}{\asmli}$ relate the same language interfaces
with queries and replies native to the associated intermediate
languages.
They are parameterized by a KMR $K$ to allow different compiler passes
to have different assumptions on memory evolution. Conceptually, this
parameterization is unnecessary as we can simply use \kinjp for every
pass due to its uniformity (as discussed
in~\secref{ssec:injp-uniform}). Nevertheless, it is useful because the
compiler proofs become simpler and more natural with the least 
restrictive KMRs which may be weaker than \kinjp. CompCertO defines
several KMRs weaker than \kinjp: $\kid$ is used when memory is unchanged;
$\kext$ is used when the source
and target memory share the same structure;
$\kinj$ is a simplified version of $\kinjp$ without its memory
protection.
The simulation conventions $\kcl:
\sctype{\cli}{\ltlli}$, $\klm: \sctype{\ltlli}{\machli}$ and $\kma:
\sctype{\machli}{\asmli}$ capture the calling convention of CompCert:
$\kcl$ relates C-level queries and replies to those in the \kwd{LTL}
language where the arguments are distributed to abstract stack slots;
$\klm$ further relates abstract stack slots with states on an
architecture independent machine; $\kma$ relates this state to
registers and memory in the assembly language (X86 assembly in our
case).
As discussed before, some refinements rely on invariants on the source
semantics. The semantic invariant $\kwt$ enforces that arguments and
return values of function calls respect function signatures. $\kro$ is
critical for ensuring the correctness of optimizations, which will be
discussed next.

\subsubsection{Open Simulation of Optimizations}

\begin{wrapfigure}{R}{.51\textwidth}
\begin{subfigure}{0.25\textwidth}
\centering
\begin{lstlisting}[language = C]
const int key = 42;
void foo(int*);
int double_key() {
  int a = key;
  foo(&key);
  return a + key;
}
\end{lstlisting}
\caption{Source Program}
\label{fig:va-src}
\end{subfigure}
\begin{subfigure}{0.25\textwidth}
\centering
\begin{lstlisting}[language = C]
const int key = 42;
void foo(int*);
int double_key() {
  int a = 42;
  foo(&key);
  return 84;
}
\end{lstlisting}  
\caption{Target Program}
\label{fig:va-tgt}
\end{subfigure}
\caption{An Example of Constant Propagation}
\label{fig:va-example}
\end{wrapfigure}
%
The optimizing passes \kwd{Constprop}, \kwd{CSE} and
\kwd{Deadcode} perform constant propagation, common subexpression
elimination and dead code elimination, respectively.
They make use of a static value analysis algorithm for collecting
information of variables during the execution. 
For each function, this algorithm starts with the
known initial values of read-only (constant) global variables. It
simulates the function execution to analyze the values of global or
local variables after executing each instruction.
In particular, for global \emph{constant} variables, their references at any
point should have the initial values of constants.
For local variables stored on the stack, their references may have
initial values or may not if interfered by other function calls. When
the analysis encounters a call to another function, it checks whether
the address of current stack frame is leaked to the callee directly
through arguments or indirectly through pointers in memory. If not,
then the stack frame is considered \emph{unreachable} from its
callee. Consequently, the references to local variables on unreachable stack frames
after function calls remain to be their initial values.
Based on this analysis, the three passes then identify and perform
optimizations. 

Most of the proofs of closed simulations for those passes can be
adapted to open simulation straightforwardly. The only and main
difficulty is to prove that information derived from static analysis
is consistent with the dynamic memory states in incoming queries and
after external calls return.
 %
%
We introduce the semantic invariant \kro and combine it with \kinjp to
ensure this consistency. The above optimization passes all
use $\kro \compsymb \kc_{\kinjp}$ as their simulation conventions (because
\kwd{RTL} conforms to the C interface).
The adaptation of optimization proofs for those passes is
similar. As an example, we only discuss constant propagation whose
correctness theorem is stated as follows:
\begin{lemma}
  $ \forall (M\; M':{\kwd{RTL}}), \kwd{Constprop}(M) = M' \imply    
    \osims{\kro \compsymb \kc_\kinjp}{\sem{M}}{\sem{M'}}.$
\end{lemma}
%

Instead of presenting its proof, we illustrate how $\kro$ and $\kinjp$
help establish the open simulation for \kwd{Constprop} through a
concrete example as depicted in~\figref{fig:va-example}. This example
covers optimization for both global constants (e.g., \code{key}) and local
variables (e.g., \code{a}).
By static analysis of~\figref{fig:va-src}, \emph{1) }$\kwd{key}$
contains $42$ at line 4 because \kwd{key} is a constant global
variable and, \emph{2) } both $\kwd{key}$ and $\kwd{a}$ contain $42$
after the external call to \kwd{foo} returns to line 6. Here, the
analysis confirms \kwd{key} has the value $42$ because \kwd{foo} (if
well-behaved) will not modify a constant global variable. Furthermore,
\kwd{a} has the value $42$ because it resides in the stack frame of
\code{double\_key} which is unreachable from \code{foo} (in fact, \kwd{a} is the only variable in the frame). 
%
As a result, the source program is optimized into~\figref{fig:va-tgt}. 

We first show that \kro guarantees the dynamic values of global
constants are consistent with static analysis. That is, global
variables are correct in incoming memory and are protected during
external calls.
\kro is defined as follows:
\begin{definition}\label{def:ro}
  $\kro : \sctype{\cli}{\cli} = \simconv{W_{\kro}}{\scname{R}_\kro^q}{\scname{R}_\kro^r}$ where
  $W_\kro = (\kwd{symtbl} \times \kmem)$ and
  \begin{tabbing}
    \quad\quad\=\kill
    \>$\scname{R}_\kro^q (\mathit{se},m) =
    \pset{(\cquery{v_f}{\sig}{\vec{v}}{m},\cquery{v_f}{\sig}{\vec{v}}{m})}{\kwd{ro-valid}(se,m)}$ \\
    \>$\scname{R}_\kro^r (\mathit{se},m) =
    \pset{(\creply{\mathit{res}}{m'},\creply{\mathit{res}}{m'})}{\macc{m}{m'}} $
  \end{tabbing}
\end{definition}
\noindent
Note that although \kro takes the form of a simulation convention, it only
relates the same queries and replies, i.e., it only enforces
invariants on the source side.
This kind of simulation conventions are what we called \emph{semantic
invariants}. A symbol table $se$ (of type $\kwd{symtbl}$) is provided
together with memory, so that the semantics can locate memory blocks
of global definitions and find the initial values of global variables.
$\kwd{ro-valid}(se,m)$ states that the values of global constant
variables in the incoming memory $m$ are the same as their initial
values. Therefore, the optimization of $\kwd{key}$ into $42$ at line 4
of~\figref{fig:va-src} is correct.
For the external call to $\kwd{foo}$, monotonicity $\macc{m}{m'}$ ensures
that read-only values in memory are unchanged, therefore the above
property is preserved from external queries to replies (i.e.,
$\rovalid{se}{m} \imply \macc{m}{m'} \imply
\rovalid{se}{m'}$). As a result, replacing $\kwd{key}$ with $42$ at line
6 makes sense.

\begin{figure}
  \def\bheight{0.5cm}
  \begin{subfigure}[b]{0.2\textwidth}
    \centering
  \begin{tikzpicture}
      
  \hblock{\bheight}{draw, minimum width = 0.5cm} (bkey) {\small $42$};
  \node[above = 0 of bkey] (bkeytxt) {\small $b_{\texttt{key}}$};

  \hblock{\bheight}{draw, right = 0.3 of bkey} (ba) {\small $42$};
  \path let \p1 = (ba) in let \p2 = (bkeytxt) in
  node (batxt) at (\x1, \y2) {\small $b_{\texttt{a}}$};
  
  \hblock{\bheight}{draw, minimum width = 0.5cm,
                    below = 0.5 of bkey} (tbkey) {\small $42$};

  \hblock{\bheight}{draw, minimum width = 0.5cm,
                    below = 0.5 of ba} (tba) {\small $42$};
  
  \node[left = 0.1 of bkey] (m1) {\small$m_1:$};
  \node[left = 0.1 of tbkey] (m2) {\small$m_2:$};
  
  \draw[-stealth] (ba.south west) -- (tba.north west);
  \draw[-stealth] (bkey.south west) -- node[left] {$j$} (tbkey.north west);
  
\end{tikzpicture}


  \caption{Before \kwd{foo}}
  \label{fig:before-foo}
\end{subfigure}
\begin{subfigure}[b]{0.2\textwidth}
    \centering
    \begin{tikzpicture}
  \hblock{\bheight}{draw, minimum width = 0.5cm} (bkey) {\small $42$};
  \node[above = 0 of bkey] (bkeytxt) {\small $b_{\texttt{key}}$};


  \hblock{\bheight}{draw, right = 0.3 of bkey,pattern = north east lines} (ba) {\small $42$};
  \path let \p1 = (ba) in let \p2 = (bkeytxt) in
  node (batxt) at (\x1, \y2) {\small $b_{\texttt{a}}$};
  
  \hblock{\bheight}{draw, minimum width = 0.5cm,
                    below = 0.5 of bkey} (tbkey) {\small $42$};

  \hblock{\bheight}{draw, minimum width = 0.5cm,pattern = north east lines,
                    below = 0.5 of ba} (tba) {\small $42$};
  
  \draw[-stealth] (bkey.south west) -- node[left]{$j_1$} (tbkey.north west);

  \node[left = 0.1 of bkey] (m1) {\small$m_1:$};
  \node[left = 0.1 of tbkey] (m2) {\small$m_2:$};
  
\end{tikzpicture}
\caption{Start of \kwd{foo}}
  \label{fig:begin-foo}
\end{subfigure}
\begin{subfigure}[b]{0.25\textwidth}
    \centering
    \begin{tikzpicture}
  \hblock{\bheight}{draw, minimum width = 0.5cm} (bkey) {\small $42$};
  \node[above = 0 of bkey] (bkeytxt) {\small $b_{\texttt{key}}$};


  \hblock{\bheight}{draw, right = 0.3 of bkey,pattern = north east lines} (ba) {\small $42$};
  \path let \p1 = (ba) in let \p2 = (bkeytxt) in
  node (batxt) at (\x1, \y2) {\small $b_{\texttt{a}}$};
  
  \hblock{\bheight}{draw, minimum width = 0.5cm,
                    below = 0.5 of bkey} (tbkey) {\small $42$};

  \hblock{\bheight}{draw, minimum width = 0.5cm,pattern = north east lines,
                    below = 0.5 of ba} (tba) {\small $42$};
                  
  \hblock{\bheight}{draw, right = 0.3 of ba,minimum width = 0.5cm} (b) {};
  \path let \p1 = (b) in let \p2 = (batxt) in
  node (btxt) at (\x1, \y2) {\small $b$};

  \hblock{\bheight}{draw, minimum width = 0.5cm,
                    below = 0.5 of b} (tb) {};
  
  \draw[-stealth] (bkey.south west) -- node[left] {$j_2$} (tbkey.north west);
  \draw[-stealth] (b.south west) -- (tb.north west);

  \node[left = 0.1 of bkey] (m1) {\small$m_1':$};
  \node[left = 0.1 of tbkey] (m2) {\small$m_2':$};
  
\end{tikzpicture}

 \caption{End of \kwd{foo}}
    
 \label{fig:end-foo}
\end{subfigure}
\begin{subfigure}[b]{0.25\textwidth}
    \centering
    \begin{tikzpicture}
  \hblock{\bheight}{draw, minimum width = 0.5cm} (bkey) {\small $42$};
  \node[above = 0 of bkey] (bkeytxt) {\small $b_{\texttt{key}}$};


  \hblock{\bheight}{draw, right = 0.3 of bkey} (ba) {\small $42$};
  \path let \p1 = (ba) in let \p2 = (bkeytxt) in
  node (batxt) at (\x1, \y2) {\small $b_{\texttt{a}}$};
  
  \hblock{\bheight}{draw, minimum width = 0.5cm,
                    below = 0.5 of bkey} (tbkey) {\small $42$};

  \hblock{\bheight}{draw, minimum width = 0.5cm,
                    below = 0.5 of ba} (tba) {\small $42$};

  \hblock{\bheight}{draw, right = 0.3 of ba,minimum width = 0.5cm} (b) {};
  \path let \p1 = (b) in let \p2 = (batxt) in
  node (btxt) at (\x1, \y2) {\small $b$};

  \hblock{\bheight}{draw, minimum width = 0.5cm,
                    below = 0.5 of b} (tb) {};

  \draw[-stealth] (ba.south west) -- (tba.north west);
  \draw[-stealth] (bkey.south west) -- node[left] {$j_3$} (tbkey.north west);
  \draw[-stealth] (b.south west) -- (tb.north west);

  \node[left = 0.1 of bkey] (m1) {\small$m_1':$};
  \node[left = 0.1 of tbkey] (m2) {\small$m_2':$};
\end{tikzpicture}


  \caption{After \kwd{foo}}
  \label{fig:after-foo}
\end{subfigure}
  \caption{Memory Injections from Call to Return of \kwd{foo}}
  \label{fig:meminj-foo}
\end{figure}
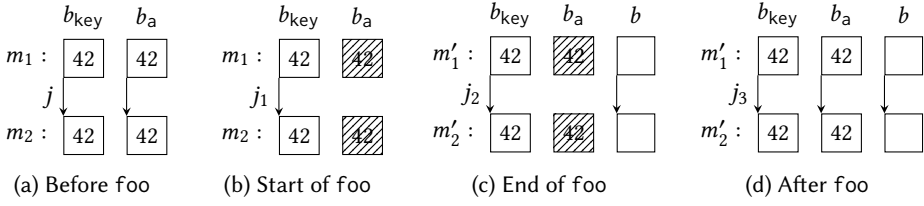

We then show that \kinjp guarantees the dynamic values of unreachable
local variables are consistent with static analysis. That is,
unreachable stack values are unchanged by external calls.
This protection is realized by \kinjp with \emph{shrinking} memory
injections.
\figref{fig:meminj-foo} shows the protection of
\code{a} when calling \code{foo}.
Before the external call to $\kwd{foo}$, the source blocks $b_\kwd{a}$
and $b_\kwd{key}$ are \emph{mapped} to target blocks by the current
injection $j$.
The analysis determines that the argument and 
memory passed to $\kwd{foo}$ do not contain any pointer to $b_{\kwd{a}}$. Therefore, we can simply remove
$b_{\kwd{a}}$ from $j$ to get a shrunk yet valid memory
injection $j_1$.
Then, $b_\kwd{a}$ is protected during the call to \code{foo}.
%
$b_{\kwd{a}}$ is added back to the injection after \code{foo} returns 
and the simulation continues.
%
%

Finally, \kwd{Unusedglob} which removes unused static global variables
is verified by assuming that global symbols remain the same throughout
the compilation and with a weaker KMR \kinj.

\subsection{Properties for Refining Simulation Conventions}
\label{ssec:refine-properties}
We present properties necessary for refining the composed simulation conventions
in~\tabref{tab:compcerto}.



\subsubsection{Commutativity of KMRs and Structural Conventions.}

\begin{lemma}\label{lem:ca-com}
  For $\kwd{Z} \in \{\kcl ,\klm ,\kma \}$ and $K \in \{\kext ,\kinj ,\kinjp \}$
  we have $\screfine{\kwd{X}_{K} \compsymb \kwd{Z}} {\kwd{Z} \compsymb \kwd{Y}_K }$.
\end{lemma}
This lemma is provided by CompCertO~\cite{compcerto}. \kwd{X} and
\kwd{Y} denote the simulation conventions for the source and target
languages of \kwd{Z}, respectively (e.g., $\kwd{X} = \kc$ and
$\kwd{Y} = \kltl$ when $\kwd{Z} = \kcl$). If $K = \kinjp$ we get
$\screfine{\kc_\kinjp \compsymb \kcl}{\kcl \compsymb \kltl_\kinjp}$.
This lemma indicates at the
outgoing (incoming) side a convention lower (higher) than $\kcl$,
$\klm$ and $\kma$ may be lifted over them to a higher position (pushed
down to a lower position).
%


\subsubsection{Absorption of KMRs into \kinjp} 
%


The lemma below is needed for absorbing KMRs into \kinjp:

\begin{lemma}\label{lem:sim-refine}
  For any $\scname{R}$,
  $(1) \scname{R}_\kinjp \compsymb \scname{R}_\kinjp \equiv \scname{R}_\kinjp$
  $\; (2)\screfine{\scname{R}_\kinjp}{\scname{R}_\kinj}$
  $\; (3)\screfine{\scname{R}_\kinjp \compsymb \scname{R}_\kinj \compsymb \scname{R}_\kinjp}{\scname{R}_\kinjp}$
  \begin{tabbing}
    \quad\=\kill
    \>$(4) \screfine{\scname{R}_\kinj \compsymb \scname{R}_\kinj}{\scname{R}_\kinj}$
      \;$(5) \scname{R}_\kext \compsymb \scname{R}_\kinj \equiv \scname{R}_\kinj$
      \;$(6) \scname{R}_\kinj \compsymb \scname{R}_\kext \equiv \scname{R}_\kinj$
      \;$(7) \scname{R}_\kext \compsymb \scname{R}_\kext \equiv \scname{R}_\kext$.
  \end{tabbing}
\end{lemma}
The simulation convention $\scname{R}$ is parameterized over a KMR.
Property (1) is a direct consequence of $\kinjp \compsymb \kinjp
\equiv \kinjp$, which is critical for merging simulations using
\kinjp.
The remaining ones either depend on transitivity of \kinjp, or
trivially hold as shown by~\citet{compcerto}.

\subsubsection{Composition of Semantic Invariants} 
\label{sssec:sinv-comp}
Lastly, we also need to handle the
two semantic invariants \kro and \kwt. They cannot be absorbed into
\kinjp because their assumptions are fundamentally
different. Therefore, our goal is to permute them to the top-level and
merge any duplicated copies. 
The following lemmas enable elimination and permutation of \kwt:
\begin{lemma}\label{lem:wt}
  For any\; $\scname{R}_K:\sctype{\cli}{\cli}$, we have
  $(1) \scequiv{\scname{R}_K \compsymb \kwt}{\kwt \compsymb \scname{R}_K \compsymb \kwt}$ and\;
  $(2) \scequiv{\scname{R}_K \compsymb \kwt}{\kwt \compsymb \scname{R}_K}$.
\end{lemma}
%

%
%

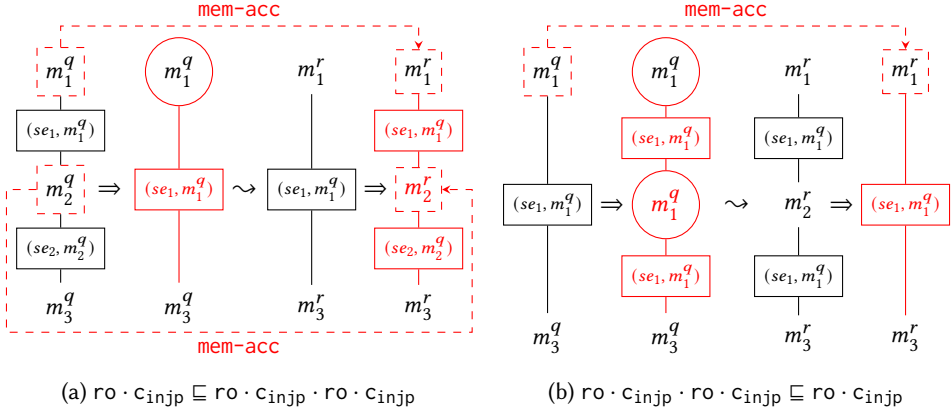
\begin{figure}
  \begin{subfigure}[b]{0.46\textwidth}
  \centering
  \begin{tikzpicture}
    \node [draw, dashed, red](m1) {\small\color{black} $m_1^q$};
    \node [below = 0.9cm of m1, draw, dashed, red] (m2) 
          {\small\color{black} $m_2^q$};
    \node [below = 0.9cm of m2] (m3) {\small $m_3^q$};
    \path (m1) -- node[draw] (w1) {\tiny $(se_1,m_1^q)$} (m2);
    \path (m2) -- node[draw] (w2) {\tiny $(se_2,m_2^q)$} (m3);
    \draw (m1) -- (w1) -- (m2);
    \draw (m2) -- (w2) -- (m3);

    \node [right = 0.8cm of m1, circle, draw, red] (m11) 
          {\small\color{black} $m_1^q$};
    \path let \p1 = (m11) in let \p2 = (m3) in 
      node (m33) at (\x1, \y2) {\small $m_3^q$};
    \path let \p1 = (m11) in let \p2 = (m2) in
      node[draw,red] (w) at (\x1,\y2) {\tiny $(se_1,m_1^q)$};
    \draw[red] (m11) -- (w) -- (m33);

    \node [right = 1cm of m11] (m11p) {\small $m_1^r$};
    \path let \p1 = (m11p) in let \p2 = (m33) in 
      node (m33p) at (\x1, \y2) {\small $m_3^r$};
    \path let \p1 = (m11p) in let \p2 = (m2) in
      node[draw] (wp) at (\x1, \y2) {\tiny $(se_1,m_1^q)$};
    \draw (m11p) -- (wp) -- (m33p);

    \node [right = 0.8cm of m11p, draw, dashed, red] (m1p) 
          {\small\color{black} $m_1^r$};
    \path let \p1 = (m1p) in let \p2 = (m2) in 
      node[red, draw, dashed] (m2p) at (\x1, \y2) {\small $m_2^r$};
    \path let \p1 = (m1p) in let \p2 = (m33) in 
      node (m3p) at (\x1, \y2) {\small $m_3^r$};
    \path (m1p) -- node[draw,red] (w1p) {\tiny $(se_1,m_1^q)$} (m2p);
    \path (m2p) -- node[draw,red] (w2p) {\tiny $(se_2,m_2^q)$} (m3p);

    \draw[red] (m1p) -- (w1p) -- (m2p);
    \draw[red] (m2p) -- (w2p) -- (m3p);

     \path (m2) -- node {\small $\imply$} (w);
     \path (w) -- node {\small $\accsymb$} (wp);
     \path (wp) -- node {\small $\imply$} (m2p);

    \draw[-stealth, red,dashed] (m1) -- ++(0, 0.6cm) -| 
        node[above, pos=0.25] {\small\code{mem-acc}} (m1p);
    \path[draw, -stealth, red, dashed] 
    let \p1 = (m2.west) in let \p2 = (m3.south west) in let \p3 = (m2p.east) in
        (m2) -- ({\x1-0.4cm}, \y1) -- ({\x1-0.4cm},\y2) 
        -- node[below, pos=0.5] {\small\code{mem-acc}} ({\x3+0.4cm}, \y2) 
        |- (m2p);



  \end{tikzpicture}
  \caption{$\screfine{\kro \cdot \kc_\kinjp}{{\kro \cdot \kc_\kinjp}\cdot{\kro \cdot\kc_\kinjp}}$}
  \label{fig:trans-ro1}
  \end{subfigure}
  \begin{subfigure}[b]{0.46\textwidth}
  \centering
  \begin{tikzpicture}
    \node [draw, dashed, red](m1) {\small\color{black} $m_1^q$};
    \node [draw,below = 1.15cm of m1] (w) {\tiny $(se_1,m_1^q)$};
    \node [below = 1.15cm of w] (m3) {\small $m_3^q$};
    
    \node [right = 0.8cm of m1, circle, red, draw] (m11) 
          {\small\color{black} $m_1^q$};
    \path let \p1 = (m11) in let \p2 = (w) in 
      node[draw, circle, red] (m2) at (\x1, \y2) {\small $m_1^q$};
    \path let \p1 = (m11) in let \p2 = (m3) in 
      node (m33) at (\x1, \y2) {\small $m_3^q$};
    \path (m11) -- node[draw,red] (w1) {\tiny $(se_1,m_1^q)$} (m2);
    \path (m2) -- node[draw,red] (w2) {\tiny $(se_1,m_1^q)$} (m33);

    \draw (m1) -- (w) -- (m3);
    \draw[red] (m11)  -- (w1) -- (m2);
    \draw[red] (m2)  -- (w2) -- (m33);

    \node [right = 1cm of m11] (m1p) {\small $m_1^r$};
    \path let \p1 = (m1p) in let \p2 = (m2) in 
      node (m2p) at (\x1, \y2) {\small $m_2^r$};
    \path let \p1 = (m1p) in let \p2 = (m33) in 
      node (m3p) at (\x1, \y2) {\small $m_3^r$};
    \path (m1p) -- node[draw] (w1p) {\tiny $(se_1,m_1^q)$} (m2p);
    \path let \p1 = (m1p) in let \p2 = (w2) in
      node[draw] (w2p) at (\x1,\y2) {\tiny $(se_1,m_1^q)$};

    \draw (m1p) -- (w1p) -- (m2p);
    \draw (m2p) -- (w2p) -- (m3p);

    \node [right = 0.8cm of m1p, draw, red, dashed] (m11p) 
          {\small\color{black} $m_1^r$};
    \path let \p1 = (m11p) in let \p2 = (m3p) in 
      node (m33p) at (\x1, \y2) {\small $m_3^r$};
    \path (m11p) -- node[draw,red] (w3p) {\tiny $(se_1,m_1^q)$} (m33p);
    
    \draw[red] (m11p) -- (w3p) -- (m33p);

    \path (w) -- node {\small $\imply$} (m2);
    \path (m2) -- node {\small $\accsymb$} (m2p);
    \path (m2p) -- node {\small $\imply$} (w3p);


    \draw[-stealth, red,dashed] (m1) -- ++(0, 0.6cm) -| 
        node[above, pos=0.25] {\small\code{mem-acc}} (m11p);

  \end{tikzpicture}
  \caption{$\screfine{{\kro \cdot \kc_\kinjp} \cdot {\kro \cdot \kc_\kinjp}}{\kro \cdot \kc_\kinjp}$}
  \label{fig:trans-ro2}
  \end{subfigure}
  \caption{Transitivity of $\kro \cdot \kc_\kinjp$}
  \label{fig:trans-ro}
\end{figure}

\kro is more difficult to handle as it does not commute with arbitrary
simulation conventions. To eliminate redundant \kro, we piggyback \kro
onto \kinjp and prove the following transitivity property:
\begin{lemma}\label{lem:roinjp-trans}
  $\scequiv{\kro \compsymb \kc_\kinjp}{ \kro \compsymb \kc_\kinjp \compsymb \kro \compsymb \kc_\kinjp}$
\end{lemma}
Its proof follows the same steps for proving
$\scequiv{\kc_\kinjp}{\kc_\kinjp \compsymb \kc_\kinjp}$ with
additional reasoning for establishing properties of \kro. A
graphic presentation of the proof is given
in~\figref{fig:trans-ro} which mirrors~\figref{fig:trans-comp}. We
focus on explaining the additional reasoning and have omitted the
$\accsymb_{\kinjp}$ relations and the worlds for \kinjp
in~\figref{fig:trans-ro}. Note that by definition the worlds $(se, m)$
for \kro do not evolve like those for \kinjp.
A red circle around a memory state $m$ indicates it is required to
prove $\kwd{ro-valid}$ in $\scname{R}_\kro^q$ holds for $m$.  The
\kwd{mem-acc} relations over dashed arrows are the properties over
replies in $\scname{R}_\kro^r$ and must also be verified.

The above additional properties are proved based on two observations.
First, the properties for queries (i.e., $\kwd{ro-valid}$) are propagated
in refinement along with copying of memory states.
For example, to prove the refinement in~\figref{fig:trans-ro1}, we are given
$\rovalid{se_1}{m_1^q}$ and $\rovalid{se_2}{m_2^q}$ according to
the initial $\scname{R}_\kro^q$ relations. By choosing $(se_1, m_1^q)$ to
be the world for the composed $\scname{R}_\kro^q$, $\rovalid{se_1}{m_1^q}$
holds trivially for $m_1^q$ in the circle. To prove the refinement
in~\figref{fig:trans-ro2}, we need to prove that the interpolating memory state
after the initial decomposition satisfies $\scname{R}_\kro^q$. By
choosing $m_1^q$ 
to be this state (in the middle circle in~\figref{fig:trans-ro2} and
 according to the proof of~\lemref{lem:injp-comp-refine-injp}), $\rovalid{se_1}{m_1^q}$ 
follows directly from the initial assumption. Second, the properties for replies
(i.e., \kwd{mem-acc}) have already been encoded into $\accsymb_\kinjp$
by~\defref{def:injp}. For example, $m_2^r$
in~\figref{fig:trans-ro1} is constructed by following exactly~\lemref{lem:injp-refine-injp-comp}. Therefore,
$\macc{m_2^q}{m_2^r}$ trivially holds.

Finally, at the top level, we need \kro and \kwt to commute
which is straightforward to prove:

\begin{lemma}\label{lem:inv-comm}
  $\scequiv {\kro \compsymb \kwt} {\kwt \compsymb \kro}$
\end{lemma}

\subsection{Proving the Direct Open Simulation for CompCert}
\label{ssec:gen-direct-refinement}

We first insert self-simulations into the compiler passes, as shown
in~\tabref{tab:compcerto}. This is to supply extra $\scname{R}_\kinj$,
$\scname{R}_\kinjp$, and \kro for absorbing $\scname{R}_\kext$
($\scname{R}_\kinj$) into $\scname{R}_\kinj$ ($\scname{R}_\kinjp$) by
properties in~\lemref{lem:sim-refine} and for transitive composition of
\kro. Self-simulations are obtained by the following lemma:
\begin{theorem}\label{thm:self-sim}
  If $p$ is a program written in $\kwd{Clight}$ or $\kwd{RTL}$ and
  $\scname{R} \in \{\kro, \kc_{\kext}, \kc_{\kinj}, \kc_{\kinjp}\}$, or $p$
  is written in $\kwd{Asm}$ and $\scname{R} \in \{\kasm_{\kext},
  \kasm_{\kinj}, \kasm_{\kinjp}\}$, then
  $\osim{\scname{R}}{\scname{R}}{\sem{p}}{\sem{p}}$ holds.
\end{theorem}
We unify the conventions at the incoming and
outgoing sides.
We start with the simulation
$\osim {\scname{R}}{\scname{S}}{L_1}{L_2}$
which is the transitive composition of compiler passes
in~\tabref{tab:compcerto} where
\begin{tabbing}
  \quad\=$\scname{R}$ \== \=\kill
  \>$\scname{R}$\>=\>$\kro \compsymb  \kc_{\kinjp} \compsymb \kc_{\kinjp} \compsymb \kc_{\kinjp} \compsymb \kwt \compsymb \kc_{\kext} \compsymb \kc_{\kext} \compsymb \kc_{\kinj} \compsymb \kc_\kext \compsymb \kc_{\kinjp} \compsymb \kc_{\kinjp} \compsymb \kro \compsymb \kc_\kinjp \compsymb \kro \compsymb \kc_\kinjp$ \\

  \>\>\>$\compsymb \kro \compsymb \kc_\kinjp \compsymb
  \kc_{\kinj} \compsymb \kwt \compsymb \kc_{\kext} \compsymb \kcl
  \compsymb \kltl_{\kext} \compsymb \kltl_{\kinjp} \compsymb \klm \compsymb \kmach_{\kext} \compsymb \kma \compsymb \kasm_{\kinj} \compsymb \kasm_{\kinjp}$\\
  \>$\scname{S}$\>=\>$\kro \compsymb \kc_{\kinjp} \compsymb \kc_{\kinj} \compsymb \kc_{\kinj} \compsymb \kwt \compsymb \kc_{\kext} \compsymb \kc_{\kext} \compsymb \kc_{\kinj} \compsymb \kc_\kext \compsymb \kc_{\kinj} \compsymb \kc_{\kinjp} \compsymb \kro \compsymb \kc_\kinjp \compsymb \kro \compsymb \kc_\kinjp$ \\
  \>\>\>$\compsymb \kro \compsymb \kc_\kinjp \compsymb \kc_{\kinj} \compsymb \kwt \compsymb \kc_{\kext} \compsymb \kcl
  \compsymb \kltl_{\kext} \compsymb \klm \compsymb \kmach_{\kinj} \compsymb \kmach_{\kext} \compsymb \kma \compsymb \kasm_{\kinj} \compsymb \kasm_{\kinjp}$.
\end{tabbing}
We then find two sequences of refinements
$\sccompcerto \sqsubseteq \scname{R}_n \sqsubseteq \ldots \sqsubseteq \scname{R}_1 \sqsubseteq \scname{R}$
and $\scname{S} \sqsubseteq \scname{S}_1 \sqsubseteq \ldots
\sqsubseteq \scname{S}_m \sqsubseteq \sccompcerto$, by which and
\thmref{thm:sim-refine} we get the simulation $\osim
{\sccompcerto} {\sccompcerto} {L_1}{L_2}$. 
The direct simulation convention is $\sccompcerto = \kro \compsymb \kwt \compsymb
\kcainjp \compsymb \kasm_{\kinjp}$.  $\kro$ enables optimizations at C
level while $\kwt$ ensures well-typedness. The definition of \kcainjp has already been discussed informally
in~\secref{ssec:key-idea-direct-refinement}; its formal definition is given as follows.
Note that, to simplify the presentation, we have omitted
minor constraints such as function values should not be undefined,
stack pointers must have a pointer type, etc. Interested readers
should consult the our artifact for a complete definition.

\begin{definition}\label{def:cainjp}
  $\kcainjp : \sctype{\cli}{\asmli}=
  \simconv{W_\kcainjp}
          {\scname{R}_\kcainjp^q}
          {\scname{R}_\kcainjp^r}$
  where $W_\kcainjp = (W_{\kinjp} \times \ksig \times \kregset)$ and
  ${\scname{R}_\kcainjp^q:\krtype{W_\kcainjp}{\cli^q}{\asmli^q}}$ and 
  ${\scname{R}_\kcainjp^r:\krtype{W_\kcainjp}{\cli^r}{\asmli^r}}$ are defined as:

  \begin{itemize}
  \item $(\cquery{v_f}{\sig}{\vec{v}}{m_1}, \asmquery{\regset}{m_2})
    \in \scname{R}_\kcainjp^q((j,m_1,m_2),\sig,\regset)$ if 

    \begin{tabbing}
      \quad\=\kill
      \> (1)  $\minj{j}{m_1}{m_2},\quad  \vinj{j}{v_f}{\regset(\pcreg)} \quad \vinj{j}{\vec{v}}{\kwd{get-args}(\sig, \regset(\spreg), \regset, m_2)}$\\
      \> (2)  $\kwd{outgoing-arguments}(\sig,\regset(\spreg)) \subseteq \outofreach{j}{m_1}$\\
      \> (3)  $\kwd{outgoing-arguments}(\sig,\regset(\spreg)) \subseteq \perm{m_2}{\kfreeable}$
    \end{tabbing}

    ${\kwd{get-args}(\sig, \regset(\spreg), \regset, m_2)}$ is a
    list of values for arguments at the assembly level obtained by
    inspecting locations for arguments in $\regset$ and $m_2$
    corresponding to the signature $\sig$ which are determined by
    CompCert's calling convention.
    $\kwd{outgoing-arguments}\allowbreak (\sig,\regset(\spreg))$ is a set of
    addresses on the stack frame for outgoing function arguments
    computed from the given signature $\sig$ and the value of stack
    pointer.

  \item 
    $(\creply{r}{m_1'}, \asmreply{\regset'}{m_2'}) \in \scname{R}_\kcainjp^r((j,m_1,m_2),\sig,\regset)$ if there is a $j'$ s.t.
    \begin{tabbing}
      \quad\=\kill
      \> (1)  $\injpacc{(j,m_1,m_2)}{(j',m_1',m_2')}$\\
      \> (2)  $\minj{j'}{m_1'}{m_2'},\quad  \vinj{j'}{r}{\kwd{get-result}(\sig,\regset')}$\\
      \> (3)  $\kwd{outgoing-arguments}(\sig,\regset(\spreg)) \subseteq \outofreach{j}{m_1}$\\
      \> (4)  $\regset'(\spreg) = \regset(\spreg),\quad \regset'(\pcreg) = \regset(\rareg), \quad \forall r \in \kwd{callee-save-regs}, \regset'(r) = \regset(r)$
    \end{tabbing}
    ${\kwd{get-result}(\sig,\regset')}$ is the return value
    stored in a register designated by CompCert's calling convention for the
    given signature $\sig$. \kwd{callee-save-regs} is the set of
    callee-save registers.
  \end{itemize}

\end{definition}

The last $\kasm_{\kinjp}$ is irrelevant as
assembly code is self-simulating by~\thmref{thm:self-sim}.
%
%
The final correctness theorem is shown below:
\begin{theorem}\label{thm:compcerto-correct}
  Compilation in CompCert is correct in terms of open simulations,
  \[
  \forall\app (M:\kwd{Clight})\app (M':\kwd{Asm}),\app
  \kwd{CompCert}(M) = M' \imply
  \osims{\sccompcerto}{\sem{M}}{\sem{M'}}.
  \]
\end{theorem}
We explain how the refinements are carried out at both sides.
%
%
The following is the sequence of refined simulation conventions
$\sccompcerto \sqsubseteq \scname{R}_n \sqsubseteq \ldots \sqsubseteq
\scname{R}_1 \sqsubseteq \scname{R}$ at the outgoing side. It begins with $\scname{R}$ and
ends with $\sccompcerto$.

{\small
\begin{tabbing}
  \quad\=$(11)$ \=\kill
  \>$(1)$\>$\kro \compsymb \textcolor{red}{\kc_{\kinjp} \compsymb \kc_{\kinjp} \compsymb \kc_{\kinjp}} \compsymb \kwt \compsymb \textcolor{red}{\kc_{\kext} \compsymb \kc_{\kext} \compsymb \kc_\kinj \compsymb \kc_{\kext}} \compsymb \kc_{\kinjp} \compsymb \textcolor{red}{\kro \compsymb \kc_\kinjp \compsymb \kro \compsymb \kc_\kinjp \compsymb \kro \compsymb \kc_\kinjp} $\\
  \>\>$ \compsymb \kc_{\kinj} \compsymb \kwt \compsymb \kc_{\kext} \compsymb \kcl
  \compsymb \kltl_{\kext} \compsymb \kltl_{\kinjp} \compsymb \klm \compsymb \kmach_{\kext} \compsymb \kma \compsymb \kasm_{\kinj} \compsymb \kasm_{\kinjp}$\\
  \>$(2)$\>$\kro \compsymb \kc_{\kinjp} \compsymb \kwt \compsymb \kc_\kinj \compsymb \kc_{\kinjp} \compsymb \kro \compsymb \kc_{\kinjp}$\\
  \>\>$\compsymb \kc_{\kinj} \compsymb \textcolor{red}{\kwt} \compsymb \kc_{\kext} \compsymb \kcl
  \compsymb \kltl_{\kext} \compsymb \kltl_{\kinjp} \compsymb \klm \compsymb \kmach_{\kext} \compsymb \kma \compsymb \kasm_{\kinj} \compsymb \kasm_{\kinjp}$ \\
  \>$(3)$\>$\kro \compsymb \kc_{\kinjp} \compsymb \textcolor{red}{\kwt \compsymb \kc_\kinj \compsymb \kwt} \compsymb \kc_{\kinjp} \compsymb \kro \compsymb \kc_{\kinjp}$\\
  \>\>$\compsymb \kc_{\kinj} \compsymb \kc_{\kext} \compsymb \kcl
  \compsymb \kltl_{\kext} \compsymb \kltl_{\kinjp} \compsymb \klm \compsymb \kmach_{\kext} \compsymb \kma \compsymb \kasm_{\kinj} \compsymb \kasm_{\kinjp}$ \\
  \>$(4)$\>$\kwt \compsymb \kro \compsymb \kc_{\kinjp} \compsymb  \kc_\kinj \compsymb \kc_{\kinjp} \compsymb \kro \compsymb \kc_{\kinjp}$\\
  \>\>$\compsymb \kc_{\kinj} \compsymb \kc_{\kext} \compsymb \kcl
  \compsymb \textcolor{red}{\kltl_{\kext} \compsymb \kltl_{\kinjp}} \compsymb \klm \compsymb \textcolor{red}{\kmach_{\kext}} \compsymb \kma \compsymb \textcolor{red}{\kasm_{\kinj}} \compsymb \kasm_{\kinjp}$ \\
  \>$(5)$\>$\kwt \compsymb \kro \compsymb \kc_{\kinjp} \compsymb \textcolor{red}{ \kc_\kinj} \compsymb \kc_{\kinjp} \compsymb \kro \compsymb \kc_{\kinjp}  \compsymb \textcolor{red}{ \kc_{\kinj} \compsymb \kc_{\kext} \compsymb \kc_{\kext}} \compsymb \kc_{\kinjp} \compsymb \textcolor{red}{\kc_{\kext} \compsymb \kc_{\kinj}} \compsymb \kcl \compsymb \klm \compsymb \kma \compsymb \kasm_{\kinjp}$\\
  \>$(6)$\>$\kwt \compsymb \kro \compsymb \textcolor{red}{\kc_{\kinjp} \compsymb \kc_\kinjp \compsymb \kc_{\kinjp}} \compsymb \kro \compsymb \textcolor{red}{\kc_{\kinjp} \compsymb \kc_\kinjp \compsymb \kc_\kinjp \compsymb \kc_\kinjp} \compsymb \kcl \compsymb \klm \compsymb \kma \compsymb \kasm_{\kinjp}$\\
  \>$(7)$\>$\kwt \compsymb \textcolor{red}{\kro \compsymb \kc_{\kinjp} \compsymb \kro \compsymb \kc_{\kinjp}}  \compsymb \kcl \compsymb \klm \compsymb \kma \compsymb \kasm_{\kinjp}$\\
  \>$(8)$\>$\textcolor{red}{\kwt \compsymb \kro} \compsymb \kc_{\kinjp} \compsymb \kcl \compsymb \klm \compsymb \kma \compsymb \kasm_{\kinjp}$\\
  \>$(9)$\>$\kro \compsymb \kwt \compsymb \textcolor{red}{\kc_{\kinjp} \compsymb \kcl \compsymb \klm \compsymb \kma} \compsymb \kasm_{\kinjp}$\\
  \>$(10)$\>$\kro \compsymb \kwt \compsymb \kcainjp \compsymb \kasm_{\kinjp}$
\end{tabbing}}
In each line, the letters in red are simulation conventions
transformed by the refinement operation at that step. 
In step (1), we merge consecutive simulation conventions by applying
property (1) in ~\lemref{lem:sim-refine} for $\kc_\kinjp$ and properties
(5-7) to compose $\kc_\kext$ and absorb it into $\kc_\kinj$. We also apply
~\lemref{lem:roinjp-trans} to merge consecutive $\kro \compsymb \kc_\kinjp$.
%
In steps (2-3), we move and eliminate $\kwt$ by~Lemmas~\ref{lem:wt}
and~\ref{lem:inv-comm}.
%
In step (4), we lift conventions over $\kcl$, $\klm$ and $\kma$ to
higher positions by~\lemref{lem:ca-com}.
In step (5), we absorb $\kc_\kext$ into $\kc_\kinj$ again and further
turns $\kc_\kinj$ into $\kc_\kinjp$ by applying
$\screfine{\kc_\kinjp}{\kc_\kinj}$ (property (2) in~\lemref{lem:sim-refine}).
In step (6), we compose $\kc_\kinjp$ by applying
$\kc_{\kinjp} \equiv {\kc_{\kinjp}} \compsymb {\kc_{\kinjp}}$.
In step (7), we apply ~\lemref{lem:roinjp-trans} again to eliminate the second $\kro \compsymb \kc_\kinjp$.
In step (8), we commute the two semantic invariants of the source semantics by
~\lemref{lem:inv-comm}.
Finally, we merge $\kc_\kinjp$ with ${\kcl \compsymb \klm \compsymb \kma}$ into $\kcainjp$
in step (9).

The original simulation conventions at the incoming side are
parameterized by $\kinj$ which does not have memory protection as in
$\kinjp$. One can modify the proofs of CompCert to make
$\kinjp$ an incoming convention. However, we show that this is
unnecessary: with the inserted self-simulations over \kinjp,
conventions over \kinj may be absorbed into them. The following is the
refinement sequence $\scname{S} \sqsubseteq \scname{S}_1 \sqsubseteq
\ldots \sqsubseteq \scname{S}_m \sqsubseteq \sccompcerto$ that
realizes this idea.
{\small
\begin{tabbing}
  \quad\=$(11)$ \=\kill
  \>$(1)$\>$\kro \compsymb \kc_{\kinjp} \compsymb \textcolor{red}{\kc_{\kinj} \compsymb \kc_{\kinj}} \compsymb \kwt \compsymb \textcolor{red}{\kc_{\kext} \compsymb \kc_{\kext} \compsymb \kc_\kinj \compsymb \kc_{\kext}} \compsymb \kc_{\kinjp} \compsymb \textcolor{red}{\kro \compsymb \kc_\kinjp \compsymb \kro \compsymb \kc_\kinjp \compsymb \kro \compsymb \kc_\kinjp} $\\
  \>\>$ \compsymb \kc_{\kinj} \compsymb \kwt \compsymb \kc_{\kext} \compsymb \kcl
  \compsymb \kltl_{\kext} \compsymb \klm \compsymb \kmach_{\kinj} \compsymb \kmach_{\kext} \compsymb \kma \compsymb \kasm_{\kinj} \compsymb \kasm_{\kinjp}$\\
  \>$(2)$\>$\kro \compsymb \kc_{\kinjp} \compsymb \kc_{\kinj} \compsymb \kwt \compsymb \kc_\kinj \compsymb \kc_{\kinjp} \compsymb \kro \compsymb \kc_\kinjp $\\
  
  \>\>$ \compsymb \kc_{\kinj} \compsymb \textcolor{red}{\kwt} \compsymb \kc_{\kext} \compsymb \kcl
  \compsymb \kltl_{\kext} \compsymb \klm \compsymb \kmach_{\kinj} \compsymb \kmach_{\kext} \compsymb \kma \compsymb \kasm_{\kinj} \compsymb \kasm_{\kinjp}$\\
  \>$(3)$\>$\kro \compsymb \kc_{\kinjp} \compsymb \kc_{\kinj} \compsymb \textcolor{red}{\kwt \compsymb \kc_\kinj \compsymb \kwt} \compsymb \kc_{\kinjp} \compsymb \kro \compsymb \kc_\kinjp $\\
  \>\>$ \compsymb \kc_{\kinj} \compsymb \kc_{\kext} \compsymb \kcl
  \compsymb \kltl_{\kext} \compsymb \klm \compsymb \kmach_{\kinj} \compsymb \kmach_{\kext} \compsymb \kma \compsymb \kasm_{\kinj} \compsymb \kasm_{\kinjp}$\\
  \>$(4)$\>$ \kwt \compsymb \kro \compsymb \kc_{\kinjp} \compsymb \kc_{\kinj} \compsymb \kc_\kinj \compsymb \kc_{\kinjp} \compsymb \kro \compsymb \textcolor{red}{\kc_\kinjp} $\\
  \>\>$ \compsymb \kc_{\kinj} \compsymb \kc_{\kext} \compsymb \kcl
  \compsymb \kltl_{\kext} \compsymb \klm \compsymb \kmach_{\kinj} \compsymb \kmach_{\kext} \compsymb \kma \compsymb \kasm_{\kinj} \compsymb \kasm_{\kinjp}$\\
  \>$(5)$\>$\kwt \compsymb \kro \compsymb \kc_{\kinjp} \compsymb \kc_{\kinj} \compsymb \kc_\kinj \compsymb \kc_{\kinjp} \compsymb \kro \compsymb \kc_\kinjp \compsymb \textcolor{red}{\kc_\kinjp} $\\
  \>\>$ \textcolor{red}{\compsymb \kc_{\kinj} \compsymb \kc_{\kext}} \compsymb \kcl \compsymb
  \textcolor{red}{ \kltl_{\kext}} \compsymb \klm \compsymb \textcolor{red}{\kmach_{\kinj} \compsymb \kmach_{\kext}} \compsymb \kma \compsymb \kasm_{\kinj} \compsymb \kasm_{\kinjp}$\\
 \>$(6)$\>$\kwt \compsymb \kro \compsymb \kc_{\kinjp} \compsymb \kc_{\kinj} \compsymb \kc_\kinj \compsymb \kc_{\kinjp} \compsymb \kro \compsymb \kc_\kinjp $\\
 \>\>$ \compsymb \kcl \compsymb \klm \compsymb \kma \compsymb \kasm_\kinjp \compsymb  \kasm_\kinj \compsymb \textcolor{red}{\kasm_\kext \compsymb \kasm_\kext \compsymb \kasm_\kinj \compsymb \kasm_\kext}  \compsymb \kasm_{\kinj} \compsymb \kasm_{\kinjp}$\\
 \>$(7)$\>$\kwt \compsymb \kro \compsymb \kc_{\kinjp} \compsymb \textcolor{red}{\kc_{\kinj} \compsymb \kc_\kinj} \compsymb \kc_{\kinjp} \compsymb \kro \compsymb \kc_\kinjp $ \\
 \>\> $ \compsymb \kcl \compsymb \klm \compsymb \kma \compsymb \kasm_\kinjp \compsymb  \textcolor{red}{\kasm_\kinj \compsymb \kasm_\kinj  \compsymb \kasm_{\kinj}} \compsymb \kasm_{\kinjp}$\\
 \>$(8)$\>$\kwt \compsymb \kro \compsymb \textcolor{red}{\kc_{\kinjp} \compsymb \kc_\kinj \compsymb \kc_{\kinjp}} \compsymb \kro \compsymb \kc_\kinjp
 \compsymb \kcl \compsymb \klm \compsymb \kma \compsymb \textcolor{red}{\kasm_\kinjp \compsymb \kasm_\kinj \compsymb \kasm_{\kinjp}}$\\
 \>$(9)$\>$\kwt \compsymb \textcolor{red}{\kro \compsymb \kc_{\kinjp} \compsymb \kro \compsymb \kc_\kinjp} \compsymb \kcl \compsymb \klm \compsymb \kma \compsymb \kasm_\kinjp$\\
 \>$(10)$\>$\textcolor{red}{\kwt \compsymb \kro} \compsymb \kc_{\kinjp} \compsymb \kcl \compsymb \klm \compsymb \kma \compsymb \kasm_\kinjp$\\
 \>$(11)$\>$\kro \compsymb \kwt \compsymb \textcolor{red}{\kc_{\kinjp} \compsymb \kcl \compsymb \klm \compsymb \kma} \compsymb \kasm_\kinjp$\\
 \>$(12)$\>$\kro \compsymb \kwt \compsymb \kcainjp \compsymb \kasm_\kinjp$\\

\end{tabbing}}

Steps (1-3) are the same as for the outgoing side except for using
property (4) in ~\lemref{lem:sim-refine}.
In step (4), we split $\kc_\kinjp$ into two, one will be used to absorb the $\kasm_\kinj$
at the target level. 
%
In step (5), we push all simulation conventions parameterized over KMRs starting
with the second split $\kc_\kinjp$ to target level by ~\lemref{lem:ca-com}.
In step (6), we absorb $\kasm_\kext$ into $\kasm_\kinj$ by properties (5-7)
in  ~\lemref{lem:sim-refine}.
In step (7), we compose the consecutive $\kc_\kinj$ and $\kasm_\kinj$ by
$\scname{R}_\kinj \compsymb \scname{R}_\kinj \sqsubseteq \scname{R}_\kinj$ (
property (4) in ~\lemref{lem:sim-refine}).
In step (8), we absorb $\kinj$ into $\kinjp$ at both levels by property (3)
in ~\lemref{lem:sim-refine}.
In step (9), we eliminate a redundant $\kro \compsymb \kc_\kinjp$ by ~\lemref{lem:roinjp-trans}.
%
The last two steps are the same as above.

\section{End-to-End Verification of Heterogeneous Modules}
\label{sec:application}

In this section, we give a formal account of end-to-end verification
of heterogeneous modules based on direct refinements. The discussion
focuses on the running example in~\figref{fig:running-exm-refinement}
and its variants. More detailed development of those examples can be
found in~\apdxref{sec:server-sim}. We also develop an additional
example adapted from CompCertM in~\apdxref{sec:mut-sum}.

\subsection{Refinement for the Hand-written Server}
\begin{figure}
  \begin{subfigure}[b]{0.55\textwidth}
    \centering
  \begin{tikzpicture}
    \node (qi) at (0,0){$q^I$};
    \node (calle) [draw, ellipse,right = 0.5 of qi,
    minimum width=1.4cm, minimum height=0.7cm] {\small \kwd{Calle}};
    \node (callp) [draw, ellipse,right = 1.5 of calle,
    minimum width=1.4cm, minimum height=0.7cm] {\small \kwd{Callp}};
    \node (retp) [draw, ellipse, below = 0.5 of callp,
    minimum width=1.4cm, minimum height=0.7cm] {\small\kwd{Retp}};
    \node (rete) [draw, ellipse, below = 0.5 of calle,
    minimum width=1.4cm, minimum height=0.7cm] {\small\kwd{Rete}};
    \node (qo) [right = 0.5 of callp] {$q^O$};
    \node (ro) [below = 0.63 of qo] {$r^O$};
    \node (ri) [left = 0.5 of rete] {$r^I$};

    \draw [-stealth](qi) -- node[above]{$I_S$} (calle);
    \draw [-stealth](callp) -- node[above]{$X_S$} (qo);
    \draw [-stealth](ro) -- node[below]{$Y_S$} (retp);
    \draw [-stealth](rete) -- node[below]{$F_S$} (ri);

    \draw [-stealth](calle) -- node[above]{\small \kwd{alloc}} node[below]{\small $\kwd{encrypt}$} (callp);
    \draw [-stealth](retp) -- node[above]{\small \kwd{free}} (rete);

    \draw [-stealth](qo) -- node[above, sloped]{\tiny \kwd{external}} (ro);
  \end{tikzpicture}
  \caption{$L_{\texttt{S}}$}
  \label{fig:server-spec}
\end{subfigure}
\begin{subfigure}[b]{0.35\textwidth}
  \centering
  \begin{tikzpicture}
    \node (qi) at (0,0){$q^I$};
    \node (aa) [right = 1.5 of qi] {};
    \node (st) [draw, ellipse,below = 0.4 of aa] {\small$\regset@m$};
    
    \node (qo) [right = 1.5 of aa] {$q^O$};
    \node (ro) [below  = 1 of qo] {$r^O$};
    \node (ri) [below  = 1 of qi] {$r^I$};

    \draw [-stealth](qi) -- node[above]{$I_\asmli$} (st);
    \draw [-stealth](st) -- node[above]{$X_\asmli$} (qo);
    \draw [-stealth](ro) -- node[below]{$Y_\asmli$} (st);
    \draw [-stealth](st) -- node[below]{$F_\asmli$} (ri);

     \path[-stealth, every loop/.style={looseness=8, out=120, in=60, distance=10mm}] 
    (st) edge [loop] ();

    \draw [-stealth](qo) -- node[above, sloped]{\tiny \kwd{external}} (ro);
  \end{tikzpicture}
  \caption{$\sem{\code{server\_opt.s}}$}
  \label{fig:server-sem}
\end{subfigure}
\caption{Specification and Open Semantics of $\code{server\_opt.s}$}
\end{figure}
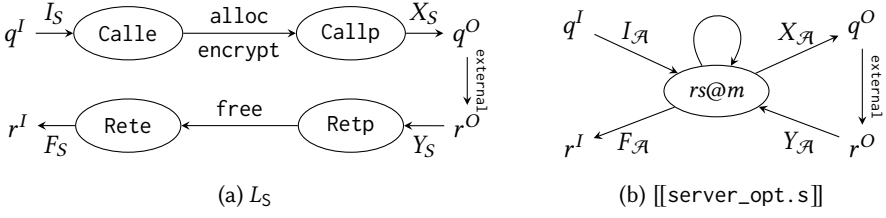



We use \kwd{server\_opt.s} instead of \kwd{server.s} to illustrate how
optimizations are enabled by \kro. The proof for the unoptimized
server is similar with only minor adjustments.
A formal definition of LTS for $L_{\texttt{S}}$ is given below and its transition diagram is given in~\figref{fig:server-spec}.
\begin{definition}
LTS of $L_{\texttt{S}}$:
\begin{tabbing}
  \quad\=$\to_A$\;\=$:=$ \=\kill
  \>$S_S$ \>$:=$ \>$\rawset{\kwd{Calle}\app i\app v_f \app m,  
          \kwd{Callp}\app \rb \app v_f\app m, 
          \kwd{Retp}\app \rb \app m, \kwd{Rete} \app m}$;\\
  \>$I_S$ \>$:=$ \>$\rawset{(\cquery{\vptr{b_e}{0}}{\intptrvoidsig}{[\Vint{i},v_f]}{m},\kwd{Calle}\app i\app v_f \app m)}$;\\  
  \>$\to_S$ \>$:=$
   \>$\{(\kwd{Calle}\app i\app v_f \app m, \kwd{Callp}\app \rb\app v_f \app m'') \;|\; (m',\rb) = \kwd{alloc} \app m \app 0 \app 8 \app \land $\\
  \>\>\> $m'' = m'[\rb \leftarrow (i \app \kwd{XOR} \app m[b_k])]\} \cup
  \pset{(\kwd{Retp}\app \rb \app m, \kwd{Rete} \app m')}{m' = \kwd{free}\app m\app \rb}$;\\    
  \>$X_S$ \>$:=$ \>$\rawset{(\kwd{Callp}\app \rb \app \vptr{b_p}{0}\app m,\cquery{\vptr{b_p}{0}}{\ptrvoidsig}{[\vptr{\rb}{0}]}{m})}$;\\
  \>$Y_S$ \>$:=$ \>$\rawset{(\kwd{Callp}\app \rb\app v_f\app m, res@m',\kwd{Retp}\app \rb \app m')}$;\\
  \>$F_S$ \>$:=$ \>$\rawset{(\kwd{Rete} \app m,\Vundef@m)}$.
  \end{tabbing}

\end{definition}
\noindent The LTS has four internal states as depicted in~\figref{fig:server-spec}.
%
Initialization is encoded in $I_S$. If the incoming query $q^I$ contains
a function pointer $\vptr{b_e}{0}$ which points to $\kencrypt$,
$L_{\texttt{S}}$ enters $\kwd{Calle} \app i \app v_f \app m$ where $i$
and $v_f$ are its arguments.
The first internal transition allocates the stack frame $\rb$ and stores
the result of encryption $i \app \kwd{XOR} \app m[b_k]$ in $\rb$ where
$b_k$ contains $\kwd{key}$. Then, it enters $\kwd{Callp}$
which is the state before calling $\kwd{process}$.
%
%
If the pointer $v_f = \vptr{b_p}{0}$ of the current
state points to an external function, $L_{\texttt{S}}$ 
issues an outgoing C query $q^O$ with a pointer
to its stack frame as its argument.
After the external call, $Y_S$ updates the memory with the reply
and enters $\kwd{Retp}$.
The second internal transition frees $sp$ and enters $\kwd{Rete}$ and
finally returns.
Note that complete semantics of $L_{\texttt{S}}$ is accompanied by a
local symbol table which determines the initial value of global
variables ($\kwd{key}$) and asserts that it is a constant (read-only).
The only difference between the specifications for \kwd{server\_opt.s} and
\kwd{server.s} is whether $\kwd{key}$ is a constant in the symbol
table.
The semantics of assembly module $\sem{\code{server\_opt.s}}$ is given
by CompCertO whose transition diagram is shown
in~\figref{fig:server-sem}. All the states, including
queries and replies, are composed of register sets and memories. 
%
%
%

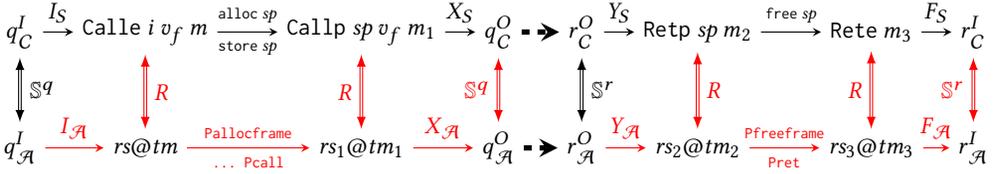
\begin{figure}
  \centering
  \begin{tikzpicture}

  \node (qci) at (0,0) {\small{$q_\cli^I$}};
  \node (qai) [below = 0.85cm of qci] {\small{$q_\asmli^I$}};
  \node (sc) [right = 0.4 cm of qci] {\small{\kwd{Calle} $i$ $v_f$ $m$}};
  \path let \p1 = (qai) in let \p2 = (sc) in
        node (sa) at (\x2, \y1) {\small{$\regset@\tm$}};
  \node (sc1) [right = 0.8 cm of sc] {\small{\kwd{Callp} $\rb$ $v_f$ $m_1$}};
  \path let \p1 = (qai) in let \p2 = (sc1) in
        node (sa1) at (\x2, \y1) {\small{$\regset_1@\tm_1$}};
  \node (qco) [right = 0.4 cm of sc1] {\small{$q_\cli^O$}};
  \path let \p1 = (qai) in let \p2 = (qco) in
        node (qao) at (\x2, \y1) {\small{$q_\asmli^O$}};      
  \node (rco) [right = 0.5 cm of qco] {\small{$r_\cli^O$}};
  \path let \p1 = (qai) in let \p2 = (rco) in
        node (rao) at (\x2, \y1) {\small{$r_\asmli^O$}};
  \node (sc2) [right = 0.4 cm of rco] {\small{\kwd{Retp} $\rb$ $m_2$}};
  \path let \p1 = (qai) in let \p2 = (sc2) in
        node (sa2) at (\x2, \y1) {\small{$\regset_2@\tm_2$}};
  \node (sc3) [right = 0.8 cm of sc2] {\small{\kwd{Rete} $m_3$}};
  \path let \p1 = (qai) in let \p2 = (sc3) in
        node (sa3) at (\x2, \y1) {\small{$\regset_3@\tm_3$}};
  \node (rci) [right = 0.4 cm of sc3] {\small{$r_\cli^I$}};
  \path let \p1 = (qai) in let \p2 = (rci) in
        node (rai) at (\x2, \y1) {\small{$r_\asmli^I$}};
    
  \draw [double, latex-latex] (qci) -- node[right] {\small{$\mathbb{S}^q$}} (qai);
  \draw [double, latex-latex, red] (qco) -- node[left] {\small{$\mathbb{S}^q$}}(qao);
  \draw [double, latex-latex, red] (rci) --node[left] {\small{$\mathbb{S}^r$}} (rai);
  \draw [double, latex-latex] (rco) --node[right] {\small{$\mathbb{S}^r$}} (rao);
  \draw [double, latex-latex, red] (sc) --node[right] {\small{$R$}} (sa);
  \draw [double, latex-latex, red] (sc1) --node[left] {\small{$R$}} (sa1);
  \draw [double, latex-latex, red] (sc2) --node[right] {\small{$R$}} (sa2);
  \draw [double, latex-latex, red] (sc3) --node[left] {\small{$R$}} (sa3);
  
  \draw [-stealth] (qci) -- node[above] {\small{$I_S$}} (sc);
  \draw [-stealth] (sc)  -- node[above] {\tiny{\kwd{alloc} $\rb$}}
                      node[below] {\tiny{\kwd{store} $\rb$}}(sc1);
  \draw [-stealth] (sc1) -- node[above] {\small{$X_S$}} (qco);
  \draw [-stealth] (rco) -- node[above] {\small{$Y_S$}} (sc2);
  \draw [-stealth] (sc2) -- node[above] {\tiny{\kwd{free} $\rb$}} (sc3);
  \draw [-stealth] (sc3) -- node[above] {\small{$F_S$}} (rci);
  \draw [-stealth, line width = 2, dashed] (qco) -- (rco);

  \draw [-stealth, red] (qai) -- node[above] {\small{$I_\asmli$}} (sa);
  \draw [-stealth, red] (sa)  -- node[above] {\tiny{\kwd{Pallocframe}}}
                      node[below] {\tiny{\kwd{... Pcall}}}(sa1);
  \draw [-stealth, red] (sa1) -- node[above] {\small{$X_\asmli$}} (qao);
  \draw [-stealth, red] (rao) -- node[above] {\small{$Y_\asmli$}} (sa2);
  \draw [-stealth, red] (sa2) -- node[above] {\tiny{\kwd{Pfreeframe}}}
                      node[below] {\tiny{\kwd{Pret}}} (sa3);
  \draw [-stealth, red] (sa3) -- node[above] {\small{$F_\asmli$}} (rai);
  \draw [-stealth, line width = 2, dashed] (qao) -- (rao);

\end{tikzpicture}
\caption{Open Simulation between the Optimized Server and its Specification}
\label{fig:server-sim}
\end{figure}

Now, we need to prove the following forward simulation. The most
important points of the proof are how \kro enables optimizations and
how \kinjp preserves memory across external calls.

%
\begin{theorem}\label{the:l2sim}
  $\osims{\scc}{\Sspec}{\sem{\code{server\_opt.s}}}$.
\end{theorem}
At the top level, we expand $\scc$ to $\kro \compsymb \kwt \compsymb \kcainjp
\compsymb \kasm_\kinjp$ and switch the order of $\kro$ and $\kwt$
by~\lemref{lem:inv-comm}.
By the vertical compositionality (\thmref{thm:v-comp}), we first establish
$\osims{\kwt}{L_{\texttt{S}}}{L_{\texttt{S}}}$ with the well-typed outgoing
arguments and return value.
$\osims{\kasm_\kinjp}{\sem{\code{server\_opt.s}}}{\sem{\code{server\_opt.s}}}$
is proved by~\thmref{thm:self-sim}.

We are left with proving
$\osims{\kro\compsymb\kcainjp}{\Sspec}{\sem{\code{server\_opt.s}}}$.
%
%
That is, we need to show the simulation diagram in~\figref{fig:server-sim} holds 
where $\mathbb{S} = \kro \compsymb \kcainjp$ (which mirrors~\figref{fig:server-refinement}).
Here, the given assumptions and the conclusions to be
proved are represented as black and red arrows, respectively.
For the proof, we need an invariant $R \in \krtype{W_{\kro \compsymb
    \kcainjp}}{S_S}{\kregset \times \kmem}$.
The most important point is that \kwd{ro} and \kinjp play essential
roles in establishing the invariant.
First, $\krovalid$ is propagated from the initial source query
$q_\cli^I$ to internal program states.
This guarantees that the value of \kwd{key} read from the source
memory states is always $42$, hence matching the constant in \code{Pxori 42
  RDI} in $\code{server\_opt.s}$.
Second, \kinjp is essential for deriving that memory locations in the target stack frame
with offset $o$ ($o < 8$ or $16 \leq o$) are unchanged since
they are designated out-of-reach by $R$. Therefore, the private stack
values of the server are protected. 
For the unoptimized server, the only difference is that we decompose
$\osims{\kro\compsymb\kcainjp}{\Sspec}{\sem{\code{server.s}}}$
into $\osims{\kro}{\Sspec}{\Sspec}$ which trivially
holds and $\osims{\kcainjp}{\Sspec}{\sem{\code{server.s}}}$ which can
be proved without the help of \kro.

\subsection{End-to-end Correctness Theorem}

We first prove the following source-level refinement where $\Topspec$ is the
top-level specification.
%
Its proof follows the same pattern as~\thmref{the:l2sim} but is
considerably simpler because the source and target
semantics share the same $\cli$ interface. 

\begin{lemma}\label{lem:source-veri}
  $\osims{\kro \compsymb \kwt \compsymb \kc_\kinjp}{\Topspec}{\sem{\code{client.c}} \semlink \Sspec}$.
\end{lemma}

We then prove the forward simulation between the top-level specification
and the linked assembly as depicted
in~\figref{fig:running-exm-refinement}, which is immediate from the horizontal
compositionality and adequacy for assembly
described in~\secref{sssec:framework},~\thmref{thm:compcerto-correct} and~\thmref{the:l2sim}:
%
%
%
\begin{lemma}\label{lem:veri-link}
  $\osims{\scc}{\sem{\code{client.c}} \semlink L_\texttt{S}}{\sem{\code{client.s} + \code{server\_opt.s}}}$.
\end{lemma}

For end-to-end direct refinement, we need to absorb
\lemref{lem:source-veri} into ~\lemref{lem:veri-link}. 
The following theorem is easily
derived by applying Lemmas~\ref{lem:wt},
\ref{lem:roinjp-trans} and \ref{lem:inv-comm}.
\begin{lemma}\label{lem:toprefine}
  $\scequiv{\scc}{\kro \compsymb \kwt \compsymb \kc_\kinjp \compsymb \scc}$.
\end{lemma}

The final end-to-end simulation is immediate by vertically
composing~\lemref{lem:source-veri}, ~\lemref{lem:veri-link} and
refining the simulation convention using ~\lemref{lem:toprefine}.

\begin{theorem}\label{thm:cs-final}
  $\osims{\scc}{\Topspec}{\sem{\code{client.s} + \code{server\_opt.s}}}$.
\end{theorem}


\subsection{Verification of the Mutually Recursive Client and Server}
\label{ssec:app-mutual-client-server}

\begin{figure}
  \begin{subfigure}{0.37\textwidth}
\centering   
\begin{lstlisting}[language = C]
  /* client.c */
  #define N 10
  int input[N] = {...};
  int result[N];
  int i;
  void encrypt(int i, 
         void(*p)(int*));
\end{lstlisting}
\end{subfigure}
\begin{subfigure}{.6\textwidth}
\centering
\begin{lstlisting}[language = C]
  void request(int *r) {
    if (i == 0) encrypt(input[i++], request);
    else if (0 < i && i < N) {
      result[i-1] = *r;
      encrypt(input[i++], request); 
    } else result[i-1] = *r;
  }
\end{lstlisting}
\end{subfigure}
\caption{Client with Multiple Encryption Requests}
\label{fig:clientMR}
\end{figure}

We introduce a variant of the running example with mutual recursion
in~\figref{fig:clientMR}. The server remains the same while the client
is changed. \code{request} itself is passed as a callback function to
\code{encrypt}, resulting in recursive calls to \code{encrypt} for
encrypting and storing an array of values.
To perform the same end-to-end verification for this example, we only
need to define a new top-level specification $\Topspec'$ and prove
$\osims{\kro \compsymb \kwt \compsymb
  \kc_\kinjp}{\Topspec'}{\sem{\code{client.c}} \semlink \Sspec}$. Other
proofs are either unchanged (e.g., the refinement of the server) or can be
derived from~\thmref{thm:compcerto-correct}, 
full compositionality and adequacy. More detailed proofs can be found 
in~\apdxref{sec:server-sim}.

\section{Generality and Limitations of Our Approach}
\label{sec:generality}


We explain how our approach may be generalized
to support other memory models, compilers and optimizations for first-order
languages. We also
discuss the limitations of our approach.


\subsection{Supporting Different Memory Models and Compilers}

At a high level, \kinjp is simply a general and transitive relation on
evolving \emph{functional memory invariants} (represented as
injections) enhanced with \emph{memory protection} to guard against
modification to private memory by external calls. Many other
first-order memory models can be viewed as employing either a richer
or a simplified version of injection as memory invariants and equipped
with a similar notion of memory protection. For example, the memory
model of CompCertS~\cite{besson2015concrete} extends injections
to map \emph{symbolic values}. The \emph{memory refinements} in the
$CH_2O$ memory model~\cite{krebbers2016formal} function like injections except
that pointer offsets are represented as abstract
\emph{paths} pointing into aggregated data structures. The memory model
defined by~\citet{kang2015formal} explicitly divides a memory state
into public and private memory. Its memory invariant is an equivalence
relation between public source and target memory which is essentially
an identity injection.
Therefore, uniform KMRs may be defined for those memory models as
variants of \kinjp.

To prove the transitivity of these KMRs, the key is the construction
of interpolating memory states after external calls as described
in~\secref{ssec:injp-trans}. This construction is based on the
following general ideas: \emph{1)} as KMRs are transitively composed,
more memory gets protected, \emph{2)} the private memory should be
identical to the initial memory, and \emph{3)} the public memory
should be projected from the updated source memory via memory
invariants. As we can see, these ideas are applicable to any memory
model with functional memory invariants and a notion of private
memory. Therefore, our approach should work for the 
aforementioned memory models and compilers based on them.



\subsection{Supporting Additional Optimization Passes}
Given any new optimization pass whose additional rely-guarantee
condition can be represented as a semantics invariant $I$, we may
piggyback $I$ onto an enriched \kinjp to achieve direct refinement.
To see that, note that any $I$ consists of two parts: a condition for
initial queries (e.g., \kwd{ro-valid} in \kro) and a condition for
replies (e.g., \kwd{mem-acc} in \kro).
By extending \kinjp to include the latter (just like that
\defref{def:injp} includes \kwd{mem-acc}), if the enriched \kinjp is
still transitive, then we can easily prove the following proposition
which is generalized from~\lemref{lem:roinjp-trans}.
\begin{proposition}\label{thm:trans-I}
  For any $I : \sctype{\cli}{\cli}$ and enriched
  $\kinjp$, $\scequiv{I \compsymb \kc_\kinjp}{I \compsymb \kc_\kinjp
    \compsymb I \compsymb \kc_\kinjp}$.
\end{proposition}
\noindent The proof follows exactly the steps for
proving~\lemref{lem:roinjp-trans}. It is based on the two observations
we made near the end of~\secref{sssec:sinv-comp}, i.e., \emph{1)} the
properties for initial queries of $I$ hold along with copying of
memory states, and \emph{2)} the properties for replies trivially
hold as they are part of the enriched \kinjp.

\subsection{Limitations}
\label{ssec:limit}

We discuss the limitations of our approach and possible
solutions. First, it does not yet support behavior refinements for
whole programs in CompCert~\cite{compcert}. This is a technical limitation and can be
solved by reducing open simulations into
closed simulations in CompCert. Second, the proof for \kwd{Unusedglob}
assumes that global symbols for removed definitions are preserved as
CompCertO's simulation framework requires the same set of global
symbols throughout compilation. We need to weaken this requirement to
enable removal of global symbols by compilation.
Third, open simulations assume given
\emph{any} input injection $j$, the execution outputs \emph{some}
injection $j'$ related to $j$ by \kinjp. This may not work for memory
models with fixed injection functions~\cite{wang-popl-2022}. 
A possible solution is to enrich \kinjp to
account for this fixed definition. Finally, given a new optimization, if
its rely-guarantee condition cannot be described as a semantic
invariant $I$ or if \kinjp enriched with $I$
becomes intransitive, then direct refinements may not be
derivable. In this case, we may need stronger restrictions on
this optimization for our approach to work.

\section{Evaluation and Related Work}
\label{sec:related}

Our Coq development took about 7 person-months and 18.3k lines of code
(LOC) on top of CompCertO. 
We added 3.7k LOC to prove the transitivity of \kinjp, 3k
LOC to verify the compiler passes 
as discussed in~\secref{ssec:single-pass}, 1.2k LOC for composing
simulation conventions as described in the rest
of~\secref{sec:refinement} and 7.3k LOC for the Client-Server
examples. We also ported CompCertM's example on mutually recursive
summation~\cite{compcertm}, which adds 3.1k LOC~\cite{direct-refinement-tr}.
For now, the cost of examples is relatively high. However, we observe
that a lot of low-level proofs such as pointer arithmetic can be
automated by proof scripts, many proofs with predictable patterns can be
directly derived from the program structures, and a lot of
duplicated lemmas in the examples can be eliminated.
We will carry out those exercises in the future which should simplify the
proofs significantly.
Below we compare our work with other frameworks for
VCC and program verification.
 
\subsection{Verified Compositional Compilation for First-Order Languages}

\def\YES{\textcolor{codegreen}{Yes}\xspace}
\def\NO{\textcolor{red}{No}\xspace}
\def\BYES{\textcolor{codegreen}{\bf Yes}\xspace}
\def\BNO{\textcolor{red}{\bf No}\xspace}
\def\RUSC{\textcolor{brown}{RUSC}\xspace}
\def\CAL{\textcolor{brown}{CAL}\xspace}

\begin{table}
  \caption{Comparison between Work on VCC Based on CompCert}
\small
\begin{tabular}{c c c c c c c}
  & CompComp & CompCertM & CompCertO & CompCertX & {\bf This Work}\\
  \hline
  \hline
  Direct Refinement & \NO & \NO & \NO & \NO & \BYES\\
  Vertical Composition & \YES & \RUSC & \textcolor{brown}{Trivial} & \CAL & \BYES\\
  Horizontal Composition & \YES & \RUSC & \YES & \CAL & \BYES \\
  Adequacy & \NO & \YES & \YES & \YES & \BYES \\
  End-to-end Verification & \NO & \YES & \textcolor{brown}{Unknown} & \CAL & \BYES \\
  Free-form Heterogeneity & \YES & \YES & \YES & \NO & \BYES\\
  Behavior Refinement & \NO & \YES & \NO & \YES & \BNO
\end{tabular}
\label{tab:vcc-compare}
\end{table}


In this work, we are concerned with VCC of first-order imperative
programs with global memory states and support of pointers. A majority
of the work in this setting is based on CompCert. We compare them from the
perspectives listed in the first column of~\tabref{tab:vcc-compare}.
An answer that is not a simple ``Yes'' or ``No'' denotes that special
constraints are enforced to support the given feature. 

\paragraph{Compositional CompCert} 
CompComp supports VCC based on \emph{interaction semantics} which
is a specialized version of open semantics with C
interfaces~\cite{stewart15}. 
%
We have already talked about its merits and limitations in~\secref{ssec:problems}.
It is interesting to note that CompComp can also be obtained based on
our approach by adopting $\kc_\kinjp$ for every compiler pass and
exploiting the transitivity of $\kc_\kinjp$, which does not require
the instrumentation of semantics in CompComp.

\paragraph{CompCertM}
CompCertM supports adequacy and end-to-end verification of mixed C and
assembly programs. A distinguishing feature of CompCertM is Refinement
Under Self-related Contexts or RUSC~\cite{compcertm}. A RUSC relation
is a \emph{fixed} collection of simulation relations. By exploiting
contexts that are self-relating under all of these
simulation relations, horizontal and vertical compositionality are
achieved.
However, refinements based on RUSC relations can be difficult to use
as they are not extensional.
For example, the complete open refinement relation $\leqslant_{R_1 +
  \ldots + R_9}$ in CompCertM carries 9 RUSC relations $R_1, \ldots,
R_9$ (6 for compiler passes and 3 for source-level
verification). To establish the refinement between \kwd{a.s} and its
specification $L_S$, one needs to prove $L_S$ are self-simulating over
\emph{all} 9 simulation relations. This can quickly get out of hand as
more modules and more compiler passes are introduced.
By contrast, we only need to prove direct
refinement \emph{for once} and the refinement is open to
further horizontal or vertical composition.
%
%
%
On the other hand, CompCertM supports behavior refinement of
closed programs which we do not yet (See~\secref{ssec:limit}).

\paragraph{CompCertO}

Vertical composition is a trivial pairing of simulations in CompCertO,
which exposes internal compilation steps. CompCertO tries to alleviate
this problem via ad-hoc refinement of simulation conventions. The
resulting top-level convention is $\scc_{\kwd{CCO}} = \mathcal{R}^*
\cdot \kwt \cdot \kcl \cdot \klm \cdot \kma \cdot \kasm_{\kwd{vainj}}$
where $\mathcal{R} = \kc_\kinjp + \kc_\kinj + \kc_\kext +
\kc_\kwd{vainj} + \kc_\kwd{vaext}$ is a sum of conventions
parameterized over KMRs. In particular, $\kc_\kwd{vaext}$ is an ad-hoc
combination of KMR and internal invariants for
optimizations. $\mathcal{R}^*$ means that $\mathcal{R}$ may be repeated
for an arbitrary number of times. Since the top-level summation of
KMRs is similar to that in CompCertM, we need to go through a
reasoning process similar to CompCertM, only more complicated because
of the need to reason about internal invariants of optimizations
in $\kc_\kwd{vaext}$ and indefinitely repeated combination of all the
KMRs by $\mathcal{R}^*$. Therefore, it is unknown if the correctness
theorem of CompCertO suffices for end-to-end program verification. 

\paragraph{CompCertX}

CompCertX~\cite{dscal15,wang2019} realizes a weaker form of VCC that
only allows assembly contexts to invoke C programs, but not the other
way around. Therefore, it does not support horizontal composition of
modules with mutual recursions. The compositionality and program
verification are delegated to Certified Abstraction Layers
(CAL)~\cite{dscal15, ccal18}. Furthermore, CompCertX does not support
stack-allocated data (e.g., our server example). However, its
top-level semantic interface is similar to our interface, albeit not
carrying a symmetric rely-guarantee condition. This indicates that our
work is a natural evolution of CompCertX.

\paragraph{VCC for Concurrent Programs}

VCC for concurrent programs needs to deal with multiple threads and
their linking.
CASCompCert is an extension of CompComp that supports compositional
compilation of concurrency with no (or benign) data
races~\cite{cascompcert}.
To make CompComp's approach to VCC work in a concurrent setting,
CASCompCert imposes some restrictions including not supporting
stack-allocated data and allowing only nondeterminism in scheduling
threads.
A recent advancement based on CASCompCert is about verifying
concurrent programs~\cite{zha2022} running on weak memory models using
the promising semantics~\cite{promising1,promising2}.
We believe the ideas in CASCompCert are complementary to this work and
can be combined with our approach to achieve VCC for concurrency with
cleaner interface and less restrictions.

\subsection{Verified Compositional Compilation for Higher-Order Languages}

Another class of work on VCC focuses on compilation of higher-order
languages. In this setting, the main difficulty comes from complex
language features together with higher-order states. 
A prominent example is the Pilsner compiler~\cite{neis15icfp} that
compiles a higher-order language into some form of assembly programs.
The technique Pilsner adopts is called \emph{parametric simulations}
that evolves from earlier work on reasoning about program equivalence
via bisimulation~\cite{hur12}. 
%
%
Another line of work is multi-language
semantics~\cite{patterson-icfp-2019,funtal,perconti14,scherer18} where
a language combining all source, intermediate and target languages is
used to formalize semantics. Compiler correctness is stated
as contextual equivalence or logical relations.
It seems that our techniques are not directly applicable to those work
because relations on higher-order states cannot
deterministically fix the interpolating states. A possible solution is
to divide the higher-order memory into a first-order and a
higher-order part such that the former does not contain pointers to
the latter (forming a closure). By encapsulating higher-order
programs inside first-order states, we may be able to
apply our approach.


The high-level ideas for constructing interpolating states for
proving transitivity of \kinjp can also be found in some of the work
on program equivalence~\cite{hur2012tr,Ahmed06esop}.
To the best of our knowledge, our approach is the first concrete
implementation of these ideas that works for a realistic optimizing
compiler for imperative languages with non-trivial memory
models.


\subsection{Frameworks for Compositional Program Verification}
Researchers have proposed frameworks for compositional program
verification based on novel semantics, refinements and separation
logics~\cite{dscal15,ccal18,xia2019interaction,paul2021,choicetrees,sammler2023dimsum,song2023conditional}. These
frameworks aim at broader program verification and may be combined
with our approach to generate more flexible end-to-end verification
techniques. For example, to support more flexible certified
abstraction layers, we may combine our approach with data
abstraction in CAL and extend horizontal linking to work with
abstraction layers.
%
It is not entirely clear whether their solutions can be successfully
applied to or combined with VCC of realistic optimizing compilers like
CompCert. However, comparing with these frameworks is still meaningful
as it provides different perspectives and potential directions for
improving our work. We discuss representative frameworks in these
categories below.

\paragraph{DimSum}
DimSum~\cite{sammler2023dimsum} is a framework for multi-language
program verification. Program semantics are defined as LTSs which emit
\emph{events} to communicate with the environment. The concept of
events in DimSum is similar to the language interfaces in CompCertO
and in our work. For verification of heterogeneous programs, it uses
\emph{wrappers} to relate the events between two languages (a C-like
language called \kwd{Rec} and assembly in its paper), which is similar to
simulation conventions. The rely-guarantee protocol is expressed in
the wrappers by angelic non-determinism and memory protection is
expressed in separation logic.
%
On one hand, it is unclear if their
framework can scale to realistic languages or compilers like CompCert.
For example, it is interesting to investigate if their wrappers can support more
complicated languages and compiler optimizations which can be handled
by our framework and refinement relations.
%
On the other hand, DimSim allows assembly modules that exploit a flat memory model.
Therefore, their framework supports refinements between semantics using different
memory models, which we do not support yet.

\paragraph{Conditional Contextual Refinement}
Conditional Contextual Refinement (CCR) is a framework which combines
contextual refinement and separation logics to achieve both conditional
and composable verification of program semantics~\cite{song2023conditional}.
CCR employs separation logics to constrain the behavior of open
modules to achieve horizontal and vertical composition of refinements,
such \emph{separation logic wrapper} plays a similar role as
simulation conventions in this paper.
On one hand, separation logics provide more fine-grained control of
shared resources. On the other hand, they have specific requirements
of contexts unlike the open protocols encoded in our direct
refinements. The horizontal composition of two refinements
requires specific knowledge of specifications of each other to control
interaction.
It is interesting to investigate if the program specific
conditions imposed by CCR can be handled or piggybacked upon our
framework.

\section{Conclusion and Future Work}
\label{sec:conc}

We have proposed an approach to compositional compiler correctness for
first-order languages via direct refinements between source and target
semantics at their native interfaces, which overcomes the limitations
of the existing approaches on compositionality, adequacy and other
important criteria for VCC. In the future, we plan to support behavior
(trace) refinement for closed programs by reducing our open simulation
into the whole-program correctness theorem for the original
CompCert. We also plan to combine our work with refinement-based
program verification like certified abstraction layers to support more
substantial applications. Another research direction is to
apply our approach to different memory models and compilers for
first-order and higher-order languages, which will better test the
limit of our approach and the usefulness of our discoveries.


\section*{Data-Availability Statement}

The Coq artifact containing the formal developments described
in this paper is available on
Zenodo~\cite{direct-refinement-artifact}.

\begin{acks}                            
We would like to thank our shepherd Yannick Zakowski and the anonymous
referees for their helpful feedback which improved this paper significantly.
This
work is supported in part by the National Natural Science Foundation
of China (NSFC) under Grant No. 62002217 and 62372290, 
and by the Natural Science Foundation of the United States (NSF) under Grant No. 1763399, 2019285, and
2313433. Any opinions, findings, and conclusions or
recommendations expressed in this material are those of the authors
and do not necessarily reflect the views of the funding agencies. 
\end{acks}

\bibliography{refs}


\begin{thebibliography}{31}


\ifx \showCODEN    \undefined \def \showCODEN     #1{\unskip}     \fi
\ifx \showDOI      \undefined \def \showDOI       #1{#1}\fi
\ifx \showISBNx    \undefined \def \showISBNx     #1{\unskip}     \fi
\ifx \showISBNxiii \undefined \def \showISBNxiii  #1{\unskip}     \fi
\ifx \showISSN     \undefined \def \showISSN      #1{\unskip}     \fi
\ifx \showLCCN     \undefined \def \showLCCN      #1{\unskip}     \fi
\ifx \shownote     \undefined \def \shownote      #1{#1}          \fi
\ifx \showarticletitle \undefined \def \showarticletitle #1{#1}   \fi
\ifx \showURL      \undefined \def \showURL       {\relax}        \fi
\providecommand\bibfield[2]{#2}
\providecommand\bibinfo[2]{#2}
\providecommand\natexlab[1]{#1}
\providecommand\showeprint[2][]{arXiv:#2}

\bibitem[Ahmed(2006)]%
        {Ahmed06esop}
\bibfield{author}{\bibinfo{person}{Amal~J. Ahmed}.}
  \bibinfo{year}{2006}\natexlab{}.
\newblock \showarticletitle{Step-Indexed Syntactic Logical Relations for
  Recursive and Quantified Types}. In \bibinfo{booktitle}{\emph{Proc. 15th
  European Symposium on Programming (ESOP'06)}} \emph{(\bibinfo{series}{LNCS},
  Vol.~\bibinfo{volume}{3924})}, \bibfield{editor}{\bibinfo{person}{Peter
  Sestoft}} (Ed.). \bibinfo{publisher}{Springer}, \bibinfo{address}{Cham},
  \bibinfo{pages}{69--83}.
\newblock
\urldef\tempurl%
\url{https://doi.org/10.1007/11693024\_6}
\showDOI{\tempurl}


\bibitem[Besson et~al\mbox{.}(2015)]%
        {besson2015concrete}
\bibfield{author}{\bibinfo{person}{Fr{\'{e}}d{\'{e}}ric Besson},
  \bibinfo{person}{Sandrine Blazy}, {and} \bibinfo{person}{Pierre Wilke}.}
  \bibinfo{year}{2015}\natexlab{}.
\newblock \showarticletitle{A Concrete Memory Model for CompCert}. In
  \bibinfo{booktitle}{\emph{Proc. 6th Interactive Theorem Proving (ITP'15)}}
  \emph{(\bibinfo{series}{LNCS}, Vol.~\bibinfo{volume}{9236})},
  \bibfield{editor}{\bibinfo{person}{Christian Urban} {and}
  \bibinfo{person}{Xingyuan Zhang}} (Eds.). \bibinfo{publisher}{Springer},
  \bibinfo{address}{Cham}, \bibinfo{pages}{67--83}.
\newblock
\urldef\tempurl%
\url{https://doi.org/10.1007/978-3-319-22102-1\_5}
\showDOI{\tempurl}


\bibitem[Chappe et~al\mbox{.}(2023)]%
        {choicetrees}
\bibfield{author}{\bibinfo{person}{Nicolas Chappe}, \bibinfo{person}{Paul He},
  \bibinfo{person}{Ludovic Henrio}, \bibinfo{person}{Yannick Zakowski}, {and}
  \bibinfo{person}{Steve Zdancewic}.} \bibinfo{year}{2023}\natexlab{}.
\newblock \showarticletitle{Choice Trees: Representing Nondeterministic,
  Recursive, and Impure Programs in Coq}.
\newblock \bibinfo{journal}{\emph{Proc. ACM Program. Lang.}}
  \bibinfo{volume}{7}, \bibinfo{number}{POPL}, Article \bibinfo{articleno}{61}
  (\bibinfo{date}{January} \bibinfo{year}{2023}), \bibinfo{numpages}{31}~pages.
\newblock
\urldef\tempurl%
\url{https://doi.org/10.1145/3571254}
\showDOI{\tempurl}


\bibitem[Gu et~al\mbox{.}(2015)]%
        {dscal15}
\bibfield{author}{\bibinfo{person}{Ronghui Gu},
  \bibinfo{person}{J{\'{e}}r{\'{e}}mie Koenig}, \bibinfo{person}{Tahina
  Ramananandro}, \bibinfo{person}{Zhong Shao},
  \bibinfo{person}{Xiongnan(Newman) Wu}, \bibinfo{person}{Shu-Chun Weng},
  \bibinfo{person}{Haozhong Zhang}, {and} \bibinfo{person}{Yu Guo}.}
  \bibinfo{year}{2015}\natexlab{}.
\newblock \showarticletitle{Deep Specifications and Certified Abstraction
  Layers}. In \bibinfo{booktitle}{\emph{Proc. 42nd ACM Symposium on Principles
  of Programming Languages (POPL'15)}},
  \bibfield{editor}{\bibinfo{person}{Sriram~K. Rajamani} {and}
  \bibinfo{person}{David Walker}} (Eds.). \bibinfo{publisher}{ACM},
  \bibinfo{address}{New York, NY, USA}, \bibinfo{pages}{595--608}.
\newblock
\urldef\tempurl%
\url{https://doi.org/10.1145/2775051.2676975}
\showDOI{\tempurl}


\bibitem[Gu et~al\mbox{.}(2018)]%
        {ccal18}
\bibfield{author}{\bibinfo{person}{Ronghui Gu}, \bibinfo{person}{Zhong Shao},
  \bibinfo{person}{Jieung Kim}, \bibinfo{person}{Xiongnan~(Newman) Wu},
  \bibinfo{person}{J{\'{e}}r{\'{e}}mie Koenig}, \bibinfo{person}{Vilhelm
  Sjober}, \bibinfo{person}{Hao Chen}, \bibinfo{person}{David Costanzo}, {and}
  \bibinfo{person}{Tahnia Ramananandro}.} \bibinfo{year}{2018}\natexlab{}.
\newblock \showarticletitle{Certified Concurrent Abstraction Layers}. In
  \bibinfo{booktitle}{\emph{Proc. 2018 ACM Conference on Programming Language
  Design and Implementation (PLDI'18)}},
  \bibfield{editor}{\bibinfo{person}{Jeffrey~S. Foster} {and}
  \bibinfo{person}{Dan Grossman}} (Eds.). \bibinfo{publisher}{ACM},
  \bibinfo{address}{New York, NY, USA}, \bibinfo{pages}{646--661}.
\newblock
\urldef\tempurl%
\url{https://doi.org/10.1145/3192366.3192381}
\showDOI{\tempurl}


\bibitem[He et~al\mbox{.}(2021)]%
        {paul2021}
\bibfield{author}{\bibinfo{person}{Paul He}, \bibinfo{person}{Eddy Westbrook},
  \bibinfo{person}{Brent Carmer}, \bibinfo{person}{Chris Phifer},
  \bibinfo{person}{Valentin Robert}, \bibinfo{person}{Karl Smeltzer},
  \bibinfo{person}{Andrei \c{S}tef\u{a}nescu}, \bibinfo{person}{Aaron Tomb},
  \bibinfo{person}{Adam Wick}, \bibinfo{person}{Matthew Yacavone}, {and}
  \bibinfo{person}{Steve Zdancewic}.} \bibinfo{year}{2021}\natexlab{}.
\newblock \showarticletitle{A Type System for Extracting Functional
  Specifications from Memory-Safe Imperative Programs}.
\newblock \bibinfo{journal}{\emph{Proc. ACM Program. Lang.}}
  \bibinfo{volume}{5}, \bibinfo{number}{OOPSLA}, Article
  \bibinfo{articleno}{135} (\bibinfo{date}{October} \bibinfo{year}{2021}),
  \bibinfo{numpages}{29}~pages.
\newblock
\urldef\tempurl%
\url{https://doi.org/10.1145/3485512}
\showDOI{\tempurl}


\bibitem[Hur et~al\mbox{.}(2012a)]%
        {hur12}
\bibfield{author}{\bibinfo{person}{Chung{-}Kil Hur}, \bibinfo{person}{Derek
  Dreyer}, \bibinfo{person}{Georg Neis}, {and} \bibinfo{person}{Viktor
  Vafeiadis}.} \bibinfo{year}{2012}\natexlab{a}.
\newblock \showarticletitle{The Marriage of Bisimulations and Kripke Logical
  Relations}. In \bibinfo{booktitle}{\emph{Proc. 39th ACM Symposium on
  Principles of Programming Languages (POPL'12)}},
  \bibfield{editor}{\bibinfo{person}{John Field} {and} \bibinfo{person}{Michael
  Hicks}} (Eds.). \bibinfo{publisher}{{ACM}}, \bibinfo{address}{New York, NY,
  USA}, \bibinfo{pages}{59--72}.
\newblock
\urldef\tempurl%
\url{https://doi.org/10.1145/2103656.2103666}
\showDOI{\tempurl}


\bibitem[Hur et~al\mbox{.}(2012b)]%
        {hur2012tr}
\bibfield{author}{\bibinfo{person}{Chung{-}Kil Hur}, \bibinfo{person}{Georg
  Neis}, \bibinfo{person}{Derek Dreyer}, {and} \bibinfo{person}{Viktor
  Vafeiadis}.} \bibinfo{year}{2012}\natexlab{b}.
\newblock \bibinfo{booktitle}{\emph{The Transitive Composability of Relation
  Transition Systems}}.
\newblock \bibinfo{type}{Technical Report, MPI-SWS-2012-002}.
  \bibinfo{institution}{MPI-SWS}.
\newblock
\urldef\tempurl%
\url{https://www.mpi-sws.org/tr/2012-002.pdf}
\showURL{%
\tempurl}


\bibitem[Jiang et~al\mbox{.}(2019)]%
        {cascompcert}
\bibfield{author}{\bibinfo{person}{Hanru Jiang}, \bibinfo{person}{Hongjin
  Liang}, \bibinfo{person}{Siyang Xiao}, \bibinfo{person}{Junpeng Zha}, {and}
  \bibinfo{person}{Xinyu Feng}.} \bibinfo{year}{2019}\natexlab{}.
\newblock \showarticletitle{Towards Certified Separate Compilation for
  Concurrent Programs}. In \bibinfo{booktitle}{\emph{Proc. 2019 ACM Conference
  on Programming Language Design and Implementation (PLDI'19)}},
  \bibfield{editor}{\bibinfo{person}{Kathryn~S. McKinley} {and}
  \bibinfo{person}{Kathleen Fisher}} (Eds.). \bibinfo{publisher}{{ACM}},
  \bibinfo{address}{New York, NY, USA}, \bibinfo{pages}{111--125}.
\newblock
\urldef\tempurl%
\url{https://doi.org/10.1145/3314221.3314595}
\showDOI{\tempurl}


\bibitem[Kang et~al\mbox{.}(2017)]%
        {promising1}
\bibfield{author}{\bibinfo{person}{Jeehoon Kang}, \bibinfo{person}{Chung{-}Kil
  Hur}, \bibinfo{person}{Ori Lahav}, \bibinfo{person}{Viktor Vafeiadis}, {and}
  \bibinfo{person}{Derek Dreyer}.} \bibinfo{year}{2017}\natexlab{}.
\newblock \showarticletitle{A Promising Semantics for Relaxed-memory
  Concurrency}. In \bibinfo{booktitle}{\emph{Proc. 44th ACM Symposium on
  Principles of Programming Languages (POPL'17)}},
  \bibfield{editor}{\bibinfo{person}{Giuseppe Castagna} {and}
  \bibinfo{person}{Andrew~D. Gordon}} (Eds.). \bibinfo{publisher}{{ACM}},
  \bibinfo{address}{New York, NY, USA}, \bibinfo{pages}{175--189}.
\newblock
\urldef\tempurl%
\url{https://doi.org/10.1145/3009837.3009850}
\showDOI{\tempurl}


\bibitem[Kang et~al\mbox{.}(2015)]%
        {kang2015formal}
\bibfield{author}{\bibinfo{person}{Jeehoon Kang}, \bibinfo{person}{Chung{-}Kil
  Hur}, \bibinfo{person}{William Mansky}, \bibinfo{person}{Dmitri Garbuzov},
  \bibinfo{person}{Steve Zdancewic}, {and} \bibinfo{person}{Viktor Vafeiadis}.}
  \bibinfo{year}{2015}\natexlab{}.
\newblock \showarticletitle{A Formal {C} Memory Model Supporting
  Integer-Pointer Casts}. In \bibinfo{booktitle}{\emph{Proc. 2015 ACM
  Conference on Programming Language Design and Implementation (PLDI'15)}},
  \bibfield{editor}{\bibinfo{person}{David Grove} {and}
  \bibinfo{person}{Stephen~M. Blackburn}} (Eds.). \bibinfo{publisher}{{ACM}},
  \bibinfo{address}{New York, NY, USA}, \bibinfo{pages}{326--335}.
\newblock
\urldef\tempurl%
\url{https://doi.org/10.1145/2737924.2738005}
\showDOI{\tempurl}


\bibitem[Koenig and Shao(2021)]%
        {compcerto}
\bibfield{author}{\bibinfo{person}{J\'{e}r\'{e}mie Koenig} {and}
  \bibinfo{person}{Zhong Shao}.} \bibinfo{year}{2021}\natexlab{}.
\newblock \showarticletitle{CompCertO: Compiling Certified Open C Components}.
  In \bibinfo{booktitle}{\emph{Proc. 2021 ACM Conference on Programming
  Language Design and Implementation (PLDI'21)}}. \bibinfo{publisher}{{ACM}},
  \bibinfo{address}{New York, NY, USA}, \bibinfo{pages}{1095–1109}.
\newblock
\showISBNx{9781450383912}
\urldef\tempurl%
\url{https://doi.org/10.1145/3453483.3454097}
\showDOI{\tempurl}


\bibitem[Krebbers(2016)]%
        {krebbers2016formal}
\bibfield{author}{\bibinfo{person}{Robbert Krebbers}.}
  \bibinfo{year}{2016}\natexlab{}.
\newblock \showarticletitle{A Formal C Memory Model for Separation Logic}.
\newblock \bibinfo{journal}{\emph{J. Autom. Reason.}}  \bibinfo{volume}{57}
  (\bibinfo{year}{2016}), \bibinfo{pages}{319--387}.
\newblock
\urldef\tempurl%
\url{https://doi.org/10.1007/s10817-016-9369-1}
\showDOI{\tempurl}


\bibitem[Lee et~al\mbox{.}(2020)]%
        {promising2}
\bibfield{author}{\bibinfo{person}{Sung-Hwan Lee}, \bibinfo{person}{Minki Cho},
  \bibinfo{person}{Anton Podkopaev}, \bibinfo{person}{Soham Chakraborty},
  \bibinfo{person}{Chung-Kil Hur}, \bibinfo{person}{Ori Lahav}, {and}
  \bibinfo{person}{Viktor Vafeiadis}.} \bibinfo{year}{2020}\natexlab{}.
\newblock \showarticletitle{Promising 2.0: Global Optimizations in Relaxed
  Memory Concurrency}. In \bibinfo{booktitle}{\emph{Proc. 2020 ACM Conference
  on Programming Language Design and Implementation (PLDI'20)}}.
  \bibinfo{publisher}{{ACM}}, \bibinfo{address}{New York, NY, USA},
  \bibinfo{pages}{362–376}.
\newblock
\showISBNx{9781450376136}
\urldef\tempurl%
\url{https://doi.org/10.1145/3385412.3386010}
\showDOI{\tempurl}


\bibitem[Leroy(2023)]%
        {compcert}
\bibfield{author}{\bibinfo{person}{Xavier Leroy}.}
  \bibinfo{year}{2005--2023}\natexlab{}.
\newblock \bibinfo{title}{{The CompCert Verified Compiler}}.
\newblock \bibinfo{howpublished}{\url{https://compcert.org/}}.
\newblock


\bibitem[Leroy et~al\mbox{.}(2012)]%
        {compcert-mem-v2}
\bibfield{author}{\bibinfo{person}{Xavier Leroy}, \bibinfo{person}{Andrew~W.
  Appel}, \bibinfo{person}{Sandrine Blazy}, {and} \bibinfo{person}{Gordon
  Stewart}.} \bibinfo{year}{2012}\natexlab{}.
\newblock \bibinfo{booktitle}{\emph{{The CompCert Memory Model, Version 2}}}.
\newblock \bibinfo{type}{Research Report} RR-7987.
  \bibinfo{institution}{{INRIA}}. \bibinfo{pages}{26} pages.
\newblock
\urldef\tempurl%
\url{https://hal.inria.fr/hal-00703441}
\showURL{%
\tempurl}


\bibitem[Neis et~al\mbox{.}(2015)]%
        {neis15icfp}
\bibfield{author}{\bibinfo{person}{Georg Neis}, \bibinfo{person}{Chung{-}Kil
  Hur}, \bibinfo{person}{Jan{-}Oliver Kaiser}, \bibinfo{person}{Craig
  McLaughlin}, \bibinfo{person}{Derek Dreyer}, {and} \bibinfo{person}{Viktor
  Vafeiadis}.} \bibinfo{year}{2015}\natexlab{}.
\newblock \showarticletitle{Pilsner: a Compositionally Verified Compiler for a
  Higher-Order Imperative Language}. In \bibinfo{booktitle}{\emph{Proc. 2015
  ACM SIGPLAN International Conference on Functional Programming (ICFP'15)}},
  \bibfield{editor}{\bibinfo{person}{Kathleen Fisher} {and}
  \bibinfo{person}{John~H. Reppy}} (Eds.). \bibinfo{publisher}{{ACM}},
  \bibinfo{address}{New York, NY, USA}, \bibinfo{pages}{166--178}.
\newblock
\urldef\tempurl%
\url{https://doi.org/10.1145/2784731.2784764}
\showDOI{\tempurl}


\bibitem[Patterson and Ahmed(2019)]%
        {patterson-icfp-2019}
\bibfield{author}{\bibinfo{person}{Daniel Patterson} {and}
  \bibinfo{person}{Amal Ahmed}.} \bibinfo{year}{2019}\natexlab{}.
\newblock \showarticletitle{The Next 700 Compiler Correctness Theorems
  (Functional Pearl)}.
\newblock \bibinfo{journal}{\emph{Proc. ACM Program. Lang.}}
  \bibinfo{volume}{3}, \bibinfo{number}{ICFP}, Article \bibinfo{articleno}{85}
  (\bibinfo{date}{August} \bibinfo{year}{2019}), \bibinfo{numpages}{29}~pages.
\newblock
\urldef\tempurl%
\url{https://doi.org/10.1145/3341689}
\showDOI{\tempurl}


\bibitem[Patterson et~al\mbox{.}(2017)]%
        {funtal}
\bibfield{author}{\bibinfo{person}{Daniel Patterson}, \bibinfo{person}{Jamie
  Perconti}, \bibinfo{person}{Christos Dimoulas}, {and} \bibinfo{person}{Amal
  Ahmed}.} \bibinfo{year}{2017}\natexlab{}.
\newblock \showarticletitle{FunTAL: Reasonably Mixing a Functional Language
  with Assembly}.
\newblock \bibinfo{journal}{\emph{SIGPLAN Not.}} \bibinfo{volume}{52},
  \bibinfo{number}{6} (\bibinfo{year}{2017}), \bibinfo{pages}{495–509}.
\newblock
\showISSN{0362-1340}
\urldef\tempurl%
\url{https://doi.org/10.1145/3140587.3062347}
\showDOI{\tempurl}


\bibitem[Perconti and Ahmed(2014)]%
        {perconti14}
\bibfield{author}{\bibinfo{person}{James~T. Perconti} {and}
  \bibinfo{person}{Amal Ahmed}.} \bibinfo{year}{2014}\natexlab{}.
\newblock \showarticletitle{Verifying an Open Compiler Using Multi-language
  Semantics}. In \bibinfo{booktitle}{\emph{Proc. 23rd European Symposium on
  Programming (ESOP'14)}} \emph{(\bibinfo{series}{LNCS},
  Vol.~\bibinfo{volume}{8410})}, \bibfield{editor}{\bibinfo{person}{Zhong
  Shao}} (Ed.). \bibinfo{publisher}{Springer}, \bibinfo{address}{Cham},
  \bibinfo{pages}{128--148}.
\newblock
\urldef\tempurl%
\url{https://doi.org/10.1007/978-3-642-54833-8\_8}
\showDOI{\tempurl}


\bibitem[Sammler et~al\mbox{.}(2023)]%
        {sammler2023dimsum}
\bibfield{author}{\bibinfo{person}{Michael Sammler}, \bibinfo{person}{Simon
  Spies}, \bibinfo{person}{Youngju Song}, \bibinfo{person}{Emanuele D'Osualdo},
  \bibinfo{person}{Robbert Krebbers}, \bibinfo{person}{Deepak Garg}, {and}
  \bibinfo{person}{Derek Dreyer}.} \bibinfo{year}{2023}\natexlab{}.
\newblock \showarticletitle{DimSum: A Decentralized Approach to Multi-Language
  Semantics and Verification}.
\newblock \bibinfo{journal}{\emph{Proc. ACM Program. Lang.}}
  \bibinfo{volume}{7}, \bibinfo{number}{POPL}, Article \bibinfo{articleno}{27}
  (\bibinfo{date}{January} \bibinfo{year}{2023}), \bibinfo{numpages}{31}~pages.
\newblock
\urldef\tempurl%
\url{https://doi.org/10.1145/3571220}
\showDOI{\tempurl}


\bibitem[Scherer et~al\mbox{.}(2018)]%
        {scherer18}
\bibfield{author}{\bibinfo{person}{Gabriel Scherer}, \bibinfo{person}{Max New},
  \bibinfo{person}{Nick Rioux}, {and} \bibinfo{person}{Amal Ahmed}.}
  \bibinfo{year}{2018}\natexlab{}.
\newblock \showarticletitle{Fabous Interoperability for ML and a Linear
  Language}. In \bibinfo{booktitle}{\emph{Foundations of Software Science and
  Computation Structures}}, \bibfield{editor}{\bibinfo{person}{Christel Baier}
  {and} \bibinfo{person}{Ugo Dal~Lago}} (Eds.). \bibinfo{publisher}{Springer},
  \bibinfo{address}{Cham}, \bibinfo{pages}{146--162}.
\newblock
\showISBNx{978-3-319-89366-2}
\urldef\tempurl%
\url{https://doi.org/10.1007/978-3-319-89366-2_8}
\showDOI{\tempurl}


\bibitem[Song et~al\mbox{.}(2020)]%
        {compcertm}
\bibfield{author}{\bibinfo{person}{Youngju Song}, \bibinfo{person}{Minki Cho},
  \bibinfo{person}{Dongjoo Kim}, \bibinfo{person}{Yonghyun Kim},
  \bibinfo{person}{Jeehoon Kang}, {and} \bibinfo{person}{Chung-Kil Hur}.}
  \bibinfo{year}{2020}\natexlab{}.
\newblock \showarticletitle{CompCertM: CompCert with C-Assembly Linking and
  Lightweight Modular Verification}.
\newblock \bibinfo{journal}{\emph{Proc. ACM Program. Lang.}}
  \bibinfo{volume}{4}, \bibinfo{number}{POPL}, Article \bibinfo{articleno}{23}
  (\bibinfo{date}{January} \bibinfo{year}{2020}), \bibinfo{numpages}{31}~pages.
\newblock
\urldef\tempurl%
\url{https://doi.org/10.1145/3371091}
\showDOI{\tempurl}


\bibitem[Song et~al\mbox{.}(2023)]%
        {song2023conditional}
\bibfield{author}{\bibinfo{person}{Youngju Song}, \bibinfo{person}{Minki Cho},
  \bibinfo{person}{Dongjae Lee}, \bibinfo{person}{Chung-Kil Hur},
  \bibinfo{person}{Michael Sammler}, {and} \bibinfo{person}{Derek Dreyer}.}
  \bibinfo{year}{2023}\natexlab{}.
\newblock \showarticletitle{Conditional Contextual Refinement}.
\newblock \bibinfo{journal}{\emph{Proc. ACM Program. Lang.}}
  \bibinfo{volume}{7}, \bibinfo{number}{POPL}, Article \bibinfo{articleno}{39}
  (\bibinfo{date}{January} \bibinfo{year}{2023}), \bibinfo{numpages}{31}~pages.
\newblock
\urldef\tempurl%
\url{https://doi.org/10.1145/3571232}
\showDOI{\tempurl}


\bibitem[Stewart et~al\mbox{.}(2015)]%
        {stewart15}
\bibfield{author}{\bibinfo{person}{Gordon Stewart}, \bibinfo{person}{Lennart
  Beringer}, \bibinfo{person}{Santiago Cuellar}, {and}
  \bibinfo{person}{Andrew~W. Appel}.} \bibinfo{year}{2015}\natexlab{}.
\newblock \showarticletitle{{Compositional CompCert}}. In
  \bibinfo{booktitle}{\emph{Proc. 42nd ACM Symposium on Principles of
  Programming Languages (POPL'15)}}. \bibinfo{publisher}{ACM},
  \bibinfo{address}{New York, NY, USA}, \bibinfo{pages}{275--287}.
\newblock
\urldef\tempurl%
\url{https://doi.org/10.1145/2676726.2676985}
\showDOI{\tempurl}


\bibitem[Wang et~al\mbox{.}(2019)]%
        {wang2019}
\bibfield{author}{\bibinfo{person}{Yuting Wang}, \bibinfo{person}{Pierre
  Wilke}, {and} \bibinfo{person}{Zhong Shao}.} \bibinfo{year}{2019}\natexlab{}.
\newblock \showarticletitle{An Abstract Stack Based Approach to Verified
  Compositional Compilation to Machine Code}.
\newblock \bibinfo{journal}{\emph{Proc. ACM Program. Lang.}}
  \bibinfo{volume}{3}, \bibinfo{number}{POPL}, Article \bibinfo{articleno}{62}
  (\bibinfo{date}{January} \bibinfo{year}{2019}), \bibinfo{numpages}{30}~pages.
\newblock
\showISSN{2475-1421}
\urldef\tempurl%
\url{https://doi.org/10.1145/3290375}
\showDOI{\tempurl}


\bibitem[Wang et~al\mbox{.}(2022)]%
        {wang-popl-2022}
\bibfield{author}{\bibinfo{person}{Yuting Wang}, \bibinfo{person}{Ling Zhang},
  \bibinfo{person}{Zhong Shao}, {and} \bibinfo{person}{J\'{e}r\'{e}mie
  Koenig}.} \bibinfo{year}{2022}\natexlab{}.
\newblock \showarticletitle{Verified Compilation of C Programs with a Nominal
  Memory Model}.
\newblock \bibinfo{journal}{\emph{Proc. ACM Program. Lang.}}
  \bibinfo{volume}{6}, \bibinfo{number}{POPL}, Article \bibinfo{articleno}{25}
  (\bibinfo{date}{January} \bibinfo{year}{2022}), \bibinfo{numpages}{31}~pages.
\newblock
\urldef\tempurl%
\url{https://doi.org/10.1145/3498686}
\showDOI{\tempurl}


\bibitem[Xia et~al\mbox{.}(2019)]%
        {xia2019interaction}
\bibfield{author}{\bibinfo{person}{Li-yao Xia}, \bibinfo{person}{Yannick
  Zakowski}, \bibinfo{person}{Paul He}, \bibinfo{person}{Chung-Kil Hur},
  \bibinfo{person}{Gregory Malecha}, \bibinfo{person}{Benjamin~C. Pierce},
  {and} \bibinfo{person}{Steve Zdancewic}.} \bibinfo{year}{2019}\natexlab{}.
\newblock \showarticletitle{Interaction Trees: Representing Recursive and
  Impure Programs in Coq}.
\newblock \bibinfo{journal}{\emph{Proc. ACM Program. Lang.}}
  \bibinfo{volume}{4}, \bibinfo{number}{POPL}, Article \bibinfo{articleno}{51}
  (\bibinfo{date}{January} \bibinfo{year}{2019}), \bibinfo{numpages}{32}~pages.
\newblock
\urldef\tempurl%
\url{https://doi.org/10.1145/3371119}
\showDOI{\tempurl}


\bibitem[Zha et~al\mbox{.}(2022)]%
        {zha2022}
\bibfield{author}{\bibinfo{person}{Junpeng Zha}, \bibinfo{person}{Hongjin
  Liang}, {and} \bibinfo{person}{Xinyu Feng}.} \bibinfo{year}{2022}\natexlab{}.
\newblock \showarticletitle{Verifying Optimizations of Concurrent Programs in
  the Promising Semantics}. In \bibinfo{booktitle}{\emph{Proc. 2021 ACM
  Conference on Programming Language Design and Implementation (PLDI'22)}}.
  \bibinfo{publisher}{{ACM}}, \bibinfo{address}{New York, NY, USA},
  \bibinfo{pages}{903–917}.
\newblock
\showISBNx{9781450392655}
\urldef\tempurl%
\url{https://doi.org/10.1145/3519939.3523734}
\showDOI{\tempurl}


\bibitem[Zhang et~al\mbox{.}(2023a)]%
        {direct-refinement-artifact}
\bibfield{author}{\bibinfo{person}{Ling Zhang}, \bibinfo{person}{Yuting Wang},
  \bibinfo{person}{Jinhua Wu}, \bibinfo{person}{J{\'{e}}r{\'{e}}mie Koenig},
  {and} \bibinfo{person}{Zhong Shao}.} \bibinfo{year}{2023}\natexlab{a}.
\newblock \bibinfo{title}{Fully Composable and Adequate Verified Compilation
  with Direct Refinements between Open Modules (Artifact)}.
\newblock
\newblock
\urldef\tempurl%
\url{https://doi.org/10.5281/zenodo.10036618}
\showURL{%
\tempurl}


\bibitem[Zhang et~al\mbox{.}(2023b)]%
        {direct-refinement-tr}
\bibfield{author}{\bibinfo{person}{Ling Zhang}, \bibinfo{person}{Yuting Wang},
  \bibinfo{person}{Jinhua Wu}, \bibinfo{person}{J{\'{e}}r{\'{e}}mie Koenig},
  {and} \bibinfo{person}{Zhong Shao}.} \bibinfo{year}{2023}\natexlab{b}.
\newblock \bibinfo{title}{Fully Composable and Adequate Verified Compilation
  with Direct Refinements between Open Modules (Technical Report)}.
\newblock
\newblock
\urldef\tempurl%
\url{https://doi.org/10.48550/arXiv.2302.12990}
\showURL{%
\tempurl}


\end{thebibliography}

\appendix
\section{Transitivity of \kinjp}
\label{sec:injp-trans-proof}

\subsection{More complete definitions of memory injection and $\kinjp$ accessibility}

We have used a simplified version of definitions of $\kwd{perm}$,
$\hookrightarrow_m$ and $\leadsto_{\kinjp}$ in Sec. \ref{sec:injp}.
To present a more detailed proof of the KMR with memory protection
($\kinjp$), we present full definitions of $\kwd{perm}$ and
$\leadsto_{\kinjp}$. We also present a more complete definition
of $\hookrightarrow_m$ which is still not 100\% complete because
we ignore two properties for simplicity. They are about alignment and
range of size $\delta$ in mapping $j(b) = \some{(b',\delta)}$.
They are not essential for this proof as preconditions and can be
proved similarly as other properties of $\hookrightarrow_m$.
Readers interested in these details can find them in our artifact.

By the definition of CompCert memory model, a memory cell has both
maximum and current permissions such that $\permcur{m}{p} \subseteq
\permmax{m}{p}$. During the execution of a program, the current permission
of a memory cell may be lowered or raised by an external call.
However, the maximum permission can only decrease in both internal
and external calls. This invariant was defined in CompCert as:
\begin{tabbing}
  \quad\=\quad\=\kill
  \> $\mpd{m_1}{m_2} \iff$\\
  \>\> $\forall \app b \app o \app p, \app b \in \validblock{m_1} \imply (b,o) \in \permmax{m_2}{p} \imply (b,o) \in \permmax{m_1}{p}$
\end{tabbing}

\begin{definition} Definition of memory injection $\hookrightarrow_m$.
  \small
  \begin{tabbing}
    \quad\=\quad\=\;\;\=\kill
    \> $\minj{j}{m_1}{m_2} := \{|$\\
    \>\> \cmnt{(* Preservation of permission under injection *)}\\
    \>\> $(1) \app \forall\app b_1\app b_2\app o_1\app o_2 \app k\app p,\app
          j(b_1) = \some{(b_2,o_2 - o_1)} \imply
          (b_1, o_1) \in \permk{k}{m_1}{p} \imply
          (b_2, o_2) \in \permk{k}{m_2}{p}$\\
    \>\> \cmnt{(* Preservation of memory values for currently readable cells under injection *)}\\
    \>\> $(2) \app \forall\app b_1\app b_2\app o_1\app o_2,\app
          j(b_1) = \some{(b_2,o_2 - o_1)} \imply
          (b_1,o_1) \in \permcur{m_1}{\kreadable}$\\ 
    \>\>\> $\imply \vinj{j}{m_1[b_1,o_1]}{m_2[b_2,o_2]}$\\
    \>\> \cmnt{(* Invalid source blocks must be unmapped *)}\\
    \>\> $(3) \app \forall\app b_1,\app b_1 \notin \validblock{m_1} \imply j(b_1) = \none$\\
    \>\> \cmnt{(* The range of $j$ must only contain valid blocks *)}\\
    \>\> $(4) \app \forall\app b_1\app b_2\app \delta,\app j(b_1) = \some{(b_2,\delta)} \imply b_2 \in \validblock{m_2}$\\
    \>\> \cmnt{(* Two disjoint source cells with non-empty permission} \\
    \>\>\> \cmnt{ do not overlap with each other after injection *)}\\
    \>\>  $(5) \app \forall\app b_1\app b_2\app o_1\app o_2\app b_1'\app b_2'\app o_1'\app o_2',\app
           b_1 \neq b_1' \imply
           j(b_1) = \some{(b_2, o_2 - o_1)} \imply
           j(b_1') = \some{(b_2', o_2' - o_1')} \imply $\\
    \>\>\> $(b_1, o_1) \in \permmax{m_1}{\knonempty} \imply
           (b_1', o_1') \in \permmax{m_1}{\knonempty}
           \imply b_2 \neq b_2' \lor o_2 \neq o_2'$\\
    \>\> \cmnt{(* Given a target cell, its corresponding source cell either }\\
    \>\>\> \cmnt{have the same permission or does not have any permission *)}\\
    \>\> $(6) \app \forall\app b_1\app o_1\app b_2\app o_2\app k\app p,\app
      j(b_1) = \some{(b_2, o_2 - o_1)} \imply
      (b_2, o_2) \in \permk{k}{m_2}{p}$\\
    \>\>\> $\imply (b_1,o_1) \in \permk{k}{m_1}{p} \lor (b_1,o_1) \not\in \permmax{m_1}{\knonempty}$
  \end{tabbing}
  \label{def:meminj}
\end{definition}
\begin{definition}\label{def:mem-acc}
  Definition of memory accessibility $\kmacc$.\\
  \begin{tabbing}
      \quad\=$\mpd{m}{m'}$\;\=$\iff\;$\=\kill
      \>$\roacc{m}{m'}$\>$ \iff \; \forall (b,o) \in m, (b,o) \notin \permmax{m}{\kwritable} \imply m'[b,o] = v$ \\
      \>\>\> $\imply (b,o) \in \permcur{m'}{\kreadable} \imply (m[b,o] = v  \; \land$ \\
      \>\>\> $(b,o) \in \permcur{m}{\kreadable})$\\
      \>$\macc{m}{m'}$\>$ \iff \;\kwd{validblock}(m) \subseteq \kwd{validblock}(m') \; \land$\\
      \>\>\>$ \; \mpd{m}{m'} \land \; \roacc{m}{m'} $\\
    \end{tabbing}
\end{definition}

\noindent
For the complete definition of $\leadsto_{\kinjp}$,
we further define the separation property for injection as:   
\begin{tabbing}
  \quad\=\quad\=\kill
  \> $\injsep{j}{j'}{m_1}{m_2} \iff$\\
  \>\> $\forall \app b_1 \app b_2 \app \delta, \app j(b_1) = \none \imply j'(b_1) = \some{(b_2,\delta)} \imply b_1 \notin \validblock{m_1} \land b_2 \notin \validblock{m_2}$
\end{tabbing}
This invariant states that when we start from $\minj{j}{m_1}{m_2}$, after executing on source and target semantics, the future injection $j'$ only increases from $j$ by relating newly allocated blocks. Note that we write
$b \in m$ for $b \in \kwd{validblock}(m)$.

\begin{definition}\label{def:injpacc}
  Accessibility relation of $\kinjp$
    \begin{tabbing}
    \quad\=$\injpacc{(j, m_1, m_2)}{(j', m_1', m_2')} \; \iff \;$\=\kill
    \>$\injpacc{(j, m_1, m_2)}{(j', m_1', m_2')} \; \iff \;j \subseteq j' \land \unmapped{j} \subseteq \unchangedon{m_1}{m_1'}$\\
    \>\>$\land\; \outofreach{j}{m_1} \subseteq \unchangedon{m_2}{m_2'}$ \\
    \>\>$\land\; \macc{m_1}{m_1'} \land \macc{m_2}{m_2'}$\\
    \>\>$\land\; \injsep{j}{j'}{m_1}{m_2}.$
  \end{tabbing}
\end{definition}

\subsection{Auxiliary Properties}
In this section we present several lemmas about properties of memory injection and $\kinjp$ accessibility. These lemmas are used in the proof of $\kinjp$ refinement.

Firstly, the memory injections are composable.
\begin{lemma}\label{lem:inj-trans}
  Given $\minj{j_{12}}{m_1}{m_2}$ and $\minj{j_{23}}{m_2}{m_3}$, we have
  \[
    \minj{j_{23} \cdot j_{12}}{m_1}{m_3}
  \]
\end{lemma}

This property is proved and used in CompCert, we do not repeat the proof here.

  \begin{lemma}\label{lem:readable-midvalue}
    Given 
    $\minj{j_{23}\cdot j_{12}}{m_1}{m_3}$, $(b_1,o_1) \in \permcur{m_1}{\kreadable}$ and $j_{23}\cdot j_{12}(b_1) = \some{(b_3,o_3-o_1)}$, then
  \[
    \exists v_2, \vinj{j_{12}}{m_1[b_1,o_1]}{v_2} \land \vinj{j_{23}}{v_2}{m_3[b_3,o_3]}.
  \]
\end{lemma}
  Note that $j_{23}\cdot j_{12}(b_1) = \some{(b_3,o_3-o_1)}$ iff $\exists\app b_2\app o_2, j_{12}(b_1) = \some{(b_2,o_2 - o_1)} \land j_{23}(b_2) = \some{(b_3,o_3-o_2)}$.

\begin{proof}
  According to property (2) in ~\defref{def:meminj}, we know that $\vinj{j_{23} \cdot j_{12}}{m_1[b_1,o_1]}{m_3[b_3,o_3]}$. We divide the value $m_1[b_1,o_1]$ into:
  \begin{itemize}
    \item
      If $m_1[b_1,o_1] = \kVundef$, we take $v_2 = \kVundef$. Then $\vinj{j_{12}}{\kVundef}{\kVundef} \land \vinj{j_{23}}{\kVundef}{m_3[b_3,o_3]}$ trivially holds.
    \item
      If $m_1[b_1,o_1]$ is a concrete value, we take $v_2 = m_1[b_1,o_1]$. In such case we have $m_1[b_1,o_1] = v_2 = m_3[b_3,o_3]$.
    \item
      If $m_1[b_1,o_1] = \vptr{b_1'}{o_1'}$, we can derive that $\vinj{j_{23}\cdot j_{12}}{\vptr{b_1'}{o_1'}}{m_3[b_3,o_3]}$ implies $\exists b_3'\ o_3', s.t. m_3[b_3,o_3] = \vptr{b_3'}{o_3'}$ and $ j_{23} \cdot j_{12}(b_1') = \some{(b_3',o_3' - o_1')}$. Therefore 
  \[
     \exists b_2'\ o_2', j_{12}(b_1') = \some{(b_2',o_2' - o_1')} \land j_{23}(b_2') = \some{(b_3',o_3' - o_2')}
   \]
   We take $v_2 = \vptr{b_2'}{o_2'}$ and $\vinj{j_{12}}{m_1[b_1,o_1]}{v_2} \land \vinj{j_{23}}{v_2}{m_3[b_3,o_3]}$ can be derived from the formula above.
   \end{itemize}
 \end{proof}

 \begin{lemma}\label{lem:outofreach-reverse}
  Given $\minj{j_{12}}{m_1}{m_2}$, $\minj{j_{23}}{m_2}{m_3}$ and $j_{23}(b_2) = \some{(b_3,o_3 - o_2)}$. If $(b_2,o_2) \in \outofreach{j_{12}}{m_1}$ and $(b_2,o_2) \in \permmax{m_2}{\knonempty}$, then
  \[
    (b_3,o_3) \in \outofreach{j_{23}\cdot j_{12}}{m_1}
  \]
  \begin{proof}
    According to the definition of $\koutofreach$, If $j_{12}(b_1) = \some{(b_2',o_2' - o_1)}$ and $j_{23}(b_2') = \some{(b_3,o_3 - o_1)}$, we need to prove that $(b_1,o_1) \notin \permmax{m_1}{\knonempty}$. If $b_2 = b_2'$, from $(b_2,o_2) \in \outofreach{j_{12}}{m_1}$ we can directly prove $(b_1,o_1) \notin \permmax{m_1}{\knonempty}$.

    If $b_2 \neq b_2'$, we assume that $(b_1,o_1) \in \permmax{m_1}{\knonempty}$, by property (1) of $\minj{j_{12}}{m_1}{m_2}$ we get $(b_2',o_2') \in \permmax{m_2}{\knonempty}$. Now $(b_2,o_2)$ and $(b_2',o_2')$ are two different positions in $m_2$ which are mapped to the same position $(m_3,o_3)$ in $m_3$. This scenario is prohibited by the non-overlapping property (5) of $\minj{j_{23}}{m_2}{m_3}$. So $(b_1,o_1) \notin \permmax{m_1}{\knonempty}$.
    \end{proof}
\end{lemma}
\subsection{Proof of Lemma \ref{lem:injp-refine-injp-comp}}
Based on definitions and lemmas before, we prove Lemma \ref{lem:injp-refine-injp-comp} in this section:
\begin{tabbing}
  \quad\=\quad\=\quad\=$\exists m_2'\app j_{12}'\app j_{23}',$\=\kill
  \>$\forall j_{12}\app j_{23}\app m_1\app m_2\app m_3,\app \minj{j_{12}}{m_1}{m_2}
  \imply \minj{j_{23}}{m_2}{m_3} \imply \exists j_{13},\app \minj{j_{13}}{m_1}{m_3} \app \land$\\
  \>\>$\forall m_1'\app m_3'\app j_{13}',\app \injpacc{(j_{13}, m_1, m_3)}{(j_{13}', m_1', m_3')} \imply \minj{j_{13}'}{m_1'}{m_3'}  \imply$\\
  \>\>\>$\exists m_2'\app j_{12}'\app j_{23}', \injpacc{(j_{12},m_1,m_2)}{(j_{12}',m_1',m_2')} \land \minj{j_{12}'}{m_1'}{m_2'}$\\
  \>\>\>\>$\land \injpacc{(j_{23},m_2,m_3)}{(j_{23}',m_2',m_3')} \land \minj{j_{23}'}{m_2'}{m_3'}.$
\end{tabbing}
  Given $\minj{j_{12}}{m_1}{m_2}$ and $\minj{j_{23}}{m_2}{m_3}$. We take $j_{13} = j_{23} \cdot j_{12}$, from Lemma \ref{lem:inj-trans} we can prove $\minj{j_{13}}{m_1}{m_3}$. After the external call, given $\injpacc{(j_{13}, m_1, m_3)}{(j_{13}', m_1', m_3')}$ and $\minj{j_{13}'}{m_1'}{m_3'}$.

  We present the construction and properties of $j_{12}', j_{23}'$ and $m_2'$ in Sec. \ref{subsubsec:construction}. Then the proof reduce to prove $\minj{j_{12}'}{m_1'}{m_2'}$, $\minj{j_{23}'}{m_2'}{m_3'}$, $\injpacc{(j_{12},m_1,m_2)}{(j_{12}',m_1',m_2')}$ and $\injpacc{(j_{23},m_2,m_3) \allowbreak }{(j_{23}',m_2',m_3')}$, they are proved in Sec. \ref{subsubsec:detail-proof}

  \subsubsection{Construction and properties of $j_{12}'$, $j_{23}'$ and $m_2'$}\label{subsubsec:construction}

\begin{definition}\label{def:construction}
We construct the memory state $m_2'$ by the following three steps, $j_{12}'$ and $j_{23}'$ are constructed in step (1).

\begin{enumerate}
\item We first extend $m_2$ by allocating new blocks, at the same time we extend $j_{12},j_{23}$ to get $j_{12}'$ and $j_{23}'$ such that $j_{13}' = j_{23}' \cdot j_{12}'$.
  Specifically, for each new block $b_1$ in $m_1'$ relative to $m_1$ which is mapped by $j_{13}'$ as $j_{13}'(b_1) = \some{(b_3,\delta)}$, we allocate a new memory block $b_2$ from $m_2$ and add new mappings $(b_1,(b_2,0))$ and $(b_2,(b_3,\delta))$ to $j_{12}$ and $j_{23}$, respectively.
\item
  We then copy the contents of new blocks in $m_1'$ into corresponding new blocks in $m_2'$ as follows.
  For each \emph{mapped} new block $b_1$ in $m_1'$ where $j_{12}'(b_1) = \some{(b_2,0)}$, we enumerate all positions $(b_1,o_1) \in \permmax{m_1'}{\knonempty}$ and copy the permission of $(b_1,o_1)$ in $m_1'$ to $(b_2,o_1)$ in $m_2'$.
  If $(b_1,o_1) \in \permcur{m_1'}{\kreadable}$, we further set $\mcontents{m_2'}{b_2}{o_1}$ to $v_2$ where $\minj{j_{12'}}{\mcontents{m_1'}{b_1}{o_1}}{v_2}$. The existence of $v_2$ here is provided by Lemma \ref{lem:readable-midvalue} with preconditions $\minj{j_{13}'}{m_1'}{m_3'}$,  $(b_1,o_1) \in \permcur{m_1'}{\kreadable}$ and $j_{13}'(b_1) = \some{(b_3,\delta)}$ (because $b_1$ is a new block chosen in step (1)).
\item Finally, we update the old blocks of $m_2$. If a position $(b_2,o_2) \in \pubtgtmem{j_{12}}{m_1} \cap \pubsrcmem{j_{23}}$, the permission and value of this position in $m_2'$ should comes from the corresponding position $(b_1,o_1)$ in $m_1'$ as depicted in ~\figref{fig:inj-after}. Note that the values are changed only if the position is not read-only in $m_2$. Other positions just remain unchanged from $m_2$ to $m_2'$.
  To complete the construction, we have to enumerate the set $\pubtgtmem{j_{12}}{m_1}\cap \pubsrcmem{j_{23}}$. We state that
  \[
    \pubtgtmem{j_{12}}{m_1} \subseteq \permmax{m_2}{\knonempty}
  \]
  where $\permmax{m_2}{\knonempty}$ is enumerable. Note that $(b_2,o_2) \in \pubtgtmem{j_{12}}{m_1} \allowbreak \iff (b_2,o_2) \notin \outofreach{j_{12}}{m_1}$ by definition. If $(b_2,o_2) \in \pubtgtmem{j_{12}}{m_1}$, then there exists $(b_1,o_1) \in \permmax{m_1}{\knonempty} $ such that $j_{12}(b_1) = \some{(b_2,o_2 - o_1)} $. The property (1) of $\minj{j_{12}}{m_1}{m_2}$ ensures that $(b_2,o_2) \in \permmax{m_2}{\knonempty}$.
  
  The concrete algorithm can be described as follows. For $(b_2,o_2) \in \permmax{m_2}{\kwd{None-}\allowbreak \kwd{empty}}$, we can enumerate $\permmax{m_1}{\knonempty}$ to find whether there exists a corresponding position  $(b_1,o_1) \in \permmax{m_1}{\knonempty} $ such that $j_{12}(b_1) = \some{(b_2,o_2 - o_1)}$. Note that the property (3) of $\minj{j_{12}}{m_1}{m_2}$ ensures that we cannot find more than one of such position. If there exists such $(b_1,o_1)$ and $j_{23}(b_2) \neq \some{(b_3,o_3)}$, We copy the permission of position $(b_1,o_1)$ in $m_1'$ to $(b_2,o_2)$. If $(b_1,o_1) \in \permcur{m_1'}{\kreadable}$ and $(b_2,o_2) \in \permmax{m_2}{\kwritable} $, we further set $\mcontents{m_2'}{b_2}{o_2}$ to $v_2$ where $\vinj{j_{12'}}{m_1'[b_1,o_1]}{v_2}$.
  

\end{enumerate}
\end{definition}

We present several lemmas about $j_{12}', j_{23}'$ and $m_2'$ according to Definition \ref{def:construction} as follows.
 \begin{lemma}\label{lem:constr-inj}
  
  \[
    (1) j_{12} \subseteq j_{12}' \; (2) j_{23} \subseteq j_{23}' \;
    (3) \injsep{j_{12}}{j_{12}'}{m_1}{m_2} \; (4) \injsep{j_{23}}{j_{23}'}{m_2}{m_3}
  \]
\end{lemma}
\begin{proof}
  Directly from the construction step (1)
\end{proof}

 \begin{lemma}\label{lem:m2-unchangedon}
  \[
    (1) \outofreach{j_{12}}{m_1} \subseteq \unchangedon{m_2}{m_2'}
    \; (2) \unmapped{j_{23}} \subseteq \unchangedon{m_2}{m_2'}
  \]
\end{lemma}
\begin{proof}
  For each changed position $(b_2,o_2)$ from $m_2$ to $m_2'$ in step (3), we enforce that $\exists b_1, j_{12}(b_1) = \some{(b_2,o_2 - o_1)} \land (b_1,o_1) \in \permmax{m_1}{\knonempty}$ and $(b_2,o_2) \notin \unmapped{j_{23}}$. Thus, if $(b_2,o_2) \in \outofreach{j_{12}}{m_1}$ or $(b_2,o_2) \in \unmapped{j_{23}}$, then $(b_2,o_2) \in \kwd{unchanged-}\allowbreak \kwd{on} (m_2,m_2')$.

\end{proof}

 \begin{lemma}\label{lem:m2-mpd}
  \[
    \mpd{m_2}{m_2'}
  \]
\end{lemma}
\begin{proof}
  For unchanged position $(b_2,o_2)$ in $m_2$, we trivially have $(b_2,o_2) \in \permmax{m_2'}{p} \iff (b_2,o_2)\in \permmax{m_2}{p}$. If $(b_2,o_2)$ is changed in step (3), then the permission of $(b_2,o_2)$ in $m_2'$ is copied from some corresponding position $(b_1,o_1)$ in $m_1'$($j_{12}(b_1) = \some{(b_2,o_2 - o_1)}$). Given $(b_2,o_2) \in \permmax{m_2'}{p}$, we get $(b_1,o_1) \in \permmax{m_1'}{p}$. Form $\mpd{m_1}{m_1'}$ we can further derive that $(b_1,o_1) \in \permmax{m_1}{p}$. Finally, by property (1) of $\minj{j_{12}}{m_1}{m_2}$ we can conclude that $(b_2,o_2) \in \permmax{m_2}{p}$.
\end{proof}

\begin{lemma}\label{lem:m2-rounc}
  \[ \roacc{m_2}{m_2'} \]
\end{lemma}
\begin{proof}
  For each position $(b_2,o_2)$ which has changed value from $m_2$ to $m_2'$ in step (3). We enforce that it is not read-only in $m_2$.
\end{proof}

\begin{lemma}\label{lem:m2-macc}
  \[ \macc{m_2}{m_2'}\]
\end{lemma}
\begin{proof}
  From step(1) we have $\validblock{m_2} \subseteq \validblock{m_2'}$.
  Together with Lemma \ref{lem:m2-mpd} and ~\lemref{lem:m2-rounc} we can
  derive this lemma.
\end{proof}

\subsubsection{Proof of remaining formulas}
\label{subsubsec:detail-proof}
 Recall that we are still proving Lemma \ref{lem:injp-refine-injp-comp}, we have constructed $j_{12}', j_{23}'$ and $m_2'$. Based on the construction and properties of them presented above, we present complete proofs of last four formulas separately in this section.

 \begin{lemma}\label{lem:constr-inj1}
   $\minj{j_{12}'}{m_1'}{m_2'}$
  \end{lemma}
\begin{proof}
We check the properties in ~\defref{def:meminj} as follows:
\begin{enumerate}
\item Given $ j_{12}'(b_1) = \some{(b_2,o_2 - o_1)} \land (b_1, o_1) \in \permk{k}{m_1'}{p}$. We prove $(b_2, o_2) \in \permk{k}{m_2'}{p}$ by cases of $j_{12}(b_1)$. Note that $j_{12}(b_1)$ is either $\none$ or the same as $j'_{12}(b_1)$ because of $j_{12} \subseteq j_{12}'$.
  \begin{itemize}
  \item If $j_{12} (b_1)= \none$, the mapping $j_{12}'(b_1) = \some{(b_2,o_2-o_1)}$ is added in step (1). As a result, we know $\exists b_3\app \delta, j_{13}'(b_1) = \some{b_3,\delta}$. Since $ \permk{k}{m_1'}{p} \subseteq \permmax{m_1'}{\knonempty}$, we know $(b_1,o_1) \in \permmax{m_1'}{\knonempty}$ and the permission of $(b_1,o_1)$ in $m_1'$ is copied to $(m_2,o_2)$ in $m_2'$ in step (2). Therefore $(b_2,o_2) \in \permk{k}{m_2'}{p}$.
  \item If $j_{12}(b_1) = \some{(b_2,o_2)}$, we further divide whether $(b_2,o_2)$ is a public position by $j_{23}(b_2)$
    \begin{itemize}
    \item If $j_{23}(b_2) = \none$, i.e. $(b_2,o_2) \in \unmapped{j_{23}}$, According to Lemma \ref{lem:m2-unchangedon}, we know $(b_2,o_2) \in \unchangedon{m_2}{m_2'}$. At the same time, we also get $(b_1,o_1) \in \unmapped{j_{13}}$ because of $j_{13} = j_{23} \cdot j_{12}$.  Together with $\injpacc{(j_{13},m_1,m_3)}{(j_{13}',m_1',m_3')}$, we can conclude that $(b_1,o_1) \in \kwd{unchanged-}\allowbreak \kwd{on} (m_1,m_1')$.
      
     Therefore, we get $(b_1,o_1) \in \permk{k}{m_1}{p}$. Using property (1) of $\minj{j_{12}}{m_1}{m_2}$ we get $(b_2,o_2) \in \permk{k}{m_2}{p}$. Since $(b_2,o_2)$ is also unchanged between $m_2$ and $m_2'$,  $(b_2, o_2) \in \permk{k}{m_2'}{p}$.
    \item If $j_{23}(b_2) = \some{(b_3,o_3 - o_2)}$, the permission of $(b_2,o_2)$ in $m_2'$ is set as the same as $(b_1,o_1)$ in $m_1'$ in step (3). So $(b_2,o_2) \in \permk{k}{m_2'}{p}$ holds trivially.
    \end{itemize}
  \end{itemize}
  
\item Given $j_{12}'(b_1) = \some{(b_2,o_2 - o_1)} \land (b_1, o_1) \in \permcur{m_1'}{\kreadable}$, following the method in (1) we can prove $\vinj{j_{12}'}{m_1'[b_1,o_1]}{m_2'[b_2,o_2]}$. Note that if $(b_2,o_2)$ is read-only in $m_2$, from property (6) of $\minj{j_{12}}{m_1}{m_2}$ we can derive that $(b_1,o_1)$ is also read-only in $m_1$. Thus the values related by $j$ are both unchanged in $m_1'$ and $m_2'$ thus can be related by $j'$ where $j \subseteq j'$.
\item Given $b_1 \notin \validblock{m_1'}$, we know $b_1 \notin \validblock{m_1}$, therefore $j_{12}(b_1) = \none$. Since $b_1$ cannot be added to $j_{12}'$ in step (1), we can conclude that $j_{12}'(b_1) = \none$.
\item 
  Given $j_{12}'(b_1) = \some{(b_2,\delta)}$, It is easy to show $b_2$ is either old block in $m_2$($j_{12}(b_1) = \some{(b_2,\delta)}$) or newly allocated block($j_{12}(b_1)= \none$), therefore $b_2 \in \validblock{m_2'}$.
\item 
  Given $j_{12}'(b_1) = \some{(b_2,o_2 - o_1)} \land (b_1,o_1) \in \permmax{m_1'}{\knonempty}$ and $j_{12}'(b_1') = \some{(b_2',o_2' - o_1')} \land (b_1',o_1') \in \permmax{m_1'}{\knonempty}$ where $b_1 \neq b_1'$. We need to prove these two positions do not overlap ($(b_2,o_2) \neq (b_2',o_2')$) by cases of whether $b_1$ and $b_1'$ are mapped by old injection $j_{12}$. Note that $j_{12} \subseteq j_{12}'$, so $j_{12}(b)$ is either $\none$ or the same as $j_{12}'(b)$.
    \begin{itemize}
    \item $j_{12} (b_1) = j_{12} (b_1') = \none$. The $j_{12}'$ mappings of them are added in step (1). It is obvious that newly added mappings in step (1) never map different blocks in $m_1'$ into the same block in $m_2'$. Therefore $b_2 \neq b_2'$.
    \item $j_{12}(b_1) = \some{(b_2,o_2 - o_1)}, j_{23}(b_1') = \none$. we can derive that $b_2 \in \validblock{m_2}$ by property (4) of $\minj{j_{12}}{m_1}{m_2}$. While $b_2'$ is newly allocated from $m_2$ in step (1). Therefore $b_2 \neq b_2'$.
    \item $j_{12}(b_1) = \none, j_{12}(b_1') = \some{(b_2',o_2' - o_1')}$. Similarly we have $b_2 \neq b_2'$.
    \item $j_{12}(b_1) = \some{(b_2,o_2 - o_1)}, j_{12}(b_1') =
    \some{(b_2',o_2' - o_1')}$. We can prove $(b_2,o_2) \neq
    (b_2,o_2')$ using the property (5) in $\minj{j_{12}}{m_1}{m_2}$ by
    showing $(b_1,o_1) \in \permmax{m_1}{\knonempty}$ and $(b_1',o_1')
    \in \permmax{m_1}{\knonempty}$. This follows from
    $\mpd{m_1}{m_1'}$ in
    $\injpacc{(j_{13},m_1,m_3)}{(j_{13}',m_1',m_3')}$.
   \end{itemize}
 \item Given $j_{12}'(b_1) = \some{(b_2, o_2 - o_1)} \land (b_2, o_2) \in \permk{k}{m_2'}{p}$. Similarly we prove $ (b_1,o_1) \in \permk{k}{m_1'}{p}$ or $(b_1,o_1) \not\in \permmax{m_1'}{\knonempty} $ by cases of $j_{12}(b_1)$:
   \begin{itemize}
   \item If $j_{12}(b_1) = \none$, then $b_1$ and $b_2$ are new blocks by $\injsep{j_{12}}{j_{12}'}{m_1}{m_2}$. According to the construction steps , every nonempty permission of $(b_2,o_2)$ in $m_2'$ is copied from $(b_1,o_1)$ in $m_1'$. Therefore $(b_1,o_1) \in \permk{k}{m_1'}{p}$.
   \item If $j_{12}(b_1) = \some{(b_2,o_2 - o_1)}$, then $b_1$ and $b_2$ are old blocks. We further divide $j_{23}(b_2)$ into two cases:
     \begin{itemize}
     \item $j_{23}(b_2) = \none$. In this case
       we have $(b_1,o_1) \in \unchangedon{m_1}{m_1'}$ and $(b_2,o_2) \in \unchangedon{m_2}{m_2'}$ (same as (1)). We can derive $(b_2,o_2) \in \permk{k}{m_2}{p}$, then $ (b_1,o_1) \in \permk{k}{m_1}{p} \lor (b_1,o_1) \not\in \permmax{m_1}{\knonempty}$ by property (6) of $\minj{j_{12}}{m_1}{m_2}$. Finally $ (b_1,o_1) \in \permk{k}{m_1'}{p} \lor (b_1,o_1) \not\in \permmax{m_1'}{\knonempty}$ by $(b_1,o_1) \in \unchangedon{m_1}{m_1'}$.
     \item $j_{23}(b_2) = \some{(b_3,o_3)}$. We assume that $(b_1,o_1)
     \in \permmax{m_1'}{\knonempty}$(other-wise the conclusion holds
     trivially), by $\mpd{m_1}{m_1'}$ we can derive that  $(b_1,o_1)
     \in \permmax{m_1}{\knonempty}$. Therefore $(b_2,o_2) \in
     \kwd{pub-tgt-}\allowbreak \kwd{mem}(j_{12}, m_1) \cap \pubsrcmem{j_{23}}$ is copied from
     $m_1'$ in step (3). As a result, we get $(b_1,o_1) \in
     \permk{k}{m_1'}{p}$ from $(b_2,o_2) \in \permk{k}{m_2'}{p}$.
     \end{itemize}
   \end{itemize}
\end{enumerate}
\end{proof}


\begin{lemma}\label{lem:constr-inj2}
  $\minj{j_{23}'}{m_2'}{m_3'}$
\end{lemma}
\begin{proof}
  $\ $
\begin{enumerate}
\item Given $j_{23}'(b_2) = \some{(b_3,o_3 - o_2)} \land (b_2,o_2) \in \permk{k}{m_2'}{p}$. We prove $(b_3,o_3) \in \permk{k}{m_3'}{p}$ by cases of whether $b_2 \in \validblock{m_2}$.
  \begin{itemize}
  \item If $b_2 \notin \validblock{m_2}$ is a new block relative to $m_2$, then $(b_2,o_2) \in \permk{k}{m_2'}{p}$ is copied from $m_1'$ in step (2). Therefore we get $(b_1,o_1) \in \permk{k}{m_1'}{p}$ and $j_{23}' \cdot j_{12}'(b_1) = \some{(b_3,o_3 - o_1)}$ according to step (1). From property (1) of $\minj{j_{13}'}{m_1'}{m_3'}$ we get $(b_3,o_3) \in \permk{k}{m_3'}{p}$;
  \item If $b_2 \in \validblock{m_2}$, then $j_{23}(b_2) = \some{(b_3,o_3)}$ from $\injsep{j_{23}}{j'_{23}}{m_2}{m_3}$. We further divide whether $(b_2,o_2) \in \outofreach{j_{12}}{m_1}$ using the same algorithm in step (2).
    \begin{itemize}
    \item If $(b_2,o_2) \in \outofreach{j_{12}}{m_1}$. According to ~\lemref{lem:m2-unchangedon}, we can derive $(b_2,o_2) \allowbreak \in \unchangedon{m_2}{m_2'}$ and $(b_2,o_2) \in \permk{k}{m_2}{p}$. From $\minj{j_{23}}{m_2}{m_3}$ we can derive $(b_3,o_3) \in \permk{k}{m_3}{p}$.
      By Lemma \ref{lem:outofreach-reverse}, $(b_3,o_3) \in \outofreach{\allowbreak j_{13}}{m_1}$. Therefore $(b_3,o_3)\in \unchangedon{m_3}{m_3'}$ and $(b_3,o_3) \in \permk{k}{\allowbreak m_3'}{p}$.
    \item If $(b_2,o_2) \notin \outofreach{j_{12}}{m_1}$, the permission of public position $(b_2,o_2)$ in $m_2'$ is copied from $m_1'$ in step (3). Thus $(b_1,o_1) \in \permk{k}{m_1'}{p}$ and $j_{13}'(b_1) = \some{(b_3,o_3 - o_1)}$.
      From property (1) of $\minj{j_{13}'}{m_1'}{m_3'}$ we get $(b_3,o_3) \in \permk{k}{m_3'}{p}$.
    \end{itemize}
  \end{itemize} 
\item The proof is similar to (1). Lemma \ref{lem:readable-midvalue} ensures that the constructed value $v_2$ in $m_2'$ can be related to the value in $m_3'$ as $\vinj{j_{23}'}{v_2}{m_3'[b_3,o_3]}$. Note that if $(b_2,o_2)$ is read-only in $m_2$, the property (1) of $\minj{j_{23}}{m_2}{m_3}$ provides that mapped position $(b_3,o_3)$ is also read-only in $m_3$.
\item Given $b_2 \notin \validblock{m_2'}$, we have $b_2 \notin \validblock{m_2}$ and $j_{23}(b_2) = \none$. Also $b_2$ is not added into the domain of $j_{23}'$ in step (1), so $j_{23}'(b_2) = \none$.
\item Given $j'_{23}(b_2) = \some{(b_3,o_3)}$. Similarly $b_3$ is either an old block in $m_3$($j_{23}(b_2) = \some{(b_3,o_3)}$) or a new block in $m_3'$($j_{23}(b_2) = \none$). Therefore $b_3 \in \validblock{m_3'}$.
\item
  Given $j_{23}'(b_2) = \some{(b_3,o_3 - o_2)} \land (b_2,o_2) \in \permmax{m_2'}{\knonempty}$ and $j_{23}'(b_2') = \some{(b_3',o_3' - o_2')} \land (b_2',o_2') \in \permmax{m_2'}{\knonempty}$ where $b_2 \neq b_2'$. We need to prove that $(b_3,o_3) \neq (b_3',o_3')$ by cases of whether $b_2$ and $b_2'$ are mapped by old injection $j_{23}$. Note that $j_{23} \subseteq j_{23}'$, so $j_{23}(b)$ is either $\none$ or the same as $j_{23}'(b)$.
    \begin{itemize}
    \item $j_{23} (b_2) = j_{23} (b_2') = \none$. The $j_{23}'$ mappings of them are added in step (1). It is obvious that newly added mappings in $j_{23}'$ never map different blocks in $m_2'$ into the same block in $m_3'$. Therefore $b_3 \neq b_3'$.
    \item $j_{23}(b_2) = \some{(b_3,o_3 - o_2)}, j_{23}(b_2') = \none$. we can derive that $b_3\in \validblock{m_3}$ By property (4) of $\minj{j_{23}}{m_2}{m_3}$. While $b_3' \notin \validblock{m_3}$ can be derived from $\injsep{j_{23}}{j_{23}'}{m_2}{m_3}$. Therefore $b_3 \neq b_3'$.
    \item $j_{23}(b_2) = \none, j_{23}(b_2') = \some{(b_3',o_3' - o_3')}$. Similarly we have $b_3 \neq b_3'$.
    \item $j_{23}(b_2) = \some{(b_3,o_3 - o_2)}, j_{23}(b_2') = \some{(b_3',o_3' - o_2')}$. We can prove $(b_3,o_3) \neq (b_3,o_3')$ using the property (5) in $\minj{j_{23}}{m_2}{m_3}$ by showing $(b_2,o_2) \in \permmax{m_2}{\knonempty}$ and $(b_2',o_2') \in \permmax{m_2}{\knonempty}$. This follows from $\mpd{m_2}{m_2'}$ (Lemma \ref{lem:m2-mpd}).
   \end{itemize}
 \item
   Given $j_{23}'(b_2) = \some{(b_3, o_3 - o_2)} \land (b_3, o_3) \in \permk{k}{m_3'}{p}$. Similarly we prove $ (b_2,o_2) \in \permk{k}{m_2'}{p}$ or $(b_2,o_2) \not\in \permmax{m_2'}{\knonempty} $ by cases of $j_{23}(b_2)$:
   \begin{itemize}
   \item If $j_{23}(b_2) = \none$, then $b_2$ and $b_3$ are new blocks by $\injsep{j_{23}}{j_{23}'}{m_2}{m_3}$. According to step (1), we know that $ \exists b_1 \app o_1, j_{13}'(b_1) = \some{(b_3,o_3- o_1)}$. At the same time, we also know that the permission of $(b_2,o_2)$ in new block of $m_2'$ is copied from $(b_1,o_1)$ in $m_1'$. Now from property (6) of $\minj{j_{13}'}{m_1'}{m_3'}$ we can derive that $(b_1,o_1) \in \permk{k}{m_1'}{p} \lor (b_1,o_1) \notin \permmax{m_1'}{\knonempty}$, therefore $(b_2,o_2) \in \permk{k}{m_2'}{p} \lor (b_2,o_2) \notin \permmax{m_2'}{\knonempty}$.
   \item If $j_{23}(b_2) = \some{(b_3,o_3 - o_2)}$, then $b_2$ and $b_3$ are old blocks. We further divide $b_2$ into two cases:
     \begin{itemize}
     \item If $(b_2,o_2) \in \outofreach{j_{12}}{m_1}$, we have
     $(b_2,o_2) \in \unchangedon{m_2}{m_2'}$ (Lemma
     \ref{lem:m2-unchangedon}). If $(b_2,o_2) \in
     \permmax{m_2'}{\knonempty}$(otherwise the conclusion holds
     trivially), then $(b_2,o_2) \in \permmax{m_2}{\knonempty}$ holds
     ($\kwd{max-perm-dec}\allowbreak (m_2,m_2')$).
     According to ~\lemref{lem:outofreach-reverse}, we get $(b_3,o_3) \in \unchangedon{m_3}{m_3'}$ and $(b_3,o_3)\in \permk{k}{m_3}{p}$. Then we can derive that $ (b_2,o_2) \in \permk{k}{m_2}{p} \lor (b_2,o_2) \not\in \permmax{m_2}{\knonempty}$ by property (6) of $\minj{j_{23}}{m_2}{m_3}$. Finally we can prove that \[(b_2,o_2) \in \permk{k}{m_2'}{p} \lor (b_2,o_2) \not\in \permmax{m_2'}{\knonempty}.\]

     \item If $(b_2,o_2) \notin \outofreach{j_{12}}{m_1}$, we know that $\exists b_1, j_{13}'(b_1)=\some{b_3,o_3 - o_1}$. From $\minj{j_{13}'}{m_1'}{m_3'}$ we can derive that $(b_1,o_1) \in \permk{k}{m_1'}{p} \lor (b_1,o_1) \notin \permmax{m_1'}{p}$. Meanwhile, the permission of $(b_2,o_2) \in \pubtgtmem{j_{12}}{m_1} \allowbreak \cap \pubsrcmem{j_{23}}$ is copied from $m_1'$ in step (3). Therefore
       \[(b_2,o_2) \in \permk{k}{m_2'}{p} \lor (b_2,o_2) \notin \permmax{m_2'}{\knonempty}\]
     \end{itemize}
   \end{itemize}
 \end{enumerate}
 \end{proof}
\begin{lemma}$\injpacc{(j_{12},m_1,m_2)}{(j_{12}',m_1',m_2')}$\label{lem:injp-acc1}
\end{lemma}
\begin{proof}
According to Definition \ref{def:injpacc}, most of the properties of $\injpacc{(j_{12},m_1,m_2)}{(j_{12}',m_1',m_2')}$ have been proved in Lemma \ref{lem:constr-inj}, Lemma \ref{lem:m2-unchangedon} and Lemma \ref{lem:m2-macc}. From $\injpacc{(j_{13},m_1,m_3)}{(j_{13}',m_1',m_3')}$ we can get $\macc{m_1}{m_1'}$ and $\unmapped{j_{13}} \subseteq \unchangedon{m_1}{m_1'}$. To get the last leaving property $\unmapped{j_{12}} \subseteq \unchangedon{m_1}{m_1'}$ we only need to show
\[
  \unmapped{j_{12}} \subseteq \unmapped{j_{13}}
\]
where $j_{13} = j_{23} \cdot j_{12}$. This relations holds simply because of $\forall b, j_{12}(b) = \none \imply j_{23} \cdot j_{12}(b) = \none$. In other word, more regions in $m_1$ is protected in $\injpacc{(j_{13},m_1,m_3)}{(j_{13}',m_1',m_3')}$ than in $\injpacc{(j_{12},m_1,m_2)}{(j_{12}',m_1',m_2')}$.
\end{proof}

\begin{lemma}$\injpacc{(j_{23},m_2,m_3)}{(j_{23}',m_2',m_3')}$\label{lem:injp-acc2}
  \end{lemma}
\begin{proof}
Similarly, we only need to show
\[
  \outofreach{j_{23}}{m_2} \subseteq \outofreach{j_{23}\cdot j_{12}}{m_1}
\]
Given $(b_3,o_3) \in \outofreach{j_{23}}{m_2}$, i.e. \[\forall b_2 \app o_2, j_{23}(b_2) = \some{(b_3,o_3)} \imply (b_2,o_2) \notin \permmax{m_2}{\knonempty}\] We need to prove $(b_3,o_3) \in \outofreach{j_{23}\cdot j_{12}}{m_1}$. as follows.
If $j_{23} \cdot j_{12} (b_1) = \some{(b_3,o_3)}$,i.e. $\exists b_2, j_{12}(b_1) = \some{(b_2,o_2)} \land j_{23}(b_2) = \some{(b_3,o_3)}$, we can derive that $(b_2,o_2) \notin \permmax{m_2}{\knonempty}$. By property (1) of $\minj{j_{12}}{m_1}{m_2}$, we can get $(b_1,o_1) \notin \permmax{m_1}{\knonempty}$. Therefore $ (b_3,o_3) \in \outofreach{j_{23}\cdot j_{12}}{m_1}$.
\end{proof}

Since we have proved all 4 required properties (~\lemref{lem:constr-inj1} to \lemref{lem:injp-acc2}) of the constructed memory state $m_2'$, ~\lemref{lem:injp-refine-injp-comp} is proved.

\subsection{Proof of Lemma \ref{lem:injp-comp-refine-injp}}
We prove Lemma \ref{lem:injp-comp-refine-injp} in this section:
\begin{tabbing}
  \quad\=\quad\=\quad\=$\exists m_2'\app j_{12}'\app j_{23}',$\=\kill
  \>$\forall j_{13}\app m_1\app m_3,\app \minj{j_{13}}{m_1}{m_3} 
     \imply \exists j_{12}\app j_{23}\app m_2,\app \minj{j_{12}}{m_1}{m_2} \land
     \minj{j_{23}}{m_2}{m_3} \land$\\
  \>\>$\forall m_1'\app m_2'\app m_3'\app j_{12}'\app j_{23}',\app
       \injpacc{(j_{12}, m_1, m_2)}{(j_{12}', m_1', m_2')} \imply
       \injpacc{(j_{23}, m_2, m_3)}{(j_{23}', m_2', m_3')} \imply$\\
  \>\>\>$\minj{j_{12}'}{m_1'}{m_2'} \imply 
       \minj{j_{23}'}{m_2'}{m_3'} \imply
       \exists j_{13'},\app \injpacc{(j_{13}, m_1, m_3)}{(j_{13}', m_1', m_3')}
         \land \minj{j_{13}'}{m_1'}{m_3'}.$
\end{tabbing}
\begin{proof}
  Given $\minj{j_{13}}{m_1}{m_3}$, take $j_{12} = \{(b,(b,0))| j_{13}(b) \neq \none\}$, $j_{23} = j_{13}$ and $m_2 = m_1$. As a result, $\minj{j_{23}}{m_2}{m_3}$ holds trivially. We show $\minj{j_{12}}{m_1}{m_1}$ as follows:
  \begin{enumerate}
  \item Given $j_{12}(b_1) = \some{(b_2,o_2 - o_1)} \land (b_1,o_1) \in \permk{k}{m_1}{p}$, according to the definition of $j_{12}$ we know that $b_2 = b_1$ and $o_2 = o_1$. Therefore $(b_2,o_2) \in \permk{k}{m_1}{p}$.
  \item Given $j_{12}(b_1) = \some{(b_2,o_2 - o_1)} \land (b_1,o_1) \in \permcur{m_1}{p}$, similar to (1) we know $b_2 = b_1$ and $o_2 = o_1$. Therefore $m_1[b_1,o_1] = m_1[b_2,o_2]$. If $m_1[b_1,o_1]$ is not in the form of $\vptr{b_1'}{o_1'}$, $\vinj{j_{12}}{m_1[b_1,o_1]}{m_1[b_2,o_2]}$ holds trivially.

    If $m_1[b_1,o_1] = \vptr{b_1'}{o_1'}$, from $j_{12}(b_1) = \some{(b_1,0)}$ we get $j_{13}(b_1) \neq \none$. According to property (2) of $\minj{j_{13}}{m_1}{m_3}$, $\exists v_3, \minj{j_{13}}{\vptr{b_1'}{o_1'}}{v_3}$. Which means that $j_{12}(b_1') = \some{(b_1',0)}$, therefore $\vinj{j_{12}}{\vptr{b_1'}{o_1'}}{\vptr{b_1'}{o_1'}}$.
  \item Given $b_1 \notin \validblock{m_1}$, we can derive that $j_{13}(b_1) = \none$ by $\minj{j_{13}}{m_1}{m_3}$. Therefore $j_{12}(b_1) = \none$ holds by definition.
  \item Given $j_{12}(b_1) = \some{(b_2,\delta)}$, we know that $j_{13}(b_1) \neq \none$. Therefore $b_1 \in \validblock{m_1}$ by $\minj{j_{13}}{m_1}{m_3}.$ Since $b_1 = b_2$, $b_2 \in \validblock{m_1}$.
  \item Given $b_1 \neq b_1'$, $j_{12}(b_1) = \some{b_2,o_2 - o_1}$ and $j_{12}(b_1') = \some{b_2', o_2' - o_1'}$. It is straightforward that $b_2 = b_1$, $b_2' = b_1'$ therefore $b_2 \neq b_2'$.
  \item Given $j_{12}(b_1) = \some{b_2,o_2 - o_1}$ and $(b_2,o_2) \in \permk{k}{m_1}{p}$. Similarly we have $b_2 = b_1$, $o_2 = o_1$ and $(b_1,o_1) \in \permk{k}{m_1}{p}$.
  \end{enumerate}
  After external calls, given preconditions $\injpacc{(j_{12}, m_1, m_2)}{(j_{12}', m_1', m_2')}$, $\injpacc{(j_{23}, m_2, m_3)}{(j_{23}', m_2', m_3')}$, $\minj{j_{12}'}{m_1'}{m_2'}$ and $\minj{j_{23}'}{m_2'}{m_3'}$. We can get $\minj{j_{13}'}{m_1'}{m_3'}$ directly by Lemma \ref{lem:inj-trans}. For $\injpacc{(j_{13}, m_1, m_3)}{(j_{13}', m_1', m_3')}$,
    \begin{enumerate}
    \item We can easily show $j_{13} = j_{23} \cdot j_{12}$ by the definition of $j_{12}$. Since $j_{12} \subseteq j_{12'}$, $j_{23} \subseteq j_{23}'$, we can conclude that $j_{23} \cdot j_{12} \subseteq j_{23}' \cdot j_{12}'$, i.e. $j_{13} \subseteq j_{13}'$.
    \item
      \[
        \unmapped{j_{13}} \subseteq \unchangedon{m_1}{m_1'}
      \]
      By definition of $j_{12}$, we have $\unmapped{j_{12}} = \unmapped{j_{13}}$. Therefore the result comes directly from $\injpacc{(j_{12}, m_1, m_2)}{(j_{12}', m_1', m_2')}$.
    \item
      \[
        \outofreach{j_{13}}{m_1} \subseteq \unchangedon{m_3}{m_3'}
      \]
      Since $j_{23} = j_{13}$ and $m_2 = m_1$, the result comes directly from $\injpacc{(j_{23}, m_2, m_3)}{(j_{23}', m_2', m_3')}$.
    \item $\macc{m_1}{m_1'}$ comes from $\injpacc{(j_{12}, m_1, m_2)}{(j_{12}', m_1', m_2')}$.
    \item $\macc{m_3}{m_3'}$ comes from $\injpacc{(j_{23}, m_2, m_3)}{(j_{23}', m_2', m_3')}$.
    \item \[\injsep{j_{13}}{j_{13}'}{m_1}{m_3}\]
      If $j_{13}(b_1) = \none$ and $j_{13}'(b_1) = \some{(b_3,o_3 - o_1)}$, we get \[j_{12}(b_1) = \none \text{ and } \exists b_2, j_{12}'(b_1) = \some{(b_2,o_2 - o_1)} \land j_{23}'(b_2) = \some{(b_3,o_3 - o_2)}\] by $\injsep{j_{12}}{j_{12}'}{m_1}{m_2}$ we get $b_1 \notin \validblock{m_1}$ and $b_2 \notin \validblock{m_2}$. By property (3) of $\minj{j_{23}}{m_2}{m_3}$ we can derive that $j_{23}(b_2) = \none$. Finally we get $b_3 \notin \validblock{m_3}$ by $\injsep{j_{23}}{j_{23}'}{m_2}{m_3}$.
    \end{enumerate}
 \end{proof}

\section{Verification of the Encryption Server and Client Example}
\label{sec:server-sim}

\subsection{Refinement of the Hand-written Server}

The following is the proof for~\thmref{the:l2sim}.

\begin{proof}
  At the top level, $\scc$ is expanded to $\kro \cdot \kwt \cdot \kcainjp \cdot
  \kasm_\kinjp$. As the invariant $\kro$ and $\kwt$ in $\cli$ level
  are commutative, i.e., $\scequiv{\kro \cdot \kwt}{\kwt \cdot \kro}$
  as stated in~\lemref{lem:inv-comm}, we can change their order in
  $\scc$.
  By the vertical compositionality, we first prove
  $\osims{\kwt}{L_{\texttt{S}}}{L_{\texttt{S}}}$ and
  $\osims{\kasm_\kinjp}{\sem{\code{server\_opt.s}}}{\sem{\code{server\_opt.s}}}$,
  which are both self simulation and straightforward (the latter one
  is provided by the adequacy theorem).
  Since the relation between source and target programs involves an
  optimization of constant propagation of the variable $\kwd{key}$, we need to
  use \kro together with \kcainjp to establish the simulation
  $\osims{\kro\cdot\kcainjp}{\Sspec}{\sem{\code{server\_opt.s}}}$.
  Note that for the unoptimized version we can prove $\osims{\kcainjp}
  {\Sspec}{\sem{\kser}}$ and prove $\kro$ using self simulation
  like $\kwt$.
  
  The key of this proof is to establish a relation $R \in \krtype{W_{\kro \cdot
      \kcainjp}}{S_S}{\kregset \times \kmem}$ satisfying the simulation diagram
  \figref{fig:server-sim}.
  %
  %
  Given $w \in W_\kro \times W_\kcainjp = ((se,m_0),((j,m,\tm), \sig, \regset))$,
  if $\sig \neq \intptrvoidsig \lor m_0\neq m$ then $R(w) = \emptyset$. Assume $\sig
  = \intptrvoidsig \land m_0 = m$ (these conditions are provided by $I$ of $L_{\texttt{S}}$
  and related incoming queries), then $R(w)$ is defined as follows:

{\small
\begin{tabbing}
  \quad\= (a.1)\quad\=\quad\=\kill
  \>(a) \>$(\kwd{Calle}\app i\app \vptr{b}{o} \app m, (\regset, \tm)) \in R(w) \iff \textcolor{orange}{\kwd{(* initial state *)}}$\\
  \>(a.1)\>\>$\regset(\rdireg)=i \land \vinj{j}{\vptr{b}{0}}{\regset(\rsireg)}
  \land \regset(\pcreg) =\vptr{b_e}{0} \land \minj{j}{m}{\tm} \land \color{red}{\rovalid{se}{m}}$\\
  \>(b) \>$(\kwd{Callp}\app \rb \app \vptr{b}{o} \app m_1, (\regset_1, \tm_1)) \in R(w) \iff {\color{orange}\kwd{(* before external call *)}}$\\
  \>(b.1)\>\>$\regset_1(\spreg) = \vptr{b_s}{0} \land \regset_1(\rdireg)=\vptr{b_s}{8} \land j'\app \rb=\some{b_s,8}$\\
  \>(b.2)\>\>$\land \regset_1({{\rareg}}) =\vptr{b_g}{5} \land \minj{j'}{m_1}{\tm_1} \land \vinj{j}{\vptr{b}{o}}{\regset_1(\pcreg)}$\\
  \>(b.3)\>\>$\land \forall r, r \in \kwd{callee-save-regs} \to \regset_1(r) = \regset(r)$ \\
  \>(b.4)\>\>$\land \injpacc{(j,m,\tm)}{(j',m_1,\tm_1)} \land \color{red}{\rovalid{se}{m_1}}$\\
  \>(b.5)\>\>$\land \tm_1[b_s,0] = \regset(\spreg) \land  \tm_1[b_s,16] = \regset(\rareg)$\\
  \>(b.6)\>\>$\land \color{blue}{\pset{(b_s,o)}{0 \leq o < 8 \vee 16 \leq o < 24} \subseteq \outofreach{j'}{m_1}} $  \\
  \>(b.7)\>\>$\land \pset{(b_s,o)}{0 \leq o < 8 \vee 16 \leq o < 24} \subseteq \permcur{\tm_1}{\kfreeable}$ \\
  \>(c) \>$(\kwd{Retp}\app \rb\app m_2, (\regset_2, \tm_2)) \in R(w) \iff {\color{orange}\kwd{(* after external call *)}}$  \\
  \>(c.1)\>\>$\regset_2(\spreg) = \vptr{b_s}{0} \land j''\app \rb=\some{b_s,8}$\\
  \>(c.2)\>\>$\regset_2({\pcreg}) =(b_g,5) \land \minj{j''}{m_2}{\tm_2}$ \\
  \>(c.3)\>\>$\land \forall r, r \in \kwd{callee-save-regs} \to \regset_2(r) = \regset(r)$ \\
  \>(c.4)\>\>$\land \injpacc{(j,m,\tm)}{(j'',m_2,\tm_2)} $\\
  \>(c.5)\>\>$\land \tm_2[b_s,0] = \regset(\spreg) \land  \tm_2[b_s,16] = \regset(\rareg)$\\
  \>(c.6)\>\>$\land \color{blue}{\pset{(b_s,o)}{0 \leq o < 8 \vee 16 \leq o < 24} \subseteq \outofreach{j''}{m_2}} $  \\
  \>(c.7)\>\>$\land \pset{(b_s,o)}{0 \leq o < 8 \vee 16 \leq o < 24} \subseteq \permcur{\tm_2}{\kfreeable}$ \\
  \>(d) \>$(\kwd{Rete}\ m_3, (\regset_3, \tm_3)) \in R(w) \iff {\color{orange}\kwd{(* final state *)}}$  \\
  \>(d.1)\>\>$ \regset_3(\spreg) = \regset(\spreg) \land \regset_3(\pcreg) = \regset(\rareg) \land \minj{j''}{m_3}{\tm_3}$ \\
  \>(d.2)\>\>$\land \forall r, r \in \kwd{callee-save-regs} \to \regset_3(r) = \regset(r)$ \\
  \>(d.3)\>\>$\land \injpacc{(j,m,\tm)}{(j'',m_3,\tm_3)}$
\end{tabbing}}
By definition the relation between internal states of $\Sspec$ and assembly states
evolve in four stages:
\begin{enumerate}[(a)]
\item Right after the initial call to \kencrypt, (a.1) indicates that
  the argument $i$ and function pointer are stored in $\rdireg$
  and $\rsireg$, the program counter is at $(b_e, 0)$ (pointing to the
  first assembly instruction of $\code{server\_opt.s}$
  in~\figref{fig:server_opt}). The $\rovalid{se}{m}$ comes from
  $\scname{R}_\kro^q$ and ensures that value of $\kwd{key}$ in $m_1$ is
  $42$.

\item Right before the external call, (b.1) indicates that
  the argument is stored in $\rdireg$, which is a pointer
  $\vptr{b_s}{8}$. Here $b_s$ is the stack block of the target
  assembly and $sp$ is injected to $\vptr{b_s}{8}$ as depicted in
  ~\figref{fig:mem-protect-exm}. (b.2)
  indicates that the return address is set to the 5th assembly
  instruction in~\figref{fig:server_opt} (right after \kwd{Pcall RSI})
  and the function pointer of $\kwd{p}$ $\vptr{b_p}{0}$ is related to
  $\pcreg$ by $j$. (b.3) indicates that callee-save registers are not
  modified since the initial call. (b.4) maintains the \kinjp accessibility
  and \kwd{ro-valid} for external call. (b.5), (b.6) and (b.7) indicate
  that the stored values (return address and previous stack block) on
  the stack are frame unchanged, protected and freeable.

\item Right after the external call, we keep the
  necessary conditions from (b), except that the program counter
  $\pcreg$ now points to the value in $\rareg$ before the external
  call. Note that the injection function is updated to $j''$ by the
  external call.

\item Right before returning from \kwd{encrypt}, (d.1) indicates that
  the stack pointer is restored and the return address is set. (d.2)
  indicates all callee-save registers are restored and (d.3) indicates
  that guarantee condition \kinjp is met.
  
\end{enumerate}

To prove $R$ is indeed an invariant to establish the simulation,
we follow the diagram in ~\figref{fig:server-sim}.
The most important points of the above proof is that \kwd{ro} and
\kinjp play essential roles in establishing the invariant (relevant
conditions are displayed in red and blue in the invariant,
respectively).
Initially, the target semantics enters the function from $\rs(\pcreg)$
which is related to the function pointer in $q_\cli^I$ as mentioned in
$\scname{R}_\kinjp^q$. The condition (a) follows from
$\scname{R}_\kinjp^q$ and $\scname{R}_\kro^q$ and hence holds at the
initial states.
Right before the execution calls, (b) holds by execution of the
instructions from \kwd{Pallocframe} to \kwd{Pcall}.
Note that $\rovalid{se}{m}$ obtained from \kro in (b.5) is essential
for proving that the value of $key$ read from $m$ is $42$, thus
matches the constant in $\code{server\_opt.s}$.
Then, we need to show (c) holds after the source and target execution
perform the external call and returns.
This is the most interesting part where the memory protection provided
by \kinjp is essential. It is achieved by combining properties
(b.1--7) with the rely-condition provided by \kcainjp of the external
call. For example, because we know the protected regions of the stack
frame $b_s$ is out-of-reach before the call, by the protection
enforced by \kinjp in $\scname{R}_\kcainjp^r$, all values in
$\tm_1[b_s, o]\ s.t.\ 0 \leq o < 8 \vee 16 \leq o < 24$ are unchanged,
therefore if $\tm_1[b_s,0] = \regset(\spreg)$ and $\tm_1[b_s,16] =
\regset(\rareg)$ (condition (b.5)) holds before the call, they also
hold after it (condition (c.5) holds).
Besides using \kinjp, we can derive (c.3) from (b.3) by the protection
over callee-save registers enforced in $\scname{R}_\kcainjp^r$.
$\rs_2(\pcreg) =(b_g,5)$ in (c.2) is derived from
$\rs_1(\rareg) =(b_g,5)$ in (b.2) via the relation between
$\pcreg$ and $\rareg$ stated in $\scname{R}_\kcainjp^r$.
After the external call, condition (d) can be derived from (c) by
following internal execution. Since $L_{\texttt{S}}$ frees $sp$,
$\sem{\code{server\_opt.s}}$ can free the corresponding region
$(b_s,8)$ to $(b_s,16)$. The remaining part are also freeable by
condition (c.7).
Finally, the target semantics returns after executing $\kwd{Pret}$ and
condition (d) provides the updated \kinjp world $(j'',m_3,\tm_3)$ with
the accessibility from the initial world $w$ and other properties
needed by $\scname{R}_\kcainjp^r$. The \kinjp accessibility also implies
$\macc{m}{m_3}$ for $\scname{R}_\kro^r$.
Therefore, we are able to establish the guarantee condition and prove
that $\osims{\scc}{\Sspec}{\sem{\code{server\_opt.s}}}$.

\end{proof}

\subsection{End-to-end Correctness Theorem}

\begin{figure}
  \begin{tikzpicture}
    \node (qr) at (0,0){$q_r$};
    \node (callr) [draw, ellipse,below right = 0.5 of qr,
    minimum width=1.2cm, minimum height=0.7cm] {\small \kwd{Callr}};
    \node (calle) [draw, ellipse,right = 0.8 of callr,
    minimum width=1.2cm, minimum height=0.7cm] {\small \kwd{Calle}};
    \node (callp) [draw, ellipse,right = 1.3 of calle,
    minimum width=1.2cm, minimum height=0.7cm] {\small \kwd{Callp}};
    \node (ret) [draw, ellipse,right = 1.0 of callp,
    minimum width=1.2cm, minimum height=0.7cm] {\small \kwd{Return}};
    \node (qe) [above left = 0.5 of calle] {$q_e$};
    \node (qp) [above left = 0.5 of callp] {$q_p$};
    \node (re) [above right = 0.5 of ret] {$r_e$};
    \node (rrp) [below right = 0.5 of ret] {$r_{\mathit{rp}}$};

    \draw [->](qr) -- node[above]{$I$} (callr);
    \draw [->](qe) -- node[above]{$I$} (calle);
    \draw [->](qp) -- node[above]{$I$} (callp);
    \draw [->](ret) -- node[above]{$F$} (re);
    \draw [->](ret) -- node[below]{$F$} (rrp);

    \draw [->](callr) -- node[above]{} (calle);
    \draw [->](calle) -- node[above]{\small \kwd{alloc}} node[below]{\small $\kwd{encrypt}$} (callp);
    \draw [->, bend left](callp) edge node[above]{\small \kwd{free}} node[below]{\small $\kwd{store}$} (ret);
    \draw [->, bend right](callp) edge node[below]{\small $\kwd{store}$} (ret);

  \end{tikzpicture}
  \vspace{-0.3cm}
  \caption{The Top-level Specification $\Topspec$}
  \label{fig:topspec}
  \vspace{-0.3cm}
\end{figure}

The top-level specification $\Topspec$ is defined as follows:

\begin{definition}
  LTS of $\Topspec$:
  {
  \begin{tabbing}
    \quad\=$\to_T$\;\=$:=$ \=\kill
    \>$S_T$ \>$:=$ \>$\rawset{\kwd{Callr}\app i \app m} \cup 
    \rawset{\kwd{Calle}\app \flag\app i\app v\app m} \cup
    \rawset{\kwd{Callp}\app \flag\app \retv\app \rb\app m} \cup 
    \rawset{\kwd{Return}\app \retv\app m}$;\\
    \>$I_T$ \>$:=$ \>$\rawset{(\cquery{\vptr{b_r}{0}}{\intintsig}{[\Vint{i}]}{m},\kwd{Callr}\app i \app m)} \cup$\\
    \>\>\>$\rawset{(\cquery{\vptr{b_e}{0}}{\intptrvoidsig}{[\Vint{i},v_c]}{m},\kwd{Calle}\app \false\ i\app v_c \app m)} \cup$\\
    \>\>\>$\rawset{(\cquery{\vptr{b_p}{0}}{\ptrvoidsig}{[\vptr{\rb}{0}]}{m},\kwd{Callp}\app \false\app \None\app \rb \app m)}$;\\
    \>$\to_T$ \>$:=$ 
    \>$\rawset{(\kwd{Callr}\app i \app m, \kwd{Calle}\app \true\app i\app \vptr{b_p}{0} \app m)} \cup$ \\
    \>\>\>$\{(\kwd{Calle}\app \flag\app i \app \vptr{b_p}{0} \app m, \kwd{Callp}\app \true\app \retv \app \rb\app m')\app |$\\
    \>\>\>$m' = m[\rb \leftarrow (i \app \kwd{xor} \app m[b_k])], \retv=\flag?\Some(i):\None\} \cup$ \\
    \>\>\>$\pset{(\kwd{Callp}\app \true \app \retv \app \rb \app m, \kwd{Return} \app \retv\app m')}{m' = m[result \leftarrow m[\rb]], m''=free\app m\app \rb} \cup $ \\
    \>\>\>$\pset{(\kwd{Callp}\app \false \app \retv \app \rb \app m, \kwd{Return} \app \retv\app m')}{m' = m[result \leftarrow m[\rb]]}$;\\
    \>$F_T$ \>$:=$ \>$\rawset{(\kwd{Return}\app \Some(i) \app m,\creply{\Vint{i}}{m})}\cup$\\
    \>\>\>$\rawset{(\kwd{Return}\app \None \app m,\creply{\Vundef}{m})}$;\\
  \end{tabbing}}
  \end{definition}
  
There are four internal states in $\Topspec$ as depicted in its
transition diagram~\figref{fig:topspec}, among which three states
correspond to the three functions (\kwd{request}, \kwd{process} and
\kwd{encrypt}).  We use $\flag$ in \kwd{Calle}(\kwd{Callp}) to
indicate whether it is called internally by
\kwd{request}(\kwd{encrypt}) or called by the environment.  The
$\retv$ is the return value which is either an integer when
$\Topspec$ is invoked by \kwd{request} or empty by other functions.
$I_T$ contains three possible initial states corresponding to calling
the three entry functions.
$\to_T$ describes the big steps starting from each call states to
another call state (e.g., \kwd{Callr} to \kwd{Calle}) or the return
state. If the stack block $\rb$ is allocated by $\kwd{encrypt}$, it
should be freed in the \kwd{Return} state.
There are two final states in $F_T$: $r_e$ returning from
\kwd{request} with an integer as the return value and $r_{\mathit{rp}}$
returning from \kwd{request} or \kwd{process} with no return value.
Note that $\Topspec$ can only be called but cannot perform external calls.

The remaining proof follows~\secref{sec:application}. We have shown
how end-to-end refinement is derived using the optimized server. For
unoptimized server, the proof is almost the same. The only difference
is that the symbol table accompanying $\Sspec$ does not mark \kwd{key}
as read-only, and the simulation invariant for~\thmref{the:l2sim} does
not contain \kwd{ro-valid} conditions as they play no role without
optimizations.

\subsection{Verification of Mutually Recursive Client and Server}
\label{ssec:app-mutual-client-server}

\begin{figure}
  \begin{subfigure}{0.35\textwidth}
  \begin{lstlisting}[language = C]
  /* client.c */
  #define N 10
  int input[N]={...};
  int result[N];
  int i;

  void encrypt(int i,
       void(*p)(int*));
\end{lstlisting}
\end{subfigure}
\begin{subfigure}{.6\textwidth}
\begin{lstlisting}[language = C]
  void request(int*r){
    if(i == 0)
      encrypt(input[i++],request);
    else if(0 < i && i < N){
      result[i-1]=*r;
      encrypt(input[i++],request); }
    else result[i-1]=*r;
    return; }
\end{lstlisting}
\end{subfigure}
\caption{Client with Multiple Encryption Request}
\label{fig:clientMR}
\end{figure}

We introduce an variant of the running example with mutual recursions
in~\figref{fig:clientMR}. The server remains the same while
\code{client.c} is modified so that its function \code{request} is
also the callback function and a sequence of encrypted values are
stored in a global array. 

To perform the same end-to-end verification for this example, we only
need to define a new top-level specification $\Topspec'$ to capture the semantics of the
multi-step encryption and prove $\osims{\kro \compsymb \kwt \compsymb
  \kc_\kinjp}{\Topspec'}{\sem{\code{client}} \semlink
  \Sspec}$. Other proofs are either unchanged (e.g., the refinement of
server) or can be derived from the verified compiler and the
horizontal compositionality.
In the following paragraphs, we briefly talk about this new top-level
specification and the updated proofs.

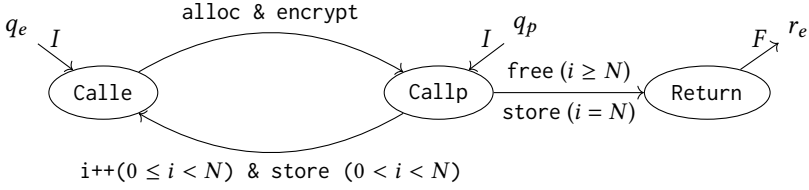
\begin{figure}
  \begin{tikzpicture}
    \node (qe) at (0,0){$q_e$};
    \node (calle) [draw, ellipse,below right = 0.5 of qe,
    minimum width=1.2cm, minimum height=0.7cm] {\small \kwd{Calle}};
    \node (callp) [draw, ellipse,right = 3.0 of calle,
    minimum width=1.2cm, minimum height=0.7cm] {\small \kwd{Callp}};
    \node (ret) [draw, ellipse,right = 2.0 of callp,
    minimum width=1.2cm, minimum height=0.7cm] {\small \kwd{Return}};
    \node (qp) [above right = 0.5 of callp] {$q_p$};
    \node (re) [above right = 0.5 of ret] {$r_e$};

    \draw [->](qe) -- node[above]{$I$} (calle);
    \draw [->](qp) -- node[above]{$I$} (callp);
    \draw [->](ret) -- node[above]{$F$} (re);

    \draw [->, bend left] (calle) edge node[above]{\small \kwd{alloc \& encrypt}} (callp);
    \draw [->, bend left] (callp) edge node[below]{\small \kwd{i++($0 \leq i < N$) \& store ($0<i<N$)}} (calle);
    \draw [->](callp) -- node[above]{\small \kwd{free}($i \geq N $)}
                         node[below]{\small \kwd{store}($i=N$)}  (ret);

  \end{tikzpicture}
  \vspace{-0.3cm}
  \caption{The Top-level Specification $\Topspec'$ of Mutually Recursive Client and Server}
  \label{fig:topspec-mr}
  \vspace{-0.3cm}
\end{figure}

\begin{definition} The LTS of $\Topspec'$:
  \begin{tabbing}
    \=$\to_T$\;\=$:=$ \=\kill
    \>$S_T$ \>$:=$ \>$\rawset{\kwd{Callr}\app sp\app sps \app m} \cup 
    \rawset{\kwd{Calle}\app i\app sps \app v\app m} \cup
    \rawset{\kwd{Return}\app m}$;\\
    \>$I_T$ \>$:=$ \>$\rawset{(\cquery{\vptr{b_r}{0}}{\ptrvoidsig}{[\vptr{sp}{0}]}{m},\kwd{Callr}\app sp\app nil \app m)} \cup$\\
    \>\>\>$\rawset{(\cquery{\vptr{b_e}{0}}{\intptrvoidsig}{[\Vint{i},v_c]}{m},\kwd{Calle}\app i\app nil\app v_c \app m)} \cup$\\
    \>$\to_T$ \>$:=$ 
    \>$\pset{(\kwd{Callr}\app sp\app sps \app m, \kwd{Calle}\app i\app sps\app \vptr{b_p}{0} \app m')}{m[b_i]==0, ...} \cup$ \\
    \>\>\>$\pset{(\kwd{Callr}\app sp\app sps \app m, \kwd{Calle}\app i\app sps\app \vptr{b_p}{0} \app m')}{0<m[b_i]<N,...} \cup$\\
    \>\>\>$\pset{(\kwd{Callr}\app sp\app sps \app m, \kwd{Return}\app m')}{m[b_i] \geq N, ...} \cup$ \\
    \>\>\>$\pset{(\kwd{Calle}\app i\app sps \app \vptr{b_p}{0}\app m, \kwd{Callr}\app sp\app (sp::sps) \app m')}{m'[sp \leftarrow i\app xor\app m[b_k]]}$;\\
    \>$F_T$ \>$:=$ \>$\rawset{(\kwd{Return} \app m,\creply{\Vundef}{m})}$;\\
  \end{tabbing}
\end{definition}

As we remove the function \code{process}, the new $\Topspec$ has only
two call states and one return state as depicted in ~\figref{fig:topspec-mr}. \kwd{sps} in \kwd{Callr} and
\kwd{Calle} is a list of blocks, each of which stores an encrypted
result. We record these blocks in the program states because we need
to de-allocate them before returning.
As described in $\to_T$, there are three internal transitions for
\kwd{Callr}, corresponding to three conditional branches in the source
code.
The transitions from \kwd{Callr} to \kwd{Calle} perform different
memory operations depending on the value of \kwd{i} according to
the code presented in~\figref{fig:clientMR}.
The transition from \kwd{Callr} to \kwd{Return} will de-allocate all the
stack blocks in \kwd{sps}.
The transition from \kwd{Calle} to \kwd{Callr} allocate a new block
\kwd{sp} to store the encrypted result and add it to \kwd{sps}.

Given this new top-level specification, we need to prove
\thmref{the:l2sim} where the $\Topspec$ is replaced by
$\Topspec'$. The key of this proof is that the simulation invariant
must relate the call stack in the target LTS (i.e., the semantics
linking of \code{client.c} and $\Sspec$) and \kwd{sps}, because each
element in \kwd{sps} is allocated by a call to$\kwd{encrypt}$ and
stores the result of encryption. The complete proofs can be found in
our Coq development.

\section{A Mutual Recursive Example for Summation}\label{sec:mut-sum}

In this section, we present the application of our method to an
example borrowed from CompCertM --- two programs that mutually invoke
each other to finish a summation task. 

\begin{figure}
\begin{subfigure}{0.45\textwidth}
\begin{lstlisting}[language = C]%how to present c code?

  /* C implementation of M_C */
  static int memoized[1000] = {0};
  int f(int i) {
    int sum;
    if (i == 0) return 0;
    sum = memoized[i];
    if (sum == 0) 
    { sum = g(i-1) + i;
      memoized[i] = sum; }
    return sum;
  }
  /* C code corresponding to M_A */
  static int s[2] = {0,0};
  int g(int i){
    int sum;
    if (i == 0) return 0;
    if (i == s[0]) 
      { sum = s[1]; } 
    else 
      { sum = f(i-1) + i;
        s[0] = i;
        s[1] = sum; }
    return sum;
  }
\end{lstlisting}
\end{subfigure}
\begin{subfigure}{0.45\textwidth}
\begin{lstlisting}[language = C]
/* Assembly implementation of M_A */  
g:  Pallocframe 24 16 0 
    Pmov RBX 8(RSP) // save RBX
    /* begin */
    Pmov RDI RBX 
    Ptestl RBX RBX  // i==0
    Pjne l0  
    Pxorl_r RAX     // rv=0
    Pjmp l1
l0: Pmov s[0] RAX 
    Pcmpl RAX RBX  // i==s[0]
    Pje l2
    Pleal -1(RBX) RDI 
    Pcall f        // f(i-1)
    Pleal (RAX,RBX) RAX//sum=f(i-1)+i
    Pmov RBX s[0]  // s[0] = i
    Pmov RAX s[1]  // s[1] = sum
    Pjmp l1 
l2: Pmov s[1] RAX  // rv=s[1]
    /* return */
l1: Pmov 8(RSP) RBX 
    Pfreeframe 24 16 0 
    Pret
\end{lstlisting}  
\end{subfigure}
  \caption{Heterogeneous Sum with Mutual Recursion}
  \label{fig:ccm-example}
\end{figure}

It consists of a \kwd{Clight} module $M_C$ and a hand-written assembly
module $M_A$.
The code of $M_A$ and $M_C$ is shown in~\figref{fig:ccm-example}. Note
that we have also shown a version of $M_A$ at the C level for
reference and given its function the name \kwd{g}; this program do not
actually exist in our example. 
We note that \kwd{f} and \kwd{g} collaborate to implement the
summation from $0$ to $i$ given an integer $i$. We shall use
$\exmsig$ to denote their signature. \kwd{f} perform
caching of results for any $i$ in a global array while \kwd{g} only
caches for the most recent $i$. When they need to compute a fresh
result, they mutually recursively call each other with a smaller
argument.
The assembly program uses pseudo X86 assembly instructions defined in
CompCert where every instruction begins with a letter \kwd{P}. The
only real pseudo instructions are \kwd{Pallocframe} and
\kwd{Pfreeframe}. \kwd{Pallocframe 24 16 0} allocates a stack block
$b_s$ of 24 bytes (3 integers on 64-bit x86), saves \kwd{RSP} and
\kwd{RA} to $(b_s,0)$ and $(b_s, 16)$ and set \kwd{RSP} to
$\vptr{b_s}{0}$. \kwd{Pfreeframe 24 16 0} recovers \kwd{RSP} and
\kwd{RA} from $(b_s, 0)$ and $(b_s, 16)$ and frees the stack block
$b_s$.
By the calling convention and the signature of \kwd{g}, \kwd{RDI} is
used to pass the only argument $i$. \kwd{RBX} is a callee-saved
register that stores $i$ during internal execution. It is saved to
$(b_s, 8)$ at the beginning of \kwd{g} and restored at the end.
Therefore, the sole purpose of $b_s$ is to save and restore \kwd{RSP},
\kwd{RA} and \kwd{RBX}.

\begin{figure}
  \begin{tikzpicture}
    \node [] (mc) at (-5,0) {$\sem{M_C}$};
    \node [above = 0.7 of mc] (lc) {$L_C$};
    \node [below = 0.7 of mc] (comp-mc) {$\sem{\kwd{CompCert}(M_C)}$};

    \path (mc) -- node[sloped, rotate=180] {$\osims{\scc}{}{}$} (comp-mc);
    \path (lc) -- node[sloped, rotate=180] {$\osims{\kwd{K}}{}{}$} (mc);

    \node [right = 2.0 cm of lc] (la) {$L_A$};
    \path let \p1 = (la) in let \p2 = (comp-mc) in
        node (ma) at (\x1, \y2) {$\sem{M_A}$};

    \path (la) -- node[sloped, rotate=180] 
          {$\osims{\scc}{}{}$} (ma);

    \path (lc) -- node[pos=0.5, sloped] {$\semlink$} (la);
    \path (comp-mc) -- node[sloped] {$\semlink$} (ma);

    \node [fit = (lc) (la), inner sep = -0.0pt] (srcsem) {};
    \node [fit = (comp-mc) (ma), inner sep = -0.0pt] (tgtsem) {};
    \node [draw, rectangle, dashed, below = 3.4 cm of srcsem] (tgtsyn) 
          {$\sem{\kwd{CompCert}(M_C) + M_A}$};

    \path (srcsem) -- node[pos=0.85, sloped,rotate=180] 
          {$\osims{\kid}{}{}$} (tgtsyn);

    \node [draw, rectangle, dashed, above = 0.7cm of srcsem] (topspec)
          {$L_{CA}$};

    \path (srcsem) -- node[sloped] 
          {$\osims{\kwd{K}}{}{}$} (topspec);

    \draw [-stealth] (tgtsyn) -- +(2.2,0) |- 
          node[pos = .26, sloped, below] 
            {$\osims{\scc}{}{}$}
            (topspec);

    \draw [-stealth] (comp-mc) -- +(-1.8,0) |- 
          node[pos = .26, sloped, above] 
            {$\osims{\scc}{}{}$}
            (lc);
  \end{tikzpicture}
  \caption{Verification of the Mutual Sum ($\kwd{K}:=\kro \cdot \kwt \cdot \kc_\kinjp$)}
  \label{fig:ccm-outline}
\end{figure}

The outline of the verification is presented
in~\figref{fig:ccm-outline}. It is similar
to~\figref{fig:running-exm-refinement} except for the additional $L_C$
which will be discussed soon. Firstly, as what we do in the
client-server example, we write down the specification for the
assembly module $M_A$ which is called $L_A$ defined
in~\defref{def:ccm-la} and prove the simulation between $L_A$ and
$\sem{M_A}$ which is declared in~\thmref{thm:ccm-la-refine}.
Secondly, we define a top-level specification $L_{CA}$ to abstract the
semantics of the composition of $M_C$ and $M_A$, which is shown
and~\defref{def:ccm-top-spec}. Intuitively, $L_{CA}$ says that the
output is the summation from zero to the input. In this example, we
additionally define a C-level specification for $M_C$ called $L_C$
(defined in~\defref{def:ccm-c-spec}) and prove $\osims{\kro \cdot \kwt
\cdot \kc_\kinjp}{L_C}{\sem{M_C}}$
(in~\thmref{thm:ccm-c-refine}). We can compose this proof with the
compiler correctness by utilizing~\thmref{lem:toprefine} to prove
$\osims{\scc}{L_C}{\sem{\kwd{CompCert}(M_C)}}$. With $L_C$,
it is simpler to prove the source refinement declared
in~\thmref{thm:ccm-top-refine}. Finally, we combine these proofs to
obtain the single refinement between top-level specification and the
target linked program as declared in~\thmref{thm:ccm-final}.

%

\begin{definition}\label{def:ccm-la} The open LTS of $L_A$ is defined
as follows:
{
\begin{tabbing}
  \quad\=$\to_A$\;\=$:=$ \=\kill
  \>$S_A$ \>$:=$ \>$\rawset{\kwd{Callg}\app i\app m} \cup 
          \rawset{\kwd{Callf}\app v_f\app i\app m} \cup 
          \rawset{\kwd{Returnf}\app i\app r\app m} \cup 
          \rawset{\kwd{Returng}\app r\app m}$;\\
  \>$I_A$ \>$:=$ \>$\rawset{(\cquery{\vptr{b_g}{0}}{\exmsig}{[\Vint{i}]}{m}),(\kwd{Callg}\app i\app m)}$;\\
  \>$\to_A$ \>$:=$ 
    \>$\pset{(\kwd{Callg}\app i\app m, \kwd{Returng}\app 0\app m)}{i = 0} \cup$ \\
    \>\>\>$\pset{(\kwd{Callg}\app i\app m,\kwd{Returng}\app r\app m)}{i \neq 0 \land i = s[0] \land r = s[1]} \cup$ \\
    \>\>\>$\pset{(\kwd{Callg}\app i\app m, \kwd{Callf}\app v_f\app i\app m)}{i \neq 0 \land i \neq s[0] \land v_f = \kwd{find-func-pointer}(\kwd{f})} \cup$ \\
    \>\>\>$\pset{(\kwd{Returnf}\app i\app res\app m, \kwd{Returng}\app (i+res)\app m')}{m' = m[s[0] \leftarrow i, s[1] \leftarrow (i+res)]}$;\\
  \>$X_A$ \>$:=$ \>$\rawset{(\kwd{Callf}\app v_f\app i\app m,\cquery{v_f}{\exmsig}{[\Vint{i-1}]}{m})}$;\\
  \>$Y_A$ \>$:=$ \>$\rawset{(\kwd{Callf}\app v_f\app i\app m,\creply{r}{m'}),\kwd{Returnf}\app i\app r\app m')}$;\\
  \>$F_A$ \>$:=$ \>$\rawset{(\kwd{Returng}\app i\app m,\creply{i}{m})}$.
\end{tabbing}}

\end{definition}

By this definition, there are four kinds of internal states:
\kwd{Callg} is at right after the initial call to \kwd{g}; \kwd{Callf}
is right before the external call to \kwd{f}; \kwd{Returnf} is right
after returning from \kwd{f}; and \kwd{Returng} is right before
returning from \kwd{g}. The definitions of transition relations
directly match the C-level version of \kwd{g} in~\figref{fig:ccm-example},
albeit in a big-step style. Note that when transiting internally from
$\kwd{Callg}\app i\app m$ to $\kwd{Callf}\app v_f\app i\app m$,
$\kwd{find-func-pointer}$ is used to query the global symbol table for
the function pointer to \kwd{f}. Also note that in $L_A$ the memory
state $m$ is not changed from \kwd{Callg} to \kwd{Callf}, while in the
assembly code $M_A$ a new stack frame is allocated by
\kwd{Pallocframe}. This indicates the stack frame is out-of-reach at
the external call to \kwd{f} and should be protected during its
execution. This point is also manifested in the proof below.

\begin{theorem}\label{thm:ccm-la-refine}
  $\osims{\scc}{L_A}{\sem{M_A}}$
\end{theorem}

\begin{proof}
The key is to identify a relation $R \in
\krtype{W_\kcainjp}{S_A}{\kregset \times \kmem}$ satisfying all the
properties in~\defref{def:open-sim}. Given $w \in W_\kcainjp =
((j,m_1,m_2), \sig, \regset)$, if $\sig \neq \exmsig$ then $R(w) =
\emptyset$. Assume $\sig = \exmsig$, then $R(w)$ is defined as
follows:
{
\begin{tabbing}
  \quad\= (a.1)\quad\=\quad\=\kill
  \>(a) \>$(\kwd{Callg}\app i\app m_1, (\regset, m_2)) \in R(w) \iff \textcolor{orange}{\kwd{(* initial state *)}}$\\
  \>(a.1)\>\>$\regset(\rdireg)=i \land \regset(\pcreg) =\vptr{b_g}{0} \land \minj{j}{m_1}{m_2}$\\
  \>(b) \>$(\kwd{Callf}\app v_f\app i\app m_1', (\regset', m_2')) \in R(w) \iff {\color{orange}\kwd{(* before external call *)}}$\\
  \>(b.1)\>\>$\regset'(rbx)=i \land \regset'({{\color{red}\rareg}}) =\vptr{b_g}{13} \land \minj{j'}{m_1'}{m_2'} \land \vinj{j'}{v_f}{\regset'(\pcreg)}$\\
  \>(b.2)\>\>$\land \forall r, (r \in \kwd{callee-saved-regs} \land r \neq \rbxreg) \to \regset'(r) = \regset(r)$ \\
  \>(b.3)\>\>$\land \regset'(\spreg) = \vptr{b_s}{0} \land {\color{blue}\lnot(\exists b\ o, j \ b = \some{b_s,o})}$  \\
  \>(b.4)\>\>$\land \injpacc{(j,m_1,m_2)}{(j',m_1',m_2')}$\\
  \>(b.5)\>\>$\land m_2'[b_s,0] = \regset(\spreg) \land m_2'[b_s,8] = \regset(\rbxreg) \land m_2'[b_s,16] = \regset(\rareg)$ \\
  \>(c) \>$(\kwd{Returnf}\app i\app res\app  m_1', (\regset', m_2')) \in R(w) \iff {\color{orange}\kwd{(* after external call *)}}$  \\
  \>(c.1)\>\>$\regset'(\rbxreg)=i \land \regset'({\color{red}\pcreg}) =(b_g,13) \land \regset'(rax) = res \land \minj{j'}{m_1'}{m_2'}$ \\
  \>(c.2)\>\>$\land \forall r, (r \in \kwd{callee-saved-regs} \land r \neq \rbxreg) \to \regset'(r) = \regset(r)$ \\
  \>(c.3)\>\>$\land \regset'(\spreg) = \vptr{b_s}{0} \land {\color{blue}\lnot(\exists b\ o, j' \ b = \some{b_s,o})}$ \\
  \>(c.4)\>\>$\land \injpacc{(j,m_1,m_2)}{(j',m_1',m_2')}$\\
  \>(c.5)\>\>$\land m_2'[b_s,0] = \regset(\spreg) \land m_2'[b_s,8] = \regset(\rbxreg) \land m_2'[b_s,16] = \regset(\rareg)$\\
  \>(d) \>$(\kwd{Returng}\ res\ m_1', (\regset', m_2')) \in R(w) \iff {\color{orange}\kwd{(* final state *)}}$  \\
  \>(d.1)\>\>$\regset'(\raxreg)=res \land \regset'(\spreg) = \regset(\spreg) \land \regset'(\pcreg) = \regset(\rareg) \land \minj{j'}{m_1'}{m_2'}$ \\
  \>(d.2)\>\>$\land \forall r, r \in \kwd{callee-saved-regs} \to \regset'(r) = \regset(r)$ \\
  \>(d.3)\>\>$\land \injpacc{(j,m_1,m_2)}{(j',m_1',m_2')}$
\end{tabbing}}
By definition the relation between internal states of
$L_A$ and assembly states evolve in four stages:
\begin{enumerate}[(a)]
\item Right after the initial call to \kwd{g}, (a.1) indicates that
  the argument $i$ is stored in $\rdireg$ and the program counter is at
  $(b_g, 0)$ (pointing to the first assembly instruction
  in~\figref{fig:ccm-example});

\item Right before the external call to \kwd{f}, (b.1) indicates $i$
  is stored in $\rbxreg$, the return address is set to the 13th
  assembly instruction in~\figref{fig:ccm-example} (right after \kwd{Pcall
    f}) and $v_f$ matches with the program counter. (b.2) indicates
  callee saved registers---except for $\rbxreg$--are not modified
  since the initial call. (b.3) indicates the entire stack frame $b_s$
  is out-of-reach. (b.4) maintains properties in \kinjp. (b.5)
  indicates values on the stack frame is not modified since the
  initial call.

\item Right after the external call to \kwd{f}, we have almost the
  same conditions as above, except that the program counter points to
  the return address set at the call to \kwd{f}.

\item Right before returning from \kwd{g}, (d.1) indicates the return
  value is in $\raxreg$, the stack pointer is restored and the return
  address is set. (d.2) indicates all callee-saved registers are
  restored and (d.3) indicates the guarantee condition \kinjp is met.
\end{enumerate}

To prove $R$ is indeed an invariant, we first show that condition (a)
holds at the initial state. We then show by internal execution we can
prove (b) holds right before the call to \kwd{f}. Now, the source and
target execution proceed by calling and returning from \kwd{f}, after
which we need to shown (c) holds. This is the most interesting part:
it is achieved by combining properties (b.1--5) with the
rely-condition provided by \kcainjp for calling \kwd{f}. For example,
because we know the stack frame $b_s$ is out-of-reach at the call to
\kwd{f}, by the accessibility enforced by \kinjp in
$\scname{R}_\kcainjp^r$ in~\defref{def:cainjp}, all values in
$m_2'[b_s, o]$ are unchanged, therefore if $m_2'[b_s,0] =
\regset(\spreg) \land m_2'[b_s,8] = \regset(\rbxreg) \land
m_2'[b_s,16] = \regset(\rareg)$ (condition (b.5)) holds before calling
\kwd{f}, they also hold after (hence condition (c.5) holds).
Similarly, we can derive (c.2) from (b.2) by the protection over
callee-saved registers enforced in $\scname{R}_\kcainjp^r$. Moreover,
$({\color{red}\pcreg}) =(b_g,13)$ in (c.1) is derived from
$({\color{red}\rareg}) =(b_g,13)$ in (b.1) via the relation between
$\pcreg$ and $\rareg$ stated in $\scname{R}_\kcainjp^r$.
After the external call, we show condition (d) can be derived from (c)
by following internal execution. We note that condition (d) provides
exactly the guarantee-condition needed by $\scname{R}_\kcainjp^r$ for
the incoming call to \kwd{g}.
Therefore, we successfully show $\osims{\scc}{L_A}{\sem{M_A}}$ indeed
holds. 


\end{proof}

\begin{definition}\label{def:ccm-c-spec}
  The C-level specification $L_C$ is defined as follows: 
{
\begin{tabbing}
  \quad\=$\to_C$\;\=$:=$ \=\kill
  \>$S_C$ \>$:=$ \>$\rawset{\kwd{Callf}\app i\app m} \cup 
          \rawset{\kwd{Callg}\app v\app i\app m} \cup 
          \rawset{\kwd{Returng}\app i\app sum\app m} \cup
          \rawset{\kwd{Returnf}\app sum\app m}$;\\
  \>$I_C$ \>$:=$ \>$\rawset{(\cquery{\vptr{b_f}{0}}{\exmsig}{[\Vint{i}]}{m}),(\kwd{Callf}\app i\app m)}$;\\
  \>$\to_C$ \>$:=$ 
    \>$\pset{(\kwd{Callf}\app i\app m, \kwd{Returnf}\app 0 \app m)}{i==0} \cup$ \\
    \>\>\>$\pset{(\kwd{Callf}\app i\app m,\kwd{Returnf}\app sum \app m)}{i\neq 0, m[b_m,4*i]\neq 0, sum=m[b_m,4*i]} \cup$ \\
    \>\>\>$\pset{(\kwd{Callf}\app i\app m,\kwd{Callg}\app \vptr{b_g}{0}\app i \app m)}{i\neq 0, m[b_m,4*i] = 0}$ \\
    \>\>\>$\pset{(\kwd{Returng}\app i\app sum\app m, \kwd{Returnf}\app sum'\app m')}{sum'=sum+i, m'=m[m[b_m,4*i]\leftarrow sum']}$; \\
  \>$X_C$ \> $:=$ \> $\rawset{(\kwd{Callg}\app v_g\app i \app m, \cquery{v_g}{\exmsig}{[\Vint{i-1}]}{m})}$;\\
  \>$Y_C$ \>$:=$ \>$\rawset{(\kwd{Callg}\app v_g\app i\app m,\creply{sum}{m'},\kwd{Returng}\app i\app sum'\app m')}$;\\
  \>$F_C$ \>$:=$ \>$\rawset{(\kwd{Returnf}\app sum\app m,\creply{sum}{m})}$;
\end{tabbing}}
  
\end{definition}

The four internal states capture the execution of function \kwd{f} in
$M_C$. From the initial state \kwd{Callf}, depending on the value of
$i$ and the cached value in $sum[b_m,4*i]$ where $b_m$ is the memory
block of \kwd{memoized}, \kwd{Callf} may return 0 if \kwd{i} is equal
to 0, may return the cached value if it is not zero, and may enter
\kwd{Callg} state to invoke function \kwd{g}. At \kwd{Callg} state, it
emits the query to the environment and when it receives the reply which
contains the summation of $i$, it would enter \kwd{Returng} state.
Finally, \kwd{Returng} state would calculate the sum of $i$, cache it
and enter the final state \kwd{Returnf}.

\begin{theorem}\label{thm:ccm-c-refine}
  $\osims{\kro \cdot \kwt \cdot \kc_\kinjp}{L_C}{\sem{M_C}}$
\end{theorem}

By decomposing $\kro \cdot \kwt \cdot \kc_\kinjp$, the key is to prove
that $\osims{\kc_\kinjp}{L_C}{\sem{M_C}}$. Because we do not
have to concern about the interleaving execution of $L_C$ and its
environment, the proof is straightforward.

\begin{definition}\label{def:ccm-top-spec}
  The top-level specification $L_{CA}$ is defined below:
  {
\begin{tabbing}
  \quad\=$\to_T$\;\=$:=$ \=\kill
  \>$S_T$ \>$:=$ \>$\rawset{\kwd{Callf}\app i\app m} \cup 
          \rawset{\kwd{Callg}\app i\app m} \cup 
          \rawset{\kwd{Return}\app i\app m}$;\\
  \>$I_T$ \>$:=$
  \>$\rawset{(\cquery{\vptr{b_f}{0}}{\exmsig}{[\Vint{i}]}{m}),(\kwd{Callf}\app i\app m)} \cup$\\
  \>\>\> $\rawset{(\cquery{\vptr{b_g}{0}}{\exmsig}{[\Vint{i}]}{m}),(\kwd{Callg}\app i\app m)}$;\\
  \>$\to_T$ \>$:=$ 
    \>$\pset{(\kwd{Callf}\app i\app m, \kwd{Return}\app r\app m')}{r=sum(i,m),m'=\mathit{cache(i,m)}} \cup$ \\
    \>\>\>$\pset{(\kwd{Callg}\app i\app m, \kwd{Return}\app r\app m')}{r=sum(i,m),m'=\mathit{cache}(i,m)}$; \\
  \>$F_T$ \>$:=$ \>$\rawset{(\kwd{Return}\app r\app m,\creply{r}{m}))}$;
\end{tabbing}}

The specification contains two call states representing invocation of
function \kwd{f} and \kwd{g}, and one return state. The internal
transitions are big step of the execution, which omit the details of
mutual recursion. The return value contained in the return state is
determined by $sum(i,m)$, meaning that the return value depends on the
contents of the memory, i.e., the contents of the initial contents of
\kwd{memorized} in $M_C$ and \kwd{s} in $M_A$. The memory of the
return state is updated by $\mathit{cache}$, which cached the values
generated during the execution to \kwd{memorized} and \kwd{s}.

\end{definition}

\begin{theorem}\label{thm:ccm-top-refine}
$\osims{\kro \cdot \kwt \cdot
\kc_\kinjp}{L_{CA}}{L_C \semlink L_A}$
\end{theorem}

The key of this proof is to relate the memory operations of
$\mathit{cache}$ and the operations of $L_C \semlink
L_A$. We achieve this by the induction of the input value $i$.
Detailed proofs can be found in our supplementary code.

\begin{theorem}\label{thm:ccm-final}
  $\osims{\scc}{L_{CA}}{\sem{\kwd{CompCert}(M_C) + M_A}}$
\end{theorem}
\begin{proof}
  Firstly, applying vertical compositionality, we decompose the proof
  to $\osims{\kro \cdot \kwt \cdot \kc_\kinjp}{L_{CA}}{L_C \semlink
  L_A}$ and $\osims{\scc}{L_C \semlink L_A}{\sem{\kwd{CompCert}(M_C)
  + M_A}}$ by expanding $\scc$ to $\kro \cdot \kwt \cdot
  \kc_\kinjp \cdot \scc$ with \lemref{lem:toprefine}. The former one
  is proved in~\thmref{thm:ccm-top-refine}. For the latter one, we
  first apply the horizontal compositionality to verify it modularly.
  The verification of $\osims{\scc}{L_A}{\sem{M_A}}$ is
  proved in~\thmref{thm:ccm-la-refine}. The verification of
  $\osims{\scc}{L_C}{\kwd{CompCert}(M_C)}$ is proved by
  applying~\lemref{lem:toprefine}, \thmref{thm:ccm-c-refine} and the
  compiler correctness theorem.
\end{proof}

\section{Comparison using the Running Example}
\label{sec:comparison}

In this section, we carefully compare the difference between our direct refinement
with other refinements in existing approaches to VCC using the running example.
The comparison is mainly based on the (estimated) effort required to prove the
running example. We also compare their results in the form of final
theorems of semantics preservation.

Since CompComp does not support open semantics using $\asmli$ interface and
the adequacy theorem for assembly code, we only introduce details about approaches
using sum of refinements (CompCertM) and product of refinements (CompCertO) for
further comparison here.

\subsection{CompCertM}
CompCertM uses Refinement Under Self-related Contexts (RUSC) to achieve vertical
and horizontal composition of refinements. We roughly describe the framework of
RUSC for presenting their verification and comparing it with ours (\apdxref{sec:mut-sum}).

Their open simulations are parameterized by different \emph{memory relations}
which mirror our use of KMRs. They do not need to lift memory relations to different
simulation conventions because the semantics of all languages from $\cli$ to $\asmli$
can perform C-style calls and returns.
The RUSC refinement is parameterized over a fixed set of memory
relations $\mathcal{R} = \{R_1,R_2,...,R_n\}$. $p \rref{R} p'$ is defined as
for any context program $c$, it $c$ is self-related by all memory relations
$R \in \mathcal{R}$, then $\kwd{Beh}( c \oplus p) \supseteq \kwd{Beh}(c \oplus p')$.
Note that the behavior $\kwd{Beh()}$ is only defined for closed program.
Thanks to this definition, RUSC refinements under a fixed $\mathcal{R}$ can be
easily composed vertically and horizontally if the \emph{end modules} are self-related
by all memory relations in $\mathcal{R}$. However, the disadvantage of RUSC-based
approach is exactly the fixed $\mathcal{R}$ and the requirement that
end programs are self-related by $\mathcal{R}$. Next, we demonstrate the impact of this
restriction on VCC through the mutual summation example.

In CompCertM, they use the example of mutual summation to illustrate the source-level
verification using RUSC. The structure of the proof is the same as depicted in
~\figref{fig:ccm-outline}.
However, they actually use a different set of memory relations ($\mathcal{R}_e$)
for source level verification with the set for compiler correctness ($\mathcal{R}_c$).
At top level, they prove that
\[
  L_{CA} \ruscref_{\mathcal{R}_e} \kwd{a.spec} \oplus \kwd{b.spec} \ruscref_{\mathcal{R}_e}
  \sem{M_C} \oplus \sem{M_A}
\]

Note that the verification here is slightly different with ours. The memory relations in
$\mathcal{R}_e$ here include specific invariant which ensures that the variables
$\kwd{memorized1}$ and $\kwd{memorized2}$ have desired values in the incoming memories.
However, the values of these \kwd{static} variables cannot be determined at the invocation
of a function.
%
In other words, they describe and verify the behavior of this program in an ideal (instead
of arbitrary) memory environment, i.e. the program always reads the correct summation from input
$i$. 
On the other hand, we prove the preservation of open semantics
for all possible memories. The return value of the program
is determined by both the input $i$ and the incoming memory
as denoted by $\mathit{sum}(i,m)$ in~\defref{def:ccm-top-spec}.
Given these differences, CompCertM use 5764 LoC to define the top-level memory
relations, specification and complete the proof. We use 3124 LoC to complete
our proof and link it with the result of VCC to achieve end-to-end direct refinement.

For further comparison, we now discuss how to compose the source verification with
compiler correctness and adequacy of assembly linking within the framework of CompCertM.
If $M_C$ is compiled to $\kwd{CompCert}(M_C)$, then the correctness of CompCertM states that
$
  \sem{M_C} \ruscref_{\mathcal{R}_c} \sem{\kwd{CompCert}(M_C)}
$
The adequacy property for linking assembly modules is stated as
$
  \kwd{Beh}(\sem{\kwd{CompCert}(M_C)} \oplus \sem{M_A}) \supseteq \kwd{Beh}(\sem{\kwd{CompCert}(M_C)} + \sem{M_A})
  $
There are two approaches to the end-to-end semantics preservation. Note that these two
approaches are speculative because CompCertM did not present such composition.

Firstly, we can compose the three parts vertically in the form of \emph{behavior refinement}.
However, this approach only makes sense when the composed modules have closed semantics.
\[
  \kwd{Beh}(L_{CA}) \supseteq \kwd{Beh}(\sem{\kwd{CompCert}(M_C)} + \sem{M_A})
\]
Since the behavior refinement is transitively composable, we need to show $\kwd{Beh}
(\sem{M_C \oplus M_A}) \supseteq \kwd{Beh}(\sem{\kwd{CompCert}(M_C)} \oplus \sem{M_A})$.
According to the definition of $\ruscref_{\mathcal{R}_c}$, it suffices to prove that 
$\sem{M_A}$ is self-related by $\mathcal{R}_c$ which consists of six different
memory relations,
In CompCertM, all \kwd{Clight} and assembly modules can be self-related by all the relations
in $\mathcal{R}_c$. Thus the result above can be easily obtained. However, we find that it is confusing
to claim that any hand-written assembly program, even those do not obey the calling convention
of CompCert, can satisfy $\mathcal{R}_c$ and safely be linked with programs compiled by CompCertM.
The ability of CompCertM to prove such a property may come from its open semantics of
assembly programs which can perform C-style calls and returns.
The interaction between open modules through external calls in assembly-style is also limited
by their ``repaired semantics''.
In this regard, our direct refinement can better describe the desired properties of
the assembly modules to be safely linked with compiled modules.

Secondly, for open semantics preservation in the form of RUSC refinement, one need
to union the memory relations as $\mathcal{R}_e \cup \mathcal{R}_c$. Since the adequacy
theorem is provided only in the form of behavior refinement, the conclusion is
\[
  L_{CA} \ruscref_{\mathcal{R}_e \cup \mathcal{R}_c}  \sem{\kwd{CompCert}(M_C)} \oplus \sem{M_A}
\]
To achieve this refinement, one need to show that the \emph{end modules} are self-related
by each $R \in \mathcal{R}_e \cup \mathcal{R}_c$. Excluding for the parts that have already
been proved, we still need to show that \emph{1)} $L_{CA}$ is self-related by
$\mathcal{R}_c$ and \emph{2)} $\sem{\kwd{CompCert}(M_C)}$ and $\sem{M_A}$ are
self-related by $\mathcal{R}_e$. These conditions are not demonstrated in CompCertM and we do not
know whether they hold or not.

In other words, one need to show that \emph{1)} the specification of source program satisfies
all memory relations used in the compilation and \emph{2)} the assembly modules satisfy
the memory relations used for the source level verification.
These limitations can increase the difficulty and reduce the extensionality of the proofs.
Our direct refinement approach overcomes these obstacles.

\subsection{CompCertO}
In CompCertO, the simulation convention of the overall simulation for the compiler is
stated as
$\scc_{\kwd{CCO}} = \mathcal{R}^* \cdot \kwt \cdot \kcl \cdot \klm \cdot \kma \cdot \kasm_{\kwd{vainj}}$
where $\mathcal{R} = \kc_\kinjp + \kc_\kinj + \kc_\kext + \kc_\kwd{vainj} + \kc_\kwd{vaext}$ is a set of
simulation conventions parameterized over used KMRs which is similar to CompCertM. $\mathcal{R}^*$ basically
means that $\mathcal{R}$ can be used for zero or arbitrary times. Which means that the source verification
can also be absorbed into $\scc_{\kwd{CCO}}$ as what we do in ~\secref{sec:application}.

As discussed in ~\secref{ssec:problems}, the main difficulty in proving the running example using
CompCertO is that the simulation convention is dependent on the details of compilation.
We take $\kser$ as an example for hand-written assembly and try to link it with
$L_\texttt{S}$ using $\scc_{\kwd{CCO}}$.

Firstly, for $\mathcal{R}^*$ we need to prove that $L_\texttt{S}$ is self-related using $\mathcal{R}$, thus
the self-simulation can be duplicated for arbitrary times. 
\[\osims{\kc_\kinjp + \kc_\kinj + \kc_\kext + \kc_\kwd{vainj} + \kc_\kwd{vaext}}{L_\texttt{S}}{L_\texttt{S}}\]
This simulation means that $L_s$ can take queries related by each of the KMRs, and choose
one of them for its external calls. $\kwd{vainj}$ and $\kwd{vaext}$ are used to capture the
consistency between static analyzer and the dynamic memories are we mentioned in ~\secref{ssec:single-pass}.
This simulation is conceptually correct but quite complex to prove.

Moreover, we need to define two extra specifications $L_{\ltlli}$ and $L_{\machli}$ for
intermediate language interfaces and prove
$\osims{\kcl}{L_{\texttt{S}}}{L_\ltlli}$,
$\osims{\klm}{L_\ltlli}{L_\machli}$ and
$\osims{\kma}{L_\machli}{\sem{\code{server.s}}}$.
This approach not only requires a significant amount of effort but also presents a technical
challenge which is how to achieve the protection of memory and registers for the target
program. In ~\figref{fig:mem-protect-exm}, we use \kinjp together with the structural simulation
convention $\kca \equiv \kcl \cdot \klm \cdot \kma$.
For example, the saved values of registers are protected as \koutofreach in
the memory (\kinjp) such that these registers can be correctly restored
to satisfy the calling convention. In $\scc_{\kwd{CCO}}$, $\klm$ requires that
callee-save registers are protected, $\kma$ requires that $\spreg$ and $\rareg$
are protected. While they do not provide any memory protection for saved values of
these registers on the stack.
This makes it extremely challenging to establish a simulation between $L_\texttt{S}$
and $\sem{\kser}$ through intermediate specifications. In fact, we came up with
the idea of direct refinement during our attempts to prove this.

\end{document}